\documentclass[10pt,journal,compsoc]{IEEEtran}
\ifCLASSOPTIONcompsoc
  \usepackage[nocompress]{cite}
\else
  \usepackage{cite}
\fi

\usepackage{amsmath,amssymb,amsfonts}
\usepackage{algorithm}
\usepackage{algpseudocode}
\usepackage{etoolbox}
\usepackage{multirow}
\usepackage[table,xcdraw]{xcolor}
\usepackage{graphicx}
\usepackage{textcomp}
\usepackage{xcolor}
\usepackage{url}
\usepackage{hyperref}
\usepackage{subcaption}
\usepackage{booktabs}
\usepackage{xr}
\usepackage{makecell}
\usepackage{printlen}

\usepackage{amsthm}
\newtheorem{theorem}{Theorem}
\newtheorem{definition}{Definition}

\newcolumntype{L}[1]{>{\raggedright\let\newline\\\arraybackslash\hspace{0pt}}m{#1}}
\newcolumntype{C}[1]{>{\centering\let\newline\\\arraybackslash\hspace{0pt}}m{#1}}
\newcolumntype{R}[1]{>{\raggedleft\let\newline\\\arraybackslash\hspace{0pt}}m{#1}}

\newcommand{\squishlist}{
   \begin{list}{$\bullet$}
    { \setlength{\itemsep}{2pt}
    \setlength{\parsep}{0pt}
      \setlength{\topsep}{5pt}     \setlength{\partopsep}{0pt}
      \setlength{\leftmargin}{1.35em} \setlength{\labelwidth}{1em}
      \setlength{\labelsep}{0.5em} } }

\newcommand{\squishlisttwo}{
   \begin{list}{$\bullet$}
    { \setlength{\itemsep}{0pt}    \setlength{\parsep}{0pt}
      \setlength{\topsep}{0pt}     \setlength{\partopsep}{0pt}
      \setlength{\leftmargin}{1.35em} \setlength{\labelwidth}{1em}
      \setlength{\labelsep}{0.5em} } }

\newcommand{\squishend}{
    \end{list}  }

\usepackage{tikz}
\usetikzlibrary{tikzmark}
\usetikzlibrary{calc}

\errorcontextlines\maxdimen

\newcommand{\ALGtikzmarkcolor}{black}
\newcommand{\ALGtikzmarkextraindent}{2.5pt}
\newcommand{\ALGtikzmarkverticaloffsetstart}{-.8ex}
\newcommand{\ALGtikzmarkverticaloffsetend}{-.8ex}
\makeatletter
\newcounter{ALG@tikzmark@tempcnta}

\newcommand\ALG@tikzmark@start{%
    \global\let\ALG@tikzmark@last\ALG@tikzmark@starttext%
    \expandafter\edef\csname ALG@tikzmark@\theALG@nested\endcsname{\theALG@tikzmark@tempcnta}%
    \tikzmark{ALG@tikzmark@start@\csname ALG@tikzmark@\theALG@nested\endcsname}%
    \addtocounter{ALG@tikzmark@tempcnta}{1}%
}

\def\ALG@tikzmark@starttext{start}
\newcommand\ALG@tikzmark@end{%
    \ifx\ALG@tikzmark@last\ALG@tikzmark@starttext
    \else
        \tikzmark{ALG@tikzmark@end@\csname ALG@tikzmark@\theALG@nested\endcsname}%
        \tikz[overlay,remember picture] \draw[\ALGtikzmarkcolor] let \p{S}=($(pic cs:ALG@tikzmark@start@\csname ALG@tikzmark@\theALG@nested\endcsname)+(\ALGtikzmarkextraindent,\ALGtikzmarkverticaloffsetstart)$), \p{E}=($(pic cs:ALG@tikzmark@end@\csname ALG@tikzmark@\theALG@nested\endcsname)+(\ALGtikzmarkextraindent,\ALGtikzmarkverticaloffsetend)$) in (\x{S},\y{S})--(\x{S},\y{E});%
    \fi
    \gdef\ALG@tikzmark@last{end}%
}

\algrenewcommand\algorithmicindent{1.2em}

\apptocmd{\ALG@beginblock}{\ALG@tikzmark@start}{}{\errmessage{failed to patch}}
\pretocmd{\ALG@endblock}{\ALG@tikzmark@end}{}{\errmessage{failed to patch}}
 \algrenewcommand\ALG@beginalgorithmic{\footnotesize}
\makeatother

\newcommand{\approach}{fPMC}

\begin{document}
\title{Software Performability Analysis\\ Using Fast Parametric Model Checking}

\author{Xinwei~Fang\IEEEauthorrefmark{1},
        Radu~Calinescu\IEEEauthorrefmark{1},
        Simos~Gerasimou,
        and~Faisal~Alhwikem
\IEEEcompsocitemizethanks{\IEEEcompsocthanksitem[]  \IEEEauthorrefmark{1} The first two authors contributed equally to the article.
\IEEEcompsocthanksitem[] X. Fang, R. Calinescu, S. Gerasimou and F. Alhwikem are with the Department of Computer
Science at the University of York, UK.}
}

\IEEEtitleabstractindextext{
\begin{abstract}
We present an efficient parametric model checking (PMC) technique for the analysis of software \emph{performability}, i.e., of the performance and dependability properties of software systems. The new PMC technique works by automatically decomposing a parametric discrete-time Markov chain (pDTMC) model of the software system under verification into fragments that can be analysed independently, yielding results that are then combined to establish the required software performability properties. Our fast parametric model checking (fPMC) technique enables the formal analysis of software systems modelled by pDTMCs that are too complex to be handled by existing PMC methods. Furthermore, for many pDTMCs that state-of-the-art parametric model checkers can analyse, \approach\ produces solutions (i.e., algebraic formulae) that are simpler and much faster to evaluate. We show experimentally that adding fPMC to the existing repertoire of PMC methods improves the efficiency of parametric model checking significantly, and extends its applicability to software systems with more complex behaviour than currently possible.
\end{abstract}

}

\maketitle

\section{Introduction}
Most software operates in environments characterised by workloads, usage profiles, failures and available resources that are stochastic in nature~\cite{Calinescu2013emerging,de2017software,hezavehi2021uncertainty}. As such, probabilistic models such as Markov chains~\cite{filieri2015supporting,gerasimou2018synthesis,paterson2018observation}, queueing networks~\cite{balsamo2003review,4620121} and stochastic Petri nets~\cite{balsamo2012methodological,lindemann1998performance,perez2010performance} have long been used to analyse the \emph{performability} (i.e., the performance, dependability and other nonfunctional properties) of software. 

In this paper, we focus on the analysis of software performability using \emph{parametric discrete-time Markov chains} (pDTMCs), i.e., Markov chains that have transition probabilities and/or that are augmented with rewards specified as rational functions over the parameters of the analysed system. The technique used to analyse these stochastic models is called \emph{parametric model checking} (PMC)~\cite{Daws:2004:SPM:2102873.2102899,hahn2011probabilistic,Jansen2014,RaduePMC}. Given a pDTMC model of a software system, and a set of nonfunctional system properties specified in probabilistic temporal logic, PMC computes algebraic formulae for these properties. The concept is straightforward. As a simple example, consider a web server that handles two types of request, and suppose that requests belong to these types with probabilities $p$ and $1-p$. If the mean times to handle the two types of request are $t_1$ and $t_2$, respectively, then the expected (i.e., mean) time for handling a request is computed by PMC as $pt_1+(1-p)t_2$. 

The algebraic formulae produced by PMC have many important applications in software engineering. They can be used to analyse the sensitivity of nonfunctional software properties to parametric variability~\cite{filieri2015supporting}, to identify optimal system configurations for software product lines~\cite{ghezzi2013model,ghezzi2011verifying}, and to establish confidence intervals for the analysed nonfunctional properties~\cite{calinescu2016formal,calinescu2016fact,alasmarietal2022}. Furthermore, PMC formulae precomputed prior to deployment (when some of the system parameters are unknown) can be evaluated at runtime (when the parameter values are determined through monitoring), to verify if the nonfunctional requirements of a system are still satisfied after environmental changes~\cite{10.1145/2884781.2884814}. Last but not least, self-adaptive systems can use these formulae to efficiently select new configurations when requirements are violated after such changes~\cite{Filieri2011,filieri2013probabilistic}. 

Despite these benefits, PMC is seldom used in practice due to its limited scalability. While theoretical advances over the past decade~\cite{hahn2011probabilistic,Jansen2014} and their implementation in state-of-the-art model checkers~\cite{param,prism,storm} have alleviated this limitation, existing PMC approaches are often unable to analyse pDTMCs with large numbers of parameters.

Our paper presents a fast parametric model checking (fPMC) technique that extends the applicability of PMC to software systems with considerably more complex behaviour and with much larger sets of parameters than currently possible. \approach\ is a compositional analysis technique that uses well-defined rules (described later in the paper) to partition the graph induced by the pDTMC under analysis\footnote{i.e., the directed graph comprising a vertex for each pDTMC state and an edge between each pair of vertices that correspond to pDTMC states between which a transition is possible} into subgraphs called \emph{fragments}. The \approach\ fragments define small pDTCMs that are analysed individually to generate PMC subexpressions. Finally, the overall PMC result is obtained by combining these subexpressions with an expression produced by analysing an \emph{abstract} model created by replacing each fragment from the original pDTMC with a single state.

\approach\ fragments are not strongly connected components (SCCs) of 
the analysed pDTMC. They can typically be assembled to ensure that each fragment is small enough for its individual PMC analysis to be feasible, and large enough to avoid the creation of so many fragments that the abstract pDTMC becomes difficult to analyse. This flexibility yields both fragments that include only a part of an SCC and fragments that include multiple SCCs, and explains why fPMC can efficiently analyse many pDTMCs not handled by the SCC-partitioning PMC approach from~\cite{Jansen2014}, as shown in our experimental evaluation from Section~\ref{sec:evaluation}. 

\approach\ builds on recent research that laid the groundwork for the use of pDTMC fragments to speed up parametric model checking~\cite{RaduePMC}. However, that research provides no algorithm for the partition of pDTMCs into fragments. The main contributions of our paper are:
\begin{enumerate}
    \item The \approach\ theoretical foundation comprising algorithms \textbf{(a)}~for pDTMC fragmentation, and \textbf{(b)}~for pDTMC \emph{restructuring}, to aid the formation of suitably sized fragments.
    \item A new parametric model checking tool that \textbf{(a)}~employs a simple heuristic to determine whether the analysis of a pDTMC requires fragmentation, and \textbf{(b)}~performs the analysis of the pDTMC by using our \approach\ fragmentation and restructuring algorithms if fragmentation is required, or by invoking the model checker Storm~\cite{storm} otherwise.
    \item An extensive evaluation of the \approach\ theoretical foundation and tool for 62 pDTMC model variants and a wide range of performability software properties taken from the 
    research literature.  
\end{enumerate}
A preliminary \approach\ version that only supports the analysis of reachability probabilistic temporal logic formulae over pDTMCs was introduced in~\cite{fang2021fast}. This paper extends the theoretical foundation from~\cite{fang2021fast} with:
\begin{enumerate}
    \item Support for the PMC of unbounded until formulae, which correspond to the analysis of software properties such as the probability of successful termination without intermediate errors or timeouts.
    \item Support for the PMC of reachability reward formulae---a significant improvement because using PMC to analyse nonfunctional properties related to the performance, cost, utility and resource usage of software systems requires the specification of these properties as reward formulae over pDTMCs.
    \item Formal correctness proofs for the \approach\ fragmentation algorithm and pDTMC restructuring methods.
    \item A formal complexity analysis of the end-to-end \approach\ fragmentation technique.
\end{enumerate}
Additionally, we considerably extended and improved the validation of \approach\ by evaluating it for a much broader range of models and properties (Section~\ref{sec:evaluation}). Finally, we augmented the \approach\ tool support with the heuristic for determining if the pDTMC under analysis requires fragmentation (Section~\ref{sec:Efficiency}). 

The remainder of the paper is structured as follows. Section~\ref{sec:preliminary} provides formal definitions and explanations of the techniques used in this work. Section~\ref{sec:example} describes a software system we use to motivate the need for \approach\ and to illustrate its application. The \approach\ algorithms and their proofs are presented in Section~\ref{theory}, followed by the implementation details in Sections~\ref{sec:implementation}. We then evaluate \approach\ in Section~\ref{sec:evaluation}, and discuss threats to validity in Section~\ref{sec:validity}. Finally, Section~\ref{sec:relatedwork} compares \approach\ to related work, and Section~\ref{sec:conclusion} provides a brief summary and discusses directions for future work.

\section{Preliminaries }
\label{sec:preliminary}
Parametric model checking~\cite{Daws:2004:SPM:2102873.2102899,Jansen2014,hahn2011probabilistic,RaduePMC} is a mathematically based technique for the analysis of pDTMC properties expressed in \emph{probabilistic computation tree logic} (PCTL)~\cite{hansson1994logic,bianco_alfaro_1995,Ciesinski04G} extended with \emph{rewards}~\cite{andova2003discrete}. This section provides formal definitions for each of these concepts. 

\subsection{Discrete-time Markov chains}
Discrete-time Markov chains (DTMCs) are finite state-transition models used to analyse the stochastic behaviour of real-world systems. They comprise states that correspond to relevant configurations of the system under analysis, and transitions that model the changes that can occur between these configurations.

\begin{definition}
A (non-parametric) discrete-time Markov chain is a tuple 
\begin{equation}
    \label{eq:dtmc}
    D=(S,s_0,\textbf{P}, L), 
\end{equation}
\mbox{where: (i)~$S$ is a finite set of states; (ii)~$s_0 \in S$ is the initial state;} 
(iii)~$\textbf{P}: S \times S \rightarrow [0,1]$ is a transition probability matrix such that, for any states $s, s' \in S$, $\textbf{P}(s,s')$ represents the probability that the Markov chain transitions from $s$ to $s'$, and, for any $s \in S$,  $\sum_{s' \in S} \textbf{P}(s,s')=1$; and (iv)~$L: S \rightarrow 2^\mathit{AP}$ is a labelling function that maps every state $s \in S$ to elements of a set of atomic propositions $\mathit{AP}$ that hold in that state.
\end{definition}

Given a discrete-time Markov chain~\eqref{eq:dtmc}, a state $s \in S$ is an \emph{absorbing state} if $\textbf{P}(s,s)=1$ and $\textbf{P}(s,s')=0$ for all $s \ne s'$, and a \emph{transient state} otherwise. A \emph{path} $\pi$ over a DTMC $D$ is a (possibly infinite) sequence of states from $S$, such that for any consecutive states $s$ and $s'$ from $\pi$, $\mathbf{P}(s,s')>0$. The $i$-th state on the path $\pi$, $i\geq 1$, is denoted $\pi(i)$. For any state $s$, $\mathit{Paths}^D(s)$ represents the set of all infinite paths over $D$ that start with state $s$. 

To enlarge the spectrum of nonfunctional properties that can be analysed using DTMCs, these models are often augmented with \emph{reward functions}.

\begin{definition}
A reward function over a DTMC~\eqref{eq:dtmc} is a function
\begin{equation}
    \label{eq:reward}
    \mathit{rwd}:S\rightarrow \mathbb{R}_{\geq 0}
\end{equation}
that associates a non-negative quantity (i.e., a \emph{reward}) with each Markov chain state.
\end{definition}

Finally, parametric DTMCs are used when a reward-augmented Markov chain contains probabilities or rewards that are unknown or that correspond to adjustable parameters of the system under analysis.

\begin{definition}
A parametric discrete-time Markov chain (pDTMC) is a Markov chain~\eqref{eq:dtmc} augmented with a set of reward functions~\eqref{eq:reward} that comprises transition probabilities and/or rewards specified as rational functions over a set of continuous variables.
\end{definition}

The continuous variables from the previous definition correspond to parameters of the modelled system and its environment.

\subsection{Probabilistic computation tree logic}

The properties of (non-parametric and parametric) dis\-crete-event Markov chains are formally specified in reward-extended PCTL.

\begin{definition}
A PCTL \emph{state formula} $\Phi$, \emph{path formula} $\Psi$, and \emph{reward state formula} $\Phi_\mathsf{R}$ over an atomic proposition set $\mathit{AP}$ are defined by the grammar:
\begin{equation}
\label{eq:pctl}
\!\!\!\!
\begin{array}{l}
    \Phi::=  true \,\vert\, a \,\vert\, \neg \Phi \,\vert\, \Phi \wedge \Phi \,\vert\, \mathcal{P}_{\!=?} [\Psi] \\[0.5mm]
    \Psi::= \mathrm{X} \Phi \;\vert\; \Phi\; \mathrm{U}\; \Phi \;\vert\; \Phi\; \mathrm{U}^{\leq k}\, \Phi\\[0.5mm]
    \Phi_\mathsf{R}^::= 
    \mathcal{R}_{\!=?}^\mathit{rwd}[\mathrm{I}^{=k}] \;\vert\; \mathcal{R}_{\!=?}^\mathit{rwd}[\mathrm{C}^{\leq k}] \;\vert\;
    \mathcal{R}_{\!=?}^\mathit{rwd}[\mathrm{F}\; \Phi] \;\vert\; 
    \mathcal{R}_{\!=?}^\mathit{rwd}[\mathrm{S}]
\end{array}
\end{equation}
where $a \in AP$ is an atomic proposition, $k \in \mathbb{N}_{>0}$ is a timestep bound, and $\mathit{rwd}$ is a reward structure~\eqref{eq:reward}.
\end{definition}

The PCTL semantics is defined using a satisfaction relation $\models$ over the states $s\in S$ and paths $\pi\in \mathit{Paths}^D(s)$ of a Markov chain~\eqref{eq:dtmc}. Thus, $s\models \Phi$ means ``$\Phi$ holds in state $s$'', $\pi\models \Psi$ means ``$\Psi$ holds for path $\pi$'', and we have: $s\models true$ for all states $s\in S$; $s \models a$ iff $a\in L(s)$; $s \models \neg \Phi$ iff $\neg (s\models \Phi)$; and $s\models \Phi_1 \wedge \Phi_2$ iff $s\models \Phi_1$ and $s\models \Phi_2$. 

The \emph{next formula} $X \Phi$ holds for a path $\pi$ if $\pi(2)\models \Phi$. The \emph{time-bounded until formula} $\Phi_1\, \mathrm{U}^{\leq k}\, \Phi_2$ holds for a path $\pi$ iff $\pi(i)\models \Phi_2$ for some $i\leq k$ and $\pi(j)\models \Phi_1$ for all $j=1,2,\ldots,i-1$; and the \emph{unbounded until formula} $\Phi_1\,\mathrm{U}\, \Phi_2$ removes the bound $k$ from the time-bounded until formula. 

The quantitative state formula $\mathcal{P}_{\!=?} [\Psi]$ specifies the probability that paths from $\mathit{Paths}^D\!(s)$ satisfy the path property $\Psi$. \emph{Reachability properties}  $\mathcal{P}_{\!=?} [\mathsf{true}\, \mathrm{U}\, \Phi]$ are equivalently written as $\mathcal{P}_{\!=?} [\mathrm{F}\, \Phi]$ or $\mathcal{P}_{\!=?} [\mathrm{F}\, R]$, where $R\!\subseteq\! S$ is the set of  states in which $\Phi$ holds.

Finally, the reward formulae specify the expected values for: the \emph{instantaneous reward} at timestep $k$ ($\mathcal{R}_{\!=?}^\mathit{rwd}[\mathrm{I}^{=k}]$); the \emph{cumulative reward} up to timestep $k$ ($\mathcal{R}_{\!=?}^\mathit{rwd}[\mathrm{C}^{\leq k}]$); the \emph{reachability reward} cumulated until reaching a state that satisfies a property $\Phi$ ($\mathcal{R}_{\!=?}^\mathit{rwd}[\mathrm{F}\; \Phi]$, or $\mathcal{R}_{\!=?}^\mathit{rwd}[\mathrm{F}\; R]$  if $R\!\subseteq\! S$ is the set of  states in which $\Phi$ holds); and the \emph{steady-state reward} in the long run ($\mathcal{R}_{\!=?}^\mathit{rwd}[\mathrm{S}]$). A complete description of the PCTL semantics is available in~\cite{bianco_alfaro_1995,hansson1994logic,andova2003discrete}.

\subsection{Parametric Model Checking through Fragmentation}
\label{subsec:fragmentationTheory}

\begin{figure}
	\centering
    \includegraphics[width=0.95\linewidth]{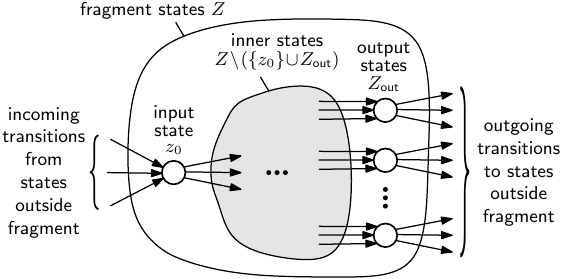}	
    \caption{pDTMC fragment $F=(Z,z_0,Z_\mathsf{out})$}
    \label{fig:fragment-diagram}
\end{figure}

Our \approach\ technique builds on recently introduced theoretical results on the use of pDTMC \emph{fragmentation} to speed up parametric model checking~\cite{RaduePMC}. These results, which explain how pDTCM fragments can be exploited---but not how they could be obtained---are summarised below. We start by introducing the concept of a pDTMC fragment.

\begin{definition}
\label{def:fragment}
A fragment of a pDTMC $D=(S,s_0,\textbf{P}, L)$ is a tuple 
\begin{equation}
    \label{eq:fragment}
  F=(Z,z_0,Z_\mathsf{out}),
\end{equation}
where (Figure~\ref{fig:fragment-diagram}):
\begin{itemize}
\item $Z\subset S$ is a subset of transient pDTMC states; 
\item $z_0$ is the (only) \emph{input state} of $F$, i.e., $\{z_0\}\!=\!\{z\!\in\! Z\mid \exists s\!\in\! S\!\setminus\! Z\:.\: \mathbf{P}(s,z)>0\}$; 
\item $Z_\mathsf{out} =\{z\in Z 
\mid (\exists s\in S\setminus Z\:.\: \mathbf{P}(z,s)>0) \wedge (\forall z'\in Z\setminus\{z_0\}\:.\: \mathbf{P}(z,z')=0)\}$ is the non-empty set of \emph{output states} of $F$.
\end{itemize}
\end{definition} 

This definition is much less restrictive than that of a strongly connected component. In particular, any pDTMC state $z$ forms a one-state, \emph{degenerate fragment} $F=(\{z\},$ $z,\{z\})$. Furthermore, the ``inner'' states $Z\setminus (\{z_0\}\cup Z_\mathsf{out})$ of a fragment can include one of several SCCs. Finally, an SCC can be split into multiple fragments, because paths that start from the output states of a fragment and reach its input state (either directly or through intermediate states outside the fragment) are permitted.  

Given a fragment $F=(Z,z_0,Z_\mathsf{out})$ of a pDTMC $D$ augmented with a reward function $\mathit{rwd}$, the PMC of reachability, unbounded until, and reachability reward properties of $D$ can be carried out compositionally by using the following four-step process introduced in~\cite{RaduePMC} and illustrated in Figure~\ref{fig:approach}:
\squishlist
    \item[1)] Use standard PMC to obtain algebraic formulae for:
    \squishlist
        \item[i)] the probabilities $\mathit{prob}_z$ of reaching each of the output states $z\in Z_\mathsf{out}$ of $F$ from the input fragment state $z_0$;
        \item[ii)] the cumulative reward $\mathit{rwd}_\mathsf{out}$ associated with reaching the output state set $Z_\mathsf{out}$ from $z_0$.
    \squishend
    \item[2)] Assemble an \emph{abstract pDTMC model} $D'=(S',s'_0,P')$ augmented with a reward function $\mathit{rwd}'$, where: 
    \squishlist
        \item[i)] $S'=(S\setminus Z)\cup \{z'\}$, i.e., the states from $Z$ are replaced with a single, abstract state $z'$;
        \item[ii)] $s'_0=s_0$ if $z_0\neq s_0$, and $s'_0=z'$ otherwise;
        \item[iii)] the incoming transitions and transition probabilities of $z'$ are inherited from $z_0$; 
        \item[iv)] $z'$ has outgoing transitions to each state that one or more states from $Z_\mathsf{out}$ have outgoing transitions to in $D$; and the probabilities of these transitions can be expressed in terms of the reachability properties computed in step~1---details about the calculation of these probabilities are provided in~\cite{RaduePMC};
        \item[v)] the new reward function is given by $\mathit{rwd}'(s)=\mathit{rwd}(s)$ for all $s\in S'\setminus\{z'\}$, and  $\mathit{rwd}'(z')=\mathit{rwd}_\mathsf{out}$.
    \squishend
    \item[3)] Compute the PMC formula for the original property under analysis, for the abstract model from step~2.
    \item[4)] Combine the PMC formulae from step~1 and the PMC formula from step~3 into a system of expressions.
\squishend

\begin{figure}
	\centering
    \includegraphics[width=\hsize]{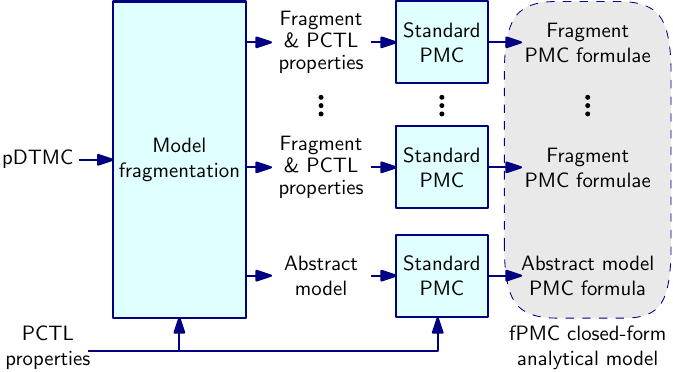}	
    \caption{Parametric model checking through fragmentation}
    \label{fig:approach}
\end{figure}

The system of expressions from step~4 provides a closed-form analytical model for the analysed property. This analytical model is equivalent to the PMC formula obtained by analysing the original pDTMC in one step, and we will refer to it as \approach\ or \approach-computed algebraic formulae in the rest of the paper. 

The PMC carried out in steps~1 and~3 uses models that are simpler and smaller than the original model $D$. As such, this four-step approach is often faster, produces much smaller algebraic formulae, and enables the analysis of models that are larger and more complex than those supported by previous PMC methods.

\section{Motivating Example}
\label{sec:example}
In this section, we introduce a software system that will be used to illustrate the use of our \approach\ approach throughout the paper. Taken from~\cite{Gerasimou2015:ASE,gerasimou2018synthesis}, this a six-operation service-based system performing trading in the foreign exchange (FX) market. The workflow of the FX system is shown in Fig.~\ref{fig:fxworkflow} and described briefly below.

\medskip\noindent
\textbf{FX workflow.}
The FX system has two execution modes that a trader can choose from: an \textit{expert} mode and a \emph{normal} mode. 

In the \textit{expert} mode, the \emph{Market watch} operation extracts real-time exchange rates (i.e., bid/ask prices) of selected currency pairs, and this information is then passed to the \emph{Technical analysis} operation for further analysis, such as evaluating the current trading conditions, predicting future price movement, and deciding which actions to take. Three actions can be taken: carrying out a trade by calling the \emph{Order} operation; performing the \emph{Market watch} operation again, e.g., on different or additional currency pairs; and reporting an error that triggers an \emph{Alarm} operation. The \emph{Order} and \emph{Alarm} operations are each followed by a user \emph{Notification} operation, and the end of the workflow. 

In the \emph{normal} mode, the system uses a \emph{Fundamental analysis} operation to evaluate the economic outlook of a country and decides whether to: call the \emph{Order} operation to trade the currency of that country; re-do the \emph{Fundamental analysis} operation; or end the execution of the workflow. 

\begin{figure}
	\centering
    \includegraphics[width=\linewidth]{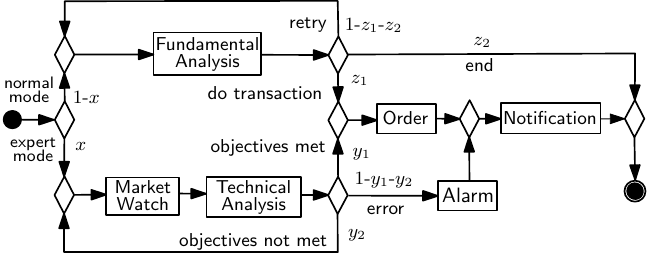}	
    \caption{FX service-based system, where $x$, $y1$, $y2$, $z1$, $z2$ are the (unknown) probabilities of different execution paths, i.e., the \emph{operational profile} of the system}
    \label{fig:fxworkflow}
    \vspace*{-2mm}
\end{figure}

Given its business-critical nature, the underlying software architecture of the FX system needs to be highly reliable. To avoid single points of failure, each FX operation is implemented by two functionally-equivalent services, and each service is invoked in order using a \emph{sequential execution strategy with retry} (SEQ\_R). For the $i$-th operation, if the first service fails, it is re-invoked with probability $r_{i1}$ or, with probability $1-r_{i1}$, the second service is invoked. If the second service also fails, it is re-invoked with probability $r_{i2}$, or the operation is abandoned with probability $1-r_{i2}$, leading to the failure of the entire workflow execution.  

\begin{figure}
	\centering
	\includegraphics[width=1.11\linewidth]{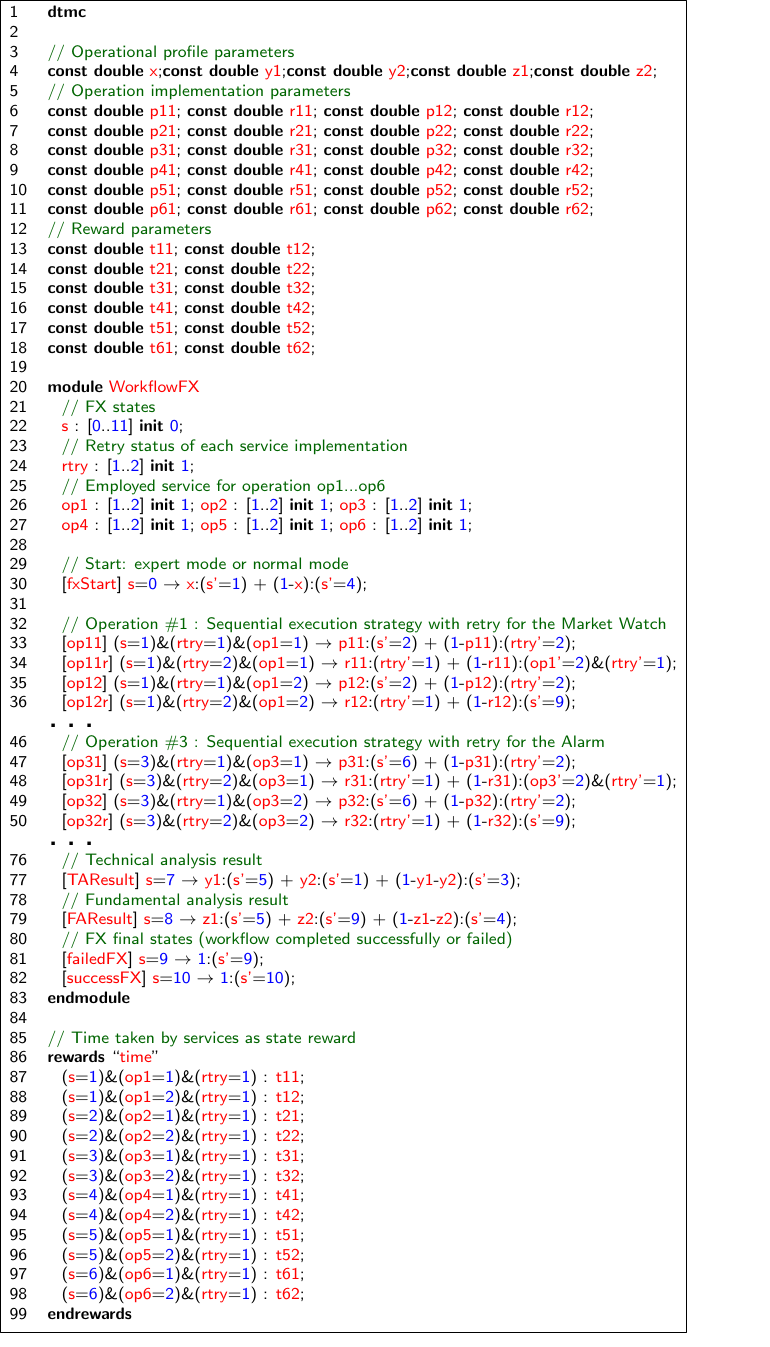}
 		\caption{pDTMC model of the FX system}
    \label{fig:fxmodel}
\end{figure}

\medskip\noindent
\textbf{FX pDTMC.}
Fig.~\ref{fig:fxmodel} shows the pDTMC model of the FX workflow, specified in the modelling language of PRISM probabilistic model checker~\cite{prism}. Lines 3--18 define the model parameters associated with \textbf{(a)}~the FX operational profile (line 4), with \textbf{(b)}~the implementations of the FX operations (lines 6--11), and with \textbf{(c)}~the mean execution times of each service used by these implementations (lines 13--18). The parameters $p_{ij}$, $r_{ij}$ and $t_{ij}$ represent the probability of successful execution, the probability of retrying and the mean execution time, respectively, for the $i$-th operation using the $j$-th service implementation, where $i\in\{1,2,\ldots,6\}$ and $j\in\{1,2\}$. The use of parameters to model the system aspects from lines 3--18 is justified because: 
\begin{enumerate}
    \item the operational profile of a system is often unknown when its model is developed;
    \item modelling the system configuration by means of a set of parameters allows the exploration and better selection of suitable system configurations;
    \item the execution times of individual services are not available until the system is deployed and executed (and may change over time).
\end{enumerate}   

\begin{table*}
\sffamily
    \centering
     \caption{Non-functional properties for the FX system from the motivating example}
     \label{Propsummary}
    \begin{tabular}{p{2mm} p{2.5cm} p{9.7cm} p{3.8cm}}
    \toprule
    \textbf{ID} & \textbf{Property type} & \textbf{Informal description} & \textbf{Property specified in PCTL} \\ \midrule
    P1 & Reachability&  probability that FX workflow completes successfully  & $\mathcal{P}_{=?}[\mathrm{F}\; \mathsf{successFX}]$ \\
    \midrule
    P2 & Reachability reward &  expected workflow execution time  & $\mathcal{R}^\mathit{time}_{=?}[\mathrm{F}\; \mathsf{failedFX} \vee \mathsf{successFX}]$ \\ 
    \midrule
    P3 & Unbounded until & probability that FX workflow completes successfully without triggering an alarm & $\mathcal{P}_{=?}[\;!\, \mathsf{Alarm} \; \mathrm{U} \; \mathsf{successFX}]$ \\ 
    \bottomrule
    \end{tabular}
\end{table*}

Inside the \emph{WorkflowFX} module, we use the local variable \emph{state} (line 22) to model the operations in the FX system, and use \emph{retry} (line 24) and \emph{$op_i$} (line 26--27) to encode the retry status of a service implementation and the employed service implementation for each operation, respectively. The selection of the \textit{expert} or \textit{normal} mode is decided in line~30, and the execution of the FX operations is modelled in lines~33--75. Due to space constraints, we only show the modelling of the \emph{Market Watch} operation (line~33--36) and the \emph{Alarm} operation (lines~46--50); all the other FX operations are modelled similarly. For both operations, the invocation of the first service succeeds with probability $p_{i1}$ and FX moves to the next operation, fails with probability $1-p_{i1}$ and is retried with probability $r_{i1}$ (lines~33 and~34, and lines~47 and 48, respectively); otherwise, the second service is executed and succeeds, or is re-invoked with probabilities $p_{i2}$ and $r_{i2}$ in lines~35 and~36, and lines~49 and~50, respectively. If both service implementations fail, the FX workflow execution terminates with a failure in line~81.

\begin{figure}
	\centering
	\includegraphics[width=\linewidth]{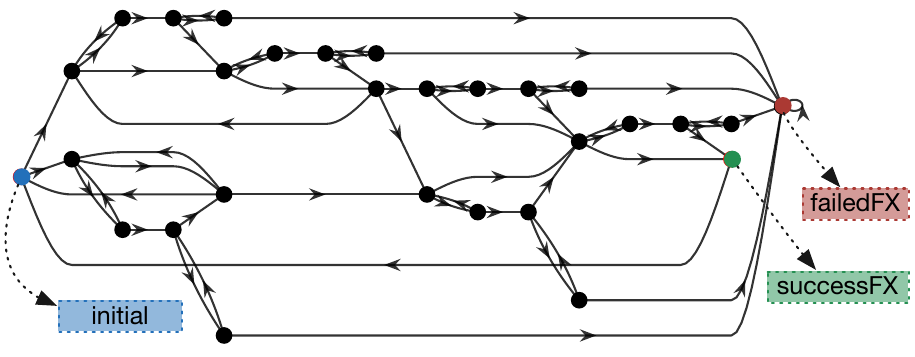}
		\caption{Graph representation of the FX pDTMC; the model comprises 29~states and 58~transitions}
    \label{fig:fxExample}
\end{figure}

Fig.~\ref{fig:fxExample} shows the directed graph induced by the FX pDTMC, with the initial and final states of the FX workflow highlighted in colour. For this pDTMC model, we assume that we are interested in analysing the three non-functional properties from Table~\ref{Propsummary}. We note that despite the relatively small number of states and transitions from the FX pDTMC model, the leading parametric model checkers PRISM~\cite{prism} and Storm~\cite{storm} could not return a closed-form PMC formula for any of these properties within an hour when run on the MacBook Pro computer we used in all our experiments (please see Section~\ref{Sec:ExperimentalSetup} for a detailed specification of this computer). In the next section, we explain how \approach\ can successfully analyse those three properties via automated model fragmentation.

\section{\approach\ Theoretical Foundation}
\label{theory}

We present the model fragmentation algorithm that underpins the \approach\ analysis of reachability, unbounded until and reachability reward pDTMC properties in Section~\ref{sse:MCFA}. A pDTMC model restructuring algorithm that aids the formation of fragments is discussed in Section~\ref{sse:MRAFF}. Section~\ref{sse:Proofs} provides formal proofs for the correctness the \approach\ algorithms. Finally, Section~\ref{ssec:fpmcapplication} illustrates the application of \approach\ to the pDTMC model and PCTL properties from our motivating example.

\subsection{Markov Chain Fragmentation Algorithm}
\label{sse:MCFA}
\approach\ partitions a pDTMC into fragments that can be analysed individually by current parametric model checkers. This partition is carried out by the function \textsc{Fragmentation} from Algorithm~\ref{algorithm:fragmentation}, supported by the auxiliary functions from Algorithms~\ref{algorithm:fragmentgrow}, \ref{algorithm:stopgrow} and~\ref{algorithm:Restructuring}. The function \textsc{Fragmentation} takes four arguments: 
\begin{enumerate}
    \item the analysed pDTMC $D(S,s_0,\textbf{P}, L)$;
    \item the analysed PCTL formula $\phi$, which can be a reachability property $\mathcal{P}_{=?} [\mathrm{F}\, \Phi]$, an unbounded until property $\mathcal{P}_{=?} [\Phi_1\, \mathrm{U}\, \Phi_2]$, or a reachability reward property $\mathcal{R}_{=?}^\mathit{rwd}[\mathrm{F}\, \Phi]$, where the inner state formulae $\Phi$, $\Phi_1$ and $\Phi_2$ cannot contain the probabilistic operator $\mathcal{P}$;\footnote{Like other parametric model checking methods~\cite{Daws:2004:SPM:2102873.2102899,hahn2011probabilistic,Jansen2014}, \approach\ only supports non-nested probabilistic properties.} 
    \item a reward function $\mathit{rwd}$ over $D$, which is only relevant if $\phi$ is a reachability reward property (we assume that $\mathit{rwd}(s)=0$ for all $s\in S$ otherwise);
    \item a \emph{fragmentation threshold} $\alpha \in \mathbb{N}_{>0}$, whose role is described later in this section.
\end{enumerate}

Given these arguments, the function returns (line~\ref{l:returnFS}):
\begin{enumerate}
    \item a restructured version of the original pDTMC, where the restructuring (described later in this section) aids the formation of fragments;
    \item a revised version of the reward function $\mathit{rwd}$ that matches the restructured pDTMC; 
    \item a set of fragments $\mathit{FS}$ that satisfy Definition~\ref{def:fragment}, with each state of the restructured pDTMC either included in a \emph{regular, multi-state fragment} or organised into a \textit{one-state (degenerate) fragment}. 
\end{enumerate}

\begin{algorithm}[t]
	\caption{pDTMC model fragmentation}\label{algorithm:fragmentation}
 		\renewcommand{\baselinestretch}{1}
		\begin{algorithmic}[1]
			\Function{Fragmentation}{$D(S,s_0,\textbf{P}, L), \phi, \mathit{rwd}, \alpha$}
			\State $V \!\gets\! \left\{\!\!\!\! \begin{array}{ll}
			\{s\!\in\! S \mid s\!\models\! \Phi_1 \!\vee\! s\!\models\! \Phi_2\}, & \!\!\!\!\textrm{if } \phi\!=\!\mathcal{P}_{=?} [\Phi_1\, \mathrm{U}\, \Phi_2]\\
			\{s\!\in\! S \mid s\!\models\! \Phi \}, & \!\!\!\!\textrm{if } \phi\!=\!\mathcal{P}_{=?} [\mathrm{F}\, \Phi] \!\vee\! \phi\!=\!\mathcal{R}_{=?}^\mathit{rwd}[\mathrm{F}\, \Phi]\\
			\end{array}
			 \right.$ \label{l:initVisited}
            \State $\mathit{FS} \gets \left\{\left(\{s\}, s,\{s\}\right) \;|\; s\in V\right\}$\label{l:satisfy}
			\ForAll{$z_0 \in S\setminus V$ } \label{l:forS}
            \State $Z \gets \{ z_0 \}$ \label{l:startFragment} 
            \State $Z_\mathsf{OUT} \gets \{\}$ \label{l:startFragmentCont} 
			\State $T \gets \textsc{EmptyStack()}$\label{l:Tinital} 
	        \State $\textsc{Traverse}(D(S,s_0,\textbf{P}, L),\mathit{rwd},\mathit{FS},V,T,Z,z_0,\mathsf{true})$ \label{l:grow1}
    			\While{$\neg \textsc{Empty}(T)$}\label{l:whileS}
    			    \State $z \gets T.\textsc{Pop}()$ 
    		        \If {$\{i\!\in\! S \:|\: \mathbf{P}(i,z)\!\neq\! 0\}\!\!\subseteq\!\! Z \wedge \{o\!\in\! S \:|\: \mathbf{P}(z,\!o)\!\neq\! 0\} \!\!\subseteq\!\! S\!\setminus\! Z$}  \label{l:ifOutput1} 
    		           \State $Z_\mathsf{OUT} \gets Z_\mathsf{OUT} \cup \{z\}$ \label{l:ifOutput2}
        		        \Else \label{checkElse}
        			    \If{$\#Z < \alpha$} \label{l:check-alphaS}
    			           \State $\!\!\!\!\textsc{Traverse}(D(S\!,\!s_0\!,\!\textbf{P}\!,\! L),\mathit{rwd},\mathit{FS},V,T,Z,z,\mathsf{false})$\label{l:grow2}
    			        \Else\label{l:Formation1}
    			           \State $\!\!\!\!\textsc{Terminate}(D(S\!,\!s_0\!,\!\textbf{P}\!,\!L),\! \phi,\mathit{rwd}, FS, V, T, Z, Z_\mathsf{OUT}, z)$\label{Alg1Terminate}
    			            \EndIf
    		            \EndIf \label{l:check-alphaE}\label{l:FormationE}
                    \State $Z \gets Z \cup \{ z \}$ \label{l:extendFragment}
    			\EndWhile\label{l:whileE} 
            \If{$\neg\textsc{ValidFragment}((Z,z_0,Z_\mathsf{OUT}))$}\label{l:validfragmentS}
                \State $Z \gets \{z_0\}$, $Z_\mathsf{OUT} \gets \{z_0\}$\label{l:degenerate2}
            \EndIf 
            \State $\mathit{FS} \gets \mathit{FS} \cup \{(Z,z_0,Z_\mathsf{OUT})\}$\label{l:addFragment}
            \State $V \gets V \cup Z$ \label{l:addToR1}
			\EndFor \label{l:forE}
			\State \Return {$D(S,s_0,\textbf{P}, L)$, $\mathit{rwd}$, $\mathit{FS}$} \label{l:returnFS} 
		    \EndFunction
		\end{algorithmic}
\end{algorithm}

The function starts by placing the pDTMC states that satisfy $\Phi_1$ and $\Phi_2$ (if the analysed property is an unbounded until formula $\mathcal{P}_{=?} [\Phi_1\, \mathrm{U}\, \Phi_2]$) or $\Phi$ (if the analysed property is a reachability formula $\mathcal{P}_{=?} [\mathrm{F}\, \Phi]$ or a reachability reward formula $\mathcal{R}_{=?}^\mathit{rwd}[\mathrm{F}\, \Phi]$) into a set of ``visited'' states $V$ (line~\ref{l:initVisited}). Each state from $V$ is then used to assemble a one-state fragment that is placed into the fragment set $\mathit{FS}$ (line~\ref{l:satisfy}). The states from $V$ are preserved as one-state fragments so that they can appear in the \approach\ abstract model (see the description from Section~\ref{subsec:fragmentationTheory}). This allows the direct analysis of the PCTL property $\phi$, which refers to these states from the abstract model. 

Next, additional fragments $(Z,z_0,Z_\mathsf{OUT})$ are generated in each iteration of the for loop from lines~\ref{l:forS}--\ref{l:forE} as follows. First, a node $z_0$ not yet included in any fragment is selected (line~\ref{l:forS}) and inserted into the fragment state set $Z$ (line~\ref{l:startFragment}), while the fragment output set $Z_\mathsf{out}$ is initialised to the empty set (line~\ref{l:startFragmentCont}). An empty stack, $T$, is created in line~\ref{l:Tinital}, and then populated with the states reached by the outgoing transitions from $z_0$ through invoking (in line~\ref{l:grow1}) the function \textsc{Traverse} from Algorithm~\ref{algorithm:fragmentgrow}. Each state $z$ from this stack, $T$, is processed by the while loop from lines~\ref{l:whileS}--\ref{l:whileE}, ending up in $Z_\mathsf{OUT}$ if it satisfies the constraints associated with output fragment states (lines~\ref{l:ifOutput1} and~\ref{l:ifOutput2}). When $z$ does not satisfy these constraints, two options are possible (lines~\ref{checkElse}--\ref{l:check-alphaE}):
    \squishlist
        \item If $Z$ has accumulated fewer states than the threshold $\alpha$, the graph traversal function \textsc{Traverse} is invoked again to add to the stack the predecessor and successor vertices of $z$ that are not already in the fragment (line~\ref{l:grow2}). 
        \item Otherwise, the function for the early fragment formation \textsc{Terminate} from Algorithm~\ref{algorithm:stopgrow} is invoked to ``force'' $z$ into becoming an output state (by restructuring the pDTMC) whenever that is possible (line~\ref{Alg1Terminate}).
    \squishend
In this way, the threshold $\alpha$ provides a soft upper bound for the fragment size. When this bound is reached, the model restructuring techniques detailed in Section~\ref{sse:MRAFF} are used to force the formation of a valid fragment if possible. Irrespective of the way in which $z$ is processed in lines~\ref{l:ifOutput1}--\ref{l:FormationE}, it becomes part of the fragment being constructed, and therefore it is added to the fragment state set $Z$ in line~\ref{l:extendFragment}.

The fragment candidate ${(Z,z_0,Z_{OUT})}$ assembled by the while loop from lines~\ref{l:whileS}--\ref{l:whileE} is validated in line~\ref{l:validfragmentS}. If the candidate does not satisfy the constraints from Definition~\ref{def:fragment}, the fragment is ``downgraded'' by using its input state, $z_0$, to form a degenerate, one-state fragment (line~\ref{l:degenerate2}). After validation (and, if necessary, degradation or restructuring), the new fragment is added to the fragment set $\mathit{FS}$ (line~\ref{l:addFragment}), and its states are added to the set of ``visited'' states $V$ already assigned to fragments (line~\ref{l:addToR1}), ensuring that they are not re-used by the loop starting in line~\ref{l:forS}.

\begin{algorithm}[t]
\caption{Traversal of pDTMC-induced graph \label{algorithm:fragmentgrow}}
		\begin{algorithmic}[1]
			\Function{Traverse}{$D(S,s_0,\textbf{P}, L),\mathit{rwd}, \mathit{FS}, V, T, Z, z, \mathit{inputState}$}
			\If{$\neg \mathit{inputState}$} \Comment{$z$ is not the fragment's input state} \label{l:ifInputS}
    			\State $I \gets \{i\in S\!\setminus\! Z \;|\; \mathbf{P}(i,z) \neq 0\}$ \label{l:inputs}
    			\If{$I \cap V = \{\}$} \label{l:wAddedIf} 
    			    \State $T.\textsc{Push}(I)$\label{l:wAdded}
    		    \Else \label{l:wFragment2S}
    		        \State $\mathit{FS} \gets \mathit{FS} \cup \{(\{z\},z,\{z\})\}$ \label{l:degenerate3}
    		        \State $V\gets V \cup \{z\}$ \label{l:addToR2}
    		        \State \Return \label{l:wFragment2E}
    			\EndIf
    		\EndIf  \label{l:ifInputE}

    			\State $O \gets \{o\in S\!\setminus\! Z  \;|\; \mathbf{P}(z,o) \neq 0\}$ \label{l:OutputS}
    			
    		\If{$O \not\subseteq V$} \label{l:notSubsetOfR}
		    \State $T.\textsc{Push}(O\!\setminus\! V)$ \label{l:pushO}
		   \If{$V\!\setminus\! O\neq \{\}$} \label{l:outgoingToOtherFragment}
                    \State $O' \gets \textsc{RestructureState}(D(S,s_0,\textbf{P}, L),\mathit{rwd},Z,z)$\label{l:restructure2}\label{l:formationAdd3}
                            \State $T.\textsc{Push}(O')$ \label{l:formationAdd3Push}
        			    \EndIf
        	   \EndIf \label{l:OutputE}
 	    \EndFunction
		\end{algorithmic}
\end{algorithm}

The growing of a fragment in Algorithm~\ref{algorithm:fragmentation} is carried out by the function $\textsc{Traverse}$ from Algorithm~\ref{algorithm:fragmentgrow}. Given a state $z$, $\textsc{Traverse}$ examines:
\squishlist
    \item its incoming transitions in lines~\ref{l:ifInputS}--\ref{l:ifInputE} (only if $z$ is not the input state $z_0$ of the fragment under construction, i.e., if $\mathit{inputState}$ is false in line~\ref{l:ifInputS});
    \item its outgoing transitions in lines~\ref{l:OutputS}--\ref{l:OutputE}. 
\squishend
The states that are directly connected to $z$ via the examined incoming and outgoing transitions, and that are not already in the set of fragment states $Z$, are collected into an input state set $I$ (line~\ref{l:inputs}) and an output state set $O$ (line~\ref{l:OutputS}), respectively.  These two sets of states are processed as follows.

Firstly, the states from $I$ are added to the stack $T$ (line~\ref{l:wAdded}) if none of them belongs to an existing fragment (line~\ref{l:wAddedIf}). Otherwise, $z$ is organised into a one-state fragment and the traversal of the pDTMC-induced graph is terminated (lines~\ref{l:degenerate3}--\ref{l:wFragment2E}) because, with an incoming transition from a state belonging to an already assembled fragment, $z$ cannot be an inner or output state of the fragment under construction. Note that creating this one-state fragment halfway through assembling another fragment may impact the construction of the other fragment. If this is the case, then the issue will be detected and dealt with by the validation process from Algorithm~\ref{algorithm:fragmentation} (line~\ref{l:validfragmentS}). 

Secondly, if the output set $O$ has at least one state not belonging to other fragments (line~\ref{l:notSubsetOfR}), growing the fragment under construction with the ``successors'' of $z$ may be feasible. As such, the function: 
\squishlist
\item places the states from $O$ that do not belong to other fragments onto the stack $T$ (line~\ref{l:pushO}); 
\item if $O$ contains states belonging to previously constructed fragments (line~\ref{l:outgoingToOtherFragment}), it attempts to continue to grow the fragment under construction by using the function \textsc{RestructureState} from Algorithm~\ref{algorithm:Restructuring} to extend the pDTMC with auxiliary states $O'$ that allow $z$ to become an inner fragment state (lines~\ref{l:restructure2} and~\ref{l:formationAdd3Push}).  
\squishend

\subsection{Early Termination of Fragment Construction and Model Restructuring to Aid Fragment Formation}
\label{sse:MRAFF}

As we will show in Section~\ref{sse:Proofs}, the function \textsc{Fragmentation} from Algorithm~\ref{algorithm:fragmentation} is guaranteed to partition a pDTMC into a set of valid fragments. However, the success of \approach\ also depends on these fragments being of an appropriate size. If a fragment is too large, existing PMC techniques (which \approach\ uses for the fragment analysis, see Figure~\ref{fig:approach}) will be unable to handle it. Conversely, partitioning a pDMTC into a very large number of small fragments may yield an abstract model whose analysis is unfeasible. 

\begin{algorithm}[t]
\caption{Early termination of fragment formation\label{algorithm:stopgrow}}
		\renewcommand{\baselinestretch}{1}
		\begin{algorithmic}[1]
			\Function{\textsc{Terminate}}{$D(S,s_0,\textbf{P}, L),\phi,\mathit{rwd},\mathit{FS},V,T,Z,Z_\mathsf{OUT},z$}
          		 \If{$(\exists\, i\!\in\!S\!\setminus\!Z . \textbf{P}(i,z)\!\neq\! 0) \wedge (\forall\, o\!\in\! Z . \textbf{P}(z,o)\!=\! 0)$ \label{ReTS} \\ \hspace*{2.6cm}$\wedge (\phi\neq\mathcal{R}_{=?}^\mathit{rwd}[\mathrm{F}\, \Phi] \vee \mathit{rwd}(z)\!=\!0)$}
        		 \State $\textsc{RestructureTrans}(D(S,s_0,\textbf{P}, L),Z,z)$\label{CallRT}
        		\State $Z_\mathsf{OUT} \gets Z_\mathsf{OUT} \cup \{z\}$\label{ReTE}
          		\ElsIf{$(\forall \: i\!\in\! S\!\setminus\! Z . \textbf{P}(i,z)\!=\! 0) \wedge (\exists\: o\!\in\! Z . \textbf{P}(z,o)\!\neq\! 0)$}\label{ReSS}
        		\State $O \gets \textsc{RestructureState}(D(S,s_0,\textbf{P}, L),\mathit{rwd},Z,z)$\label{l:formationAdd1}
        		\State  $Z \gets Z \cup O$, $Z_\mathsf{OUT} \gets Z_\mathsf{OUT} \cup O$\label{ReSE}
          	    \Else\label{Exp2}
          	    \State $\textsc{Traverse}(D(S,s_0,\textbf{P}, L),\mathit{rwd},\mathit{FS},V,T,Z,z,\mathsf{false})$\label{Exp2Solution}
         		\EndIf\label{l:restructuringA1E}
		    \EndFunction
		\end{algorithmic}
\end{algorithm}

Based on our experience (see~Section~\ref{sec:evaluation}), pDTMCs that model complex systems often comprise many loops (e.g., see the pDTMC from Figure~\ref{fig:fxExample}), whichs favour the formation of fragments that may be too large for existing PMC techniques to analyse. To address this issue, \textsc{Fragmentation} uses the threshold $\alpha$ to decide when to force the formation of a fragment, preventing it from growing too large (line~\ref{l:check-alphaS} from Algorithm~\ref{algorithm:fragmentation}). This early termination of the fragment formation is accomplished by the function \textsc{Terminate}. This function is supplied (in line~\ref{Alg1Terminate}) with complete information about the fragmentation process so far and, importantly, with a state $z$ that does not satisfy the condition from line~\ref{l:ifOutput1} and therefore cannot be an output state for the fragment under construction. The role of \textsc{Terminate} is to modify the pDTMC states and/or transitions such that: (i)~$z$ meets the condition for being an output fragment state in the restructured pDTMC; (ii)~the modifications do not affect the PMC result. This restructuring is possible in one of the following two scenarios, which are handled in lines~\ref{ReTS}--\ref{ReTE} and~\ref{ReSS}--\ref{ReSE} of \textsc{Terminate}, respectively. If neither scenario applies, \textsc{Terminate} cannot support the early termination of the fragment construction, and therefore needs to invoke the function \textsc{Traverse}, which will continue to grow the fragment (line~\ref{Exp2Solution}). As such, the threshold $\alpha$ only provides a soft upper bound for the size of an \approach\ fragment.

\begin{figure*}
     \centering
     \begin{subfigure}[b]{0.85\textwidth}
         \centering
         \includegraphics[width=\textwidth]{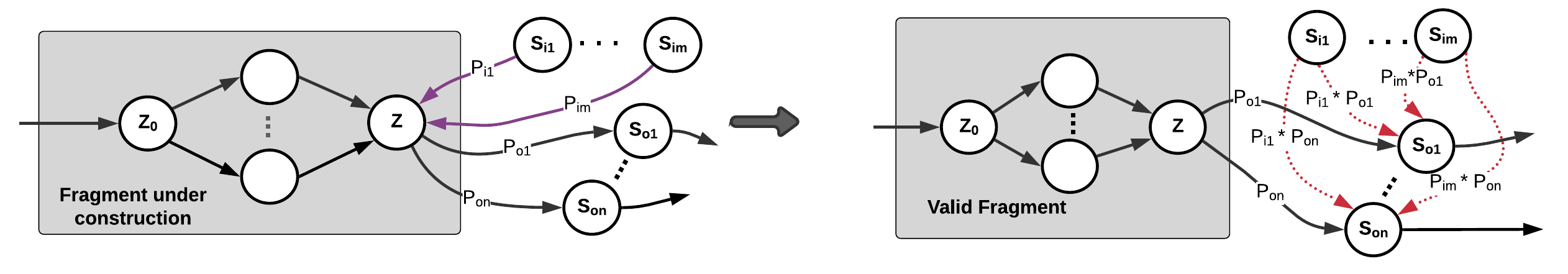}
         \caption{Transition replacement to force creation of output fragment state}
         \label{fig:trick2}
     \end{subfigure}
     \hfill
     \begin{subfigure}[b]{0.85\textwidth}
         \centering
         \includegraphics[width=\textwidth]{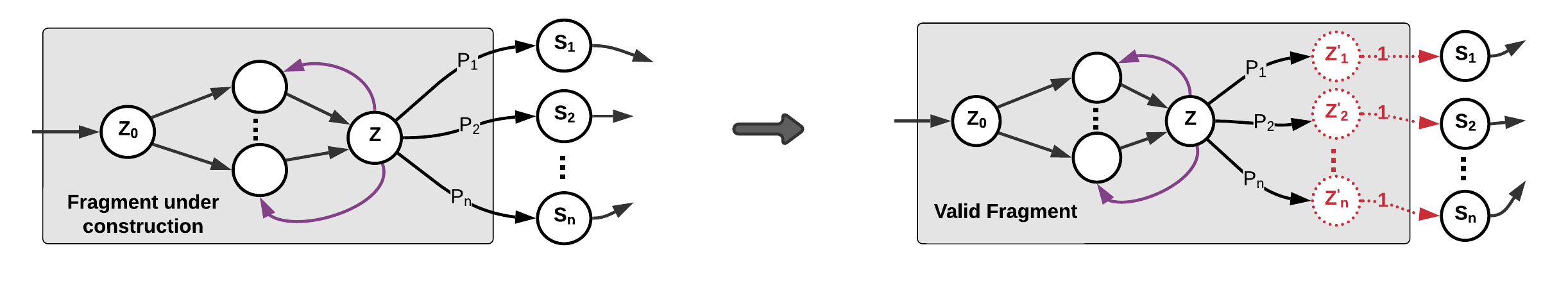}
         \caption{Auxiliary state insertion to force formation of fragment containing the new states among its output states}
         \label{fig:trick1}
     \end{subfigure}
        \caption{Model restructuring techniques supporting fragment formation}
        \label{fig:ModelReconstruction}
\end{figure*}

\medskip
\noindent
\textbf{Scenario 1.} The state $z$ has $m\!\geq\! 1$ incoming transitions (of probabilities $p_{i1}$, $p_{i2}$, \ldots, $p_{im}$) from states $s_{i1}$, $s_{i2}$, \ldots, $s_{im}$ outside the fragment, and $n\!\geq\! 1$ outgoing transitions (of probabilities $p_{o1}$, $p_{o2}$, \ldots, $p_{on}$) to states $s_{o1}$, $s_{o2}$, \ldots, $s_{on}$ outside the fragment (line~\ref{ReTS} from Algorithm~\ref{algorithm:stopgrow} and \mbox{Figure~\ref{fig:trick2}--left}); additionally (for reasons explained in Theorem~\ref{th:model-equivalence} in Section~\ref{sse:Proofs}) either the analysed property $\phi$ is not a reachability reward property, or $z$ is a zero-reward state. In this scenario, we replace the transition between each state $s_{ij}$, $1\leq j\leq m$, and $z$ with transitions of probabilities $p_{ij}p_{o1}$, $p_{ij}p_{o2}$, \ldots, $p_{ij}p_{on}$, between $s_{ij}$ and the states $s_{o1}$, $s_{o2}$, \ldots, $s_{on}$, respectively. This modification of the pDTMC structure is shown in Figure~\ref{fig:trick2}--right, and is carried out by function \textsc{RestructureTrans} from Algorithm~\ref{algorithm:Restructuring}. To perform the restructuring, the function first assembles the sets of states $I=\{s_{i1}, s_{i2}, \ldots, s_{im}\}$ (line~\ref{l:assembleSetI}) and $O=\{s_{o1}, s_{o2}, \ldots, s_{on}\}$  (line~\ref{l:assembleSetO}), and then iterates through state pairs $(i,o)\in I\times O$, removing the transition from state $i$ to state $z$ (line~\ref{l:ModTran3}) and inserting a transition from state $i$ to state $o$. As a result of this modification of the original pDTMC, state $z$ meets the condition for becoming an output state of the fragment under construction, and will be placed into the set of outputs states $Z_\mathsf{OUT}$ in line~\ref{ReTE} from Algorithm~\ref{algorithm:stopgrow}.

\begin{algorithm}[t]
\caption{pDTMC restructuring }\label{algorithm:Restructuring}
		\renewcommand{\baselinestretch}{1}
		\begin{algorithmic}[1]
				\Function{RestructureTrans}{$D(S,s_0,\textbf{P}, L),Z,z$}
		        \State $I \gets \{i\!\in\! S\!\setminus\! Z \:|\:\textbf{P}(i,\!z)\!\neq\! 0\}$\label{l:assembleSetI}
			\State $O \gets \{o\!\in\! S\:|\:\textbf{P}(z,\!o)\!\neq\! 0\}$\label{l:assembleSetO}
			    \ForAll{$i \in I$}\label{l:ModPairwise}
			        \State $\textbf{P}(i,z) \gets 0$\label{l:ModTran3}
			        \ForAll{$o \in O$}\label{l:ModOut}
			            \State $\textbf{P}(i,o) \gets \textbf{P}(i,z)\cdot \textbf{P}(z,o)$\label{l:ModTran3b}
			        \EndFor
			    \EndFor\label{l:ModPairwiseE}
        \EndFunction
        \vspace{0.2cm}
		\Function{RestructureState}{$D(S,s_0,\textbf{P}, L),\mathit{rwd},Z,z$}
			\State $O\gets \{o\!\in\!S\!\setminus\! Z\:|\:\textbf{P}(z,o)\!\neq\! 0\}$ \label{l:AssembleSetO}
			\State $\mathit{NewStates} \gets \{ \}$ \label{l:NS}
			   \ForAll{$o \!\in\! O $}\label{l:addpair}
			        \State $z' \gets \textsc{NewState()}$ \label{l:addCreateV}
			        \State $S \gets S \cup \{z'\}$ \label{l:addCreateVE}
			        \ForAll{$s\in S$}     \label{l:setTransProbabilitiesS}
			          \State $\textbf{P}(s,z')\gets 0$, $\textbf{P}(z',s)\gets 0$
			        \EndFor
			        \State $\textbf{P}(z,z') \!\gets\! \ \textbf{P}(z,o)$, $\textbf{P}(z',o) \!\gets\! 1$ \label{l:addTransition}\label{l:setTransProbabilitiesE}
			        \State $L(z') \gets \{\}$ \label{l:addEmptyLabelSet}
			        \State $\mathit{rwd}(z') \gets 0$ \label{l:addZeroReward}
			        \State $\mathit{NewStates} \gets \mathit{NewStates} \cup \{z'\}$
			        \label{l:addUnion}
			   \EndFor\label{l:addpairE}
           		\State \Return {$NewStates$} \label{l:addReturn} 
        \EndFunction
		\end{algorithmic}
\end{algorithm}

\medskip
\noindent
\textbf{Scenario 2.} 
The state $z$ has no incoming transitions from outside the fragment under construction, but has outgoing transitions to one or more states inside the fragment in addition to $n\geq 1$ outgoing transitions (of probabilities $p_1$, $p_2$, \ldots, $p_n$) to states $s_1$, $s_2$, \ldots $s_n$ outside of the fragment (Figure~\ref{fig:trick1}--left). In this scenario, we augment the pDTMC with states $z'_1$, $z'_2$, \ldots, $z'_n$, and we replace each transition between states $z$ and $s_j$, $1\leq j\leq n$ with a transition of probability $p_j$ between $z$ and $z'_j$ \emph{and} a transition of probability $1$ between $z'_j$ and $s_j$. This change supports the formation of a fragment whose output state set includes the auxiliary states  $z'_1$, $z'_2$, \ldots, $z'_n$ (Figure~\ref{fig:trick1}--right), and is performed by the function \textsc{RestructureState} from  Algorithm~\ref{algorithm:Restructuring}. This function assembles a set $O$ comprising the states $s_1$, $s_2$, \ldots $s_n$ in line~\ref{l:AssembleSetO} and creates the set of $\mathit{NewStates}$ $z'_1$, $z'_2$, \ldots, $z'_n$ in the for loop from lines~\ref{l:addpair}--\ref{l:addpairE}. After it is created in line~\ref{l:addCreateV}, each new state is added to the state set $S$ in line~\ref{l:addCreateVE}, has its incoming and outgoing transition probabilities initialised in lines~\ref{l:setTransProbabilitiesS}--\ref{l:setTransProbabilitiesE}, is associated with an empty label set and with a zero reward in lines~\ref{l:addEmptyLabelSet} and~\ref{l:addZeroReward}, respectively, and is added to the set of $\mathit{NewStates}$ in line~\ref{l:addUnion}. This $\mathit{NewStates}$ set is returned in line~\ref{l:addReturn}, so that the function \textsc{Terminate}  can add the new states to both the set $Z$ of fragment states and the set $Z_\mathsf{OUT}$ of output fragment states (line~\ref{ReSE} from Algorithm~\ref{algorithm:stopgrow}). 

\medskip
Before providing correctness proofs for the \approach\ fragmentation in the next section, we note that the function \textsc{RestructureState} is also used to support the growing of the fragment under construction in the function \textsc{Traverse} from Algorithm~\ref{algorithm:fragmentgrow}. This use occurs when \textsc{Traverse} processes (in lines~\ref{l:formationAdd3} and~\ref{l:formationAdd3Push}) a state $z$ that has outgoing transitions to states belonging to previously constructed fragments. 

\subsection{Correctness of the \approach\ Fragmentation}
\label{sse:Proofs}

We start by showing that the pDTMC fragmentation produced by \approach\ is valid. 

\begin{theorem}
Function \textsc{Fragmentation} returns a valid fragmentation of the pDTMC $D(S,s_0,\textbf{P}, L)$ as restructured by its auxiliary functions.
\end{theorem}

\begin{proof}
We prove this result by showing that: (a)~the set $\mathit{FS}$ assembled by \textsc{Fragmentation} comprises only valid fragments; (b)~the fragments from $\mathit{FS}$ are disjoint (i.e., no state from $S$ is included in more than one fragment); and (c)~the function terminates and returns a set of fragments $\mathit{FS}$ that includes all the states from the state set $S$ of the restructured pDTMC.

To prove part (a), we note that new fragments are added to $\mathit{FS}$ in three lines from Algorithms~\ref{algorithm:fragmentation} and ~\ref{algorithm:fragmentgrow}. In line~\ref{l:satisfy} of Algorithm~\ref{algorithm:fragmentation} and in line~\ref{l:degenerate3} of 
Algorithm~\ref{algorithm:fragmentgrow}, $\mathit{FS}$ is augmented with tuples that represent degenerate, one-state fragments according to the definition given in Section~\ref{subsec:fragmentationTheory}. Finally, in line~\ref{l:addFragment} of Algorithm~\ref{algorithm:fragmentation}, $\mathit{FS}$ is augmented with a tuple $(Z,z_0,Z_\mathsf{OUT})$ that either passes the fragment-validity check from line~\ref{l:validfragmentS}, or is reduced to a degenerate, one-state fragment in line~\ref{l:degenerate2} before it is included in $\mathit{FS}$. As such, any tuple inserted into $\mathit{FS}$ is a valid fragment.

To prove part (b), we note that the states from the set $V$ assembled in line~\ref{l:initVisited} of \textsc{Fragmentation} are each placed into a degenerate, one-state fragment in line~\ref{l:satisfy} of Algorithm~\ref{algorithm:fragmentation}, and that the state set $Z$ of any fragment $(Z,z_0,Z_\mathsf{OUT})$ added to $\mathit{FS}$ comes from $S\setminus V$, where $V$ is updated (in line~\ref{l:addToR1} of Algorithm~\ref{algorithm:fragmentation}, and in line~\ref{l:addToR2} of Algorithm~\ref{algorithm:fragmentgrow}) to include all the states of new fragments included in $\mathit{FS}$. To see that the states of all new fragments come from $S\setminus V$, observe that: 
\squishlist
    \item[(i)] state $z_0$ added to $Z$ in line~\ref{l:startFragment} and line~\ref{l:degenerate2} of Algorithm~\ref{algorithm:fragmentation} comes directly from $S\setminus V$ (line~\ref{l:forS}); 
    \item[(ii)] state $z$ added to $Z$ in line~\ref{l:extendFragment} of the same algorithm (or included into $\mathit{FS}$ as the only vertex of a degenerate fragment in line~\ref{l:degenerate3} of Algorithm~\ref{algorithm:fragmentgrow}) comes from the stack $T$, which can only acquire vertices from $S\setminus V$ (as enforced by the if statements before lines~\ref{l:wAdded} and~\ref{l:pushO} from Algorithm~\ref{algorithm:fragmentgrow}, and by the use of states newly created by function \textsc{RestructureState} in line~\ref{l:formationAdd3Push} from Algorithm~\ref{algorithm:fragmentgrow}); 
    \item[(iii)] the states added to $Z$ in line~\ref{ReSE} from Algorithm~\ref{algorithm:stopgrow} are states newly created by function \textsc{RestructureState}, which do not belong to any existing fragment. 
\squishend
Therefore, the fragments from $\mathit{FS}$ are disjoint.

To prove part (c), we note that all the functions from Algorithms~\ref{algorithm:fragmentation}--\ref{algorithm:Restructuring} terminate. \textsc{RestructureTrans} and \textsc{RestructureState} from Algorithm~\ref{algorithm:Restructuring} terminate because each of their statements (including the assembly of the state sets $I$ and $O$ in \textsc{RestructureTrans} and of the state set $O$ in \textsc{RestructureState}, and their for loops) operate with finite numbers of states. As such, \textsc{Traverse} also terminates because it builds and operates with finite sets of states $I$ and $O$, and invokes a function  that terminates (i.e., \textsc{RestructureState}). The function \textsc{Terminate} contains no loops and invokes one of three functions, each of which is guaranteed to terminate; therefore, \textsc{Terminate} is also guaranteed to terminate. Finally, \textsc{Fragmentation} terminates because: 
\squishlist
    \item[(i)] each iteration of its for loop adds at least state $z_0$ from $S\setminus V$ to $V$ in line~\ref{l:addToR1}, until $S\setminus V = \{\}$ in line~\ref{l:forS} (since $S$ is a finite set of states) and the loop terminates with all states from $S$ included in fragments from $\mathit{FS}$; 
    \item[(ii)] its while loop terminates since it iterates over the elements of stack $T$ that can only contain one instance of each state from the finite set $S$ and is therefore finite; 
    \item[(iii)] \textsc{RestructureState} invocations can occur at most once for each state of the initial pDTMC, and can only add a finite number of states to $S$, as shown in Section~\ref{sse:MRAFF}. 
\squishend
Thus, \textsc{Fragmentation} terminates, returning a fragment set $\mathit{FS}$ that includes all the vertices from $S$.
\end{proof}

Having demonstrated that the function \textsc{Fragmentation} yields a valid fragmentation of the restructured version of the pDTMC received as its first argument, we will show next that using this restructured pDTMC instead of the original pDTMC to analyse the PCTL formula $\phi$ under verification does not change the PMC result. 

\begin{theorem} 
\label{th:model-equivalence}
Applying the model restructuring techniques from Algorithm~\ref{algorithm:Restructuring} and Figure~\ref{fig:ModelReconstruction} to a pDTMC does not affect its reachability, unbounded until and reachability reward properties.
\end{theorem}
\begin{proof}
We first show that the theorem holds for any reachability property $\phi=\mathcal{P}_{=?} [\mathrm{F}\, \Phi]$. To that end, we consider a generic pDTMC~$D$, the pDTMC~$D'$ obtained by applying one of the model restructuring techniques from Figure~\ref{fig:ModelReconstruction} to $D$, and the sets of all paths $\Pi$ over $D$ and $\Pi'$ over $D'$ that satisfy $\phi$. 
According to the semantics of PCTL, we need to show that $\mathrm{Pr}_{s_0}(\Pi)=\mathrm{Pr}'_{s_0}(\Pi')$, where $\mathrm{Pr}_{s_0}$ is a probability measure defined over all paths $\pi=s_0s_1s_2\ldots s_n$ starting in the initial state $s_0$ of $D$ such that $\mathrm{Pr}_{s_0}(\pi)=\prod_{i=0}^{n-1}\mathbf{P}(s_i,s_{i+1})$, and $\mathrm{Pr}'_{s_0}$ is a similarly defined probability measure for $D'$. We focus on the paths that differ between $\Pi$ and $\Pi'$, and show that $\mathrm{Pr}_{s_0}(\Pi\!\setminus\! \Pi')=\mathrm{Pr}'_{s_0}(\Pi'\!\setminus\! \Pi)$ for each technique from Figure~\ref{fig:ModelReconstruction} in turn:
\squishlist
\item For the technique from Figure~\ref{fig:trick2}, a path from $\Pi\!\setminus\! \Pi'$ has the form $\pi=s_0\omega_1s_{ij}zs_{ok}\omega_2$, with $j\in\{1,2,\ldots,m\}$,  $k\in\{1,2,\ldots,n\}$, and $\omega_1$, $\omega_2$ subpaths such that $\omega_2$ ends in a state that satisfies $\Phi$. Path $\pi$ has a corresponding path $\pi'=s_0\omega_1s_{ij}s_{ok}\omega_2\in \Pi'\!\setminus\!\Pi$ (and the other way around) such that
\[
\begin{array}{ll}
   \mathrm{Pr}_{s_0}(\pi)\!\!\!\!\! & = \mathrm{Pr}_{s_0}(s_0\omega_1s_{ij})\mathbf{P}(s_{ij},z)\mathbf{P}(z,s_{ok})
   \mathrm{Pr}_{s_{ok}}(s_{ok}\omega_2)\\
   & = \mathrm{Pr}'_{s_0}(s_0\omega_1s_{ij}) \mathbf{P}'(s_{ij},s_{ok})\mathrm{Pr}'_{s_{ok}}(s_{ok}\omega_2)\\
   & = \mathrm{Pr}'_{s_0}(\pi')
\end{array}
\]
since the restructuring technique guarantees that $\mathbf{P'}(s_{ij},s_{ok})=\mathbf{P}(s_{ij},z)\mathbf{P}(z,s_{ok})$.
\item For the technique from Figure~\ref{fig:trick1}, a path from $\Pi\!\setminus\! \Pi'$ has the form $\pi=s_0\omega_1 zs_i\omega_2$ for some $i\in\{1,2,\ldots, n\}$ and subpaths $\omega_1$, $\omega_2$, with $\omega_2$ ending in a state that satisfies $\Phi$. Path $\pi$ has a corresponding path $\pi'=s_0\omega_1 zz'_is_i\omega_2 \in \Pi'\!\setminus\! \Pi$ (and the other way around) such that 
\[
\begin{array}{ll}
   \mathrm{Pr}_{s_0}(\pi)\!\!\!\!\! & =\mathrm{Pr}_{s_0}(s_0\omega_1z)\mathbf{P}(z,s_i)\mathrm{Pr}_{s_i}(s_i\omega_2)\\
   & =\mathrm{Pr}'_{s_0}(s_0\omega_1z) \mathbf{P}'(z,z'_i)\mathbf{P}'(z'_i,s_i)\mathrm{Pr}'_{s_i}(s_i\omega_2)\\
   & =\mathrm{Pr}'_{s_0}(\pi') 
\end{array}
\]
since the restructuring technique guarantees that $\mathbf{P}'(z,z'_i)=1$ and $\mathbf{P}'(z'_i,s_i)=\mathbf{P}(z,s_i)$.
\squishend
We showed that neither of the restructuring techniques from Figure~\ref{fig:ModelReconstruction} affects the value of the reachability property $\phi$, and therefore the finite number of applications of these techniques within the \approach\ fragmentation approach do not affect this value either.

To show that the theorem holds for a generic unbounded until property $\phi=\mathcal{P}_{=?} [\Phi_1\, \mathrm{U}\, \Phi_2]$, we first note that the sets of paths $\Pi$ over $D$ and $\Pi'$ over $D'$ that satisfy $\phi$ are subsets of the path sets $\Pi_\mathsf{reach}$ and $\Pi'_\mathsf{reach}$ that satisfy the reachability property $\mathcal{P}_{=?} [\mathrm{F}\, \Phi_2]$ over $D$ and $D'$, respectively. We know from the first part of the proof that $\Pi_\mathsf{reach}$ and $\Pi'_\mathsf{reach}$ are equiprobable. Consider now a generic path $\pi\in\Pi_\mathsf{reach}\!\setminus\!\Pi$, i.e., a path that ends in a state that satisfies $\Phi_2$, but without any intermediate state where $\Phi_1$ is satisfied. We have two cases. If $\pi$ is unaffected by the restructuring technique used to obtain the pDTMC $D'$ from $D$, then $\pi\in\Pi'_\mathsf{reach}\!\setminus\!\Pi'$. Otherwise, the equiprobable path $\pi'\in\Pi_\mathsf{reach}$ constructed from $\pi$ as in the first part of the theorem will not be in $\Pi'$ because none of its intermediate states can satisfy $\Phi_2$. Indeed, any such states that are identical to states from $\pi$ do not satisfy $\Phi_1$ because no intermediate state of $\pi$ does, and any new states created by the pDTMC restructuring is labelled with an empty set of atomic propositions (in line~\ref{l:addEmptyLabelSet} of Algorithm~\ref{algorithm:Restructuring}) and thus does not satisfy any PCTL formula. We showed that path sets $\Pi$ and $\Pi'$ are obtained by removing equiprobable paths from the equiprobable path sets $\Pi_\mathsf{reach}$ and $\Pi'_\mathsf{reach}$. As such, $\Pi$ and $\Pi'$ are also equiprobable, and the theorem holds for unbounded until properties.

Finally, for a generic reachability reward property \mbox{$\phi=\mathcal{R}_{=?}^\mathit{rwd}[\mathrm{F}\, \Phi]$,} we note that the value of $\phi$ is given by a weighted sum of the probabilities of all DTMC paths that satisfy the associated reachability property  $\mathcal{P}_{=?} [\mathrm{F}\, \Phi]$, where the weight associated with a path $\pi=s_0s_1s_2\ldots$ is the cumulative reward $\mathit{rwd}(s_0)+\mathit{rwd}(s_1)+\mathit{rwd}(s_2)+\ldots$ for the states on the path. As shown in the first part of the theorem, the path formula $\mathrm{F}\, \Phi$ is satisfied by pairs of equiprobable paths $\pi$ and $\pi'$ over $D$ and $D'$, respectively. We will show that the paths in every such pair have the same cumulative reward. We have three cases. First, if $\pi$ is unaffected by the restructuring technique used to obtain the pDTMC $D'$ from $D$, then $\pi'=\pi$, and the two cumulative rewards are trivially equal. Second, when the restructuring from Figure~\ref{fig:trick2} is used, a state $z$ from path $\pi$ is skipped on path $\pi'$, but otherwise the two paths are identical. However, $z$ is in this case a zero-reward state (cf.~line~\ref{ReTS} from Algorithm~\ref{algorithm:stopgrow}), so the two cumulative rewards are equal. Finally, when the restructuring from Figure~\ref{fig:trick1} is used, path $\pi'$ only differs from $\pi$ through the inclusion of an auxiliary state $z'_i$. Since $\mathit{rwd}(z_i)=0$ (cf.\ line~\ref{l:addZeroReward} from Algorithm~\ref{algorithm:Restructuring}), the cumulative rewards for the two paths are again equal. As such, the theorem also holds for reachability reward properties.
\end{proof}

We have shown so far that \approach\ produces valid pDTMC fragmentations, and that the model restructuring used during this fragmentation does not impact the PMC of reachability, unbounded until and reachability reward properties. The next result establishes the complexity of the \approach\  fragmentation. To derive this result, we adopt the standard graph notation $\mathit{indegree}(s)=\#\{s'\in S\,|\, \mathbf{P}(s',s)\neq 0\}$, $\mathit{outdegree}(s)=\#\{s'\in S\,|\, \mathbf{P}(s,s')\neq 0\}$ and $\mathit{degree}(s)=\max\{\mathit{indegree}(s), \mathit{outdegree}(s)\}$ to denote the number of incoming transitions, the number of outgoing transitions, and the maximum between the two numbers, respectively, for a state $s$ of a pDTMC. 

\begin{theorem}
\label{th:complexity}
The function \textsc{Fragmentation} requires at most $\mathsf{O}(n^3d)$ steps, where $n$ represents the number of pDTMC states after model restructuring\footnote{\textsc{RestructureState} may add up to $\mathit{outdegree}(z)-1$ new states for each of the $n_0$ states of the initial pDTMC (i.e., of the pDTMC that \textsc{Fragmentation} receives as its first argument), yielding a restructured pDTMC with $n_0(d-1)$ vertices in the worst-case scenario.} and $d=\max_{s\in S} \mathit{degree}(s)$.
\end{theorem}
\begin{proof}
The function \textsc{RestructureTrans} (Algorithm~\ref{algorithm:Restructuring}) requires $O(n)$ steps to assemble the state sets $I$ and $O$, and at most $\mathsf{O}(d^2)$ steps to process up to $\mathit{indegree}(z)\cdot\mathit{outdegree}(z)$ pairs of incoming-outgoing transitions of a state $z$. As such, it has $O(\max\{n,d^2\})$ overall complexity. Given a state $z$, the function \textsc{RestructureState} requires $O(n)$ steps to assemble the state set $O$, and then $\mathsf{O}(nd)$ steps to initialise $2n$ transition probabilities for each new state it creates, since one new state is created for each of the up to $\mathit{outdegree}(z)-1<d$ outgoing transitions of $z$. Therefore, the overall complexity of \textsc{RestructureState} is $O(nd)$.

The function \textsc{Traverse} requires at most $\mathsf{O}(nd)$ steps due to the invocation of \textsc{RestructureState}, with its other operations (i.e., building the set $I$ and placing it onto the stack $T$, and building the set $O$ and placing $O\!\setminus\! V$ onto the stack $T$) performed in $\mathsf{O}(n)$ time.

The function \textsc{Terminate} requires $O(n)$ time to evaluate the conditions whose values determine which of the functions \textsc{RestructureTrans}, \textsc{RestructureState} and \textsc{Traverse} it needs to invoke. As such, the complexity of \textsc{Terminate} is given by the highest complexity among these three functions, i.e., $O(nd)$ for both \textsc{RestructureState} and \textsc{Traverse}. Note that this complexity is higher than the $O(\max\{n,d^2\})$ complexity of \textsc{RestructureTrans} since $nd\geq n$, and (because $n\geq d$) $nd\geq d^2$.

In the worst-case scenario where the fragmentation produces only one-state fragments, the execution of \textsc{Fragmentation} requires the execution of its for loop from lines~\ref{l:forS}--\ref{l:forE} for each of the $n_0\leq n$ states of the initial pDTMC (with any new states created by \textsc{RestructureState} included into the same fragment as the state that led to their creation, cf.~Figure~\ref{fig:trick1}). Each iteration of this loop executes: 
\squishlist
\item[i)] \textsc{Traverse} (line~\ref{l:grow1}) in $\mathsf{O}(nd)$ steps; 
\item[ii)] a while loop (lines~\ref{l:whileS}--\ref{l:whileE}) with at most $n$ iterations (one for each pDTMC state) that may each invoke the $\mathsf{O}(nd)$-step \textsc{Traverse} or the $\mathsf{O}(nd)$-step \textsc{Terminate}, yielding an $O(n^2d)$ complexity for the while loop; 
\item[iii)] the fragment validity check from line~\ref{l:validfragmentS}, which requires no more than $\mathsf{O}(nd)$ operations. 
\squishend
Thus, each iteration of the for loop from lines~\ref{l:forS}--\ref{l:forE} is completed in no more than $\mathsf{O}(n^2d)$ steps (due to the while loop from lines~\ref{l:whileS}--\ref{l:whileE}), and the entire \textsc{Fragmentation} requires $\mathsf{O}(n^3d)$ steps in the worst-case scenario.
\end{proof}

We note that the coefficients associated with $n$ and $d$ from the big-$\mathsf{O}$ notation in Theorem~\ref{th:complexity} are typically well below $1$. For instance, in all our experiments, the for loop from \textsc{Fragmentation} was only executed for a small fraction of the pDTMC states (because many fragments with multiple states are typically produced), and \textsc{Terminate},  \textsc{RestructureTrans} and \textsc{RestructureState} were only executed sparingly. Furthermore, it is worth noting that pDTMCs are typically sparsely connected graphs, and therefore $d$ is relatively small.

\begin{figure}
	\begin{subfigure}[b]{\hsize}
         \centering
	\includegraphics[width=0.82\linewidth]{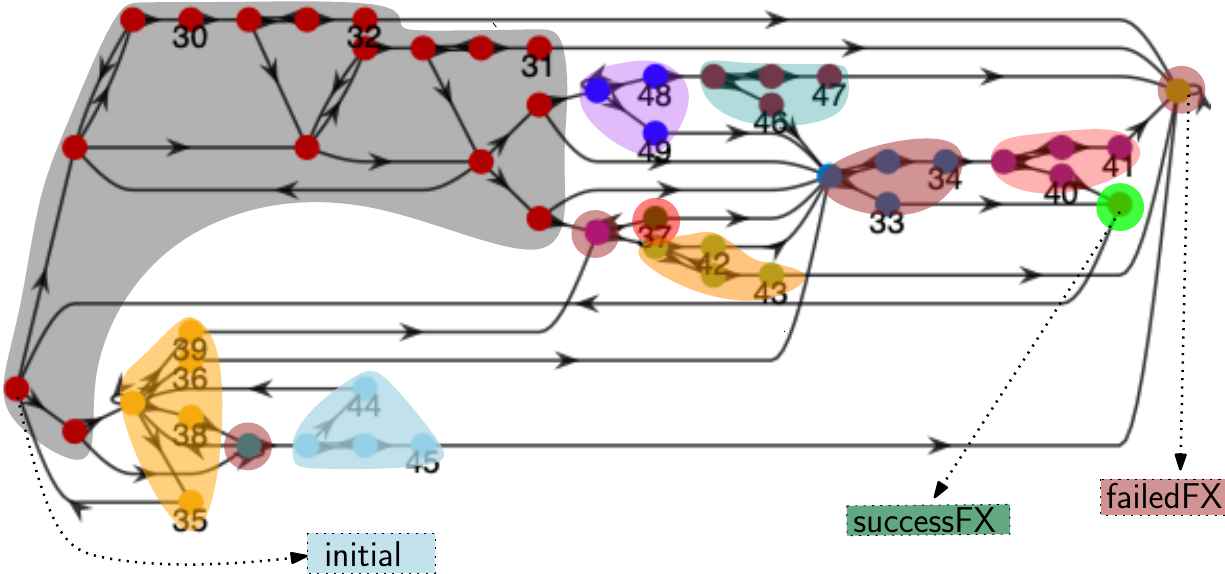}	
	\label{ig:runningExampleApplicationReachability}
    \caption{pDTMC fragmentation for the reachability property from the first row of Table~\ref{Propsummary}: 13~fragments were obtained, including five one-state fragments.}
    \end{subfigure}
    
         \vspace*{2mm}
   \begin{subfigure}[b]{\hsize}
	\centering
    \includegraphics[width=0.82\textwidth]{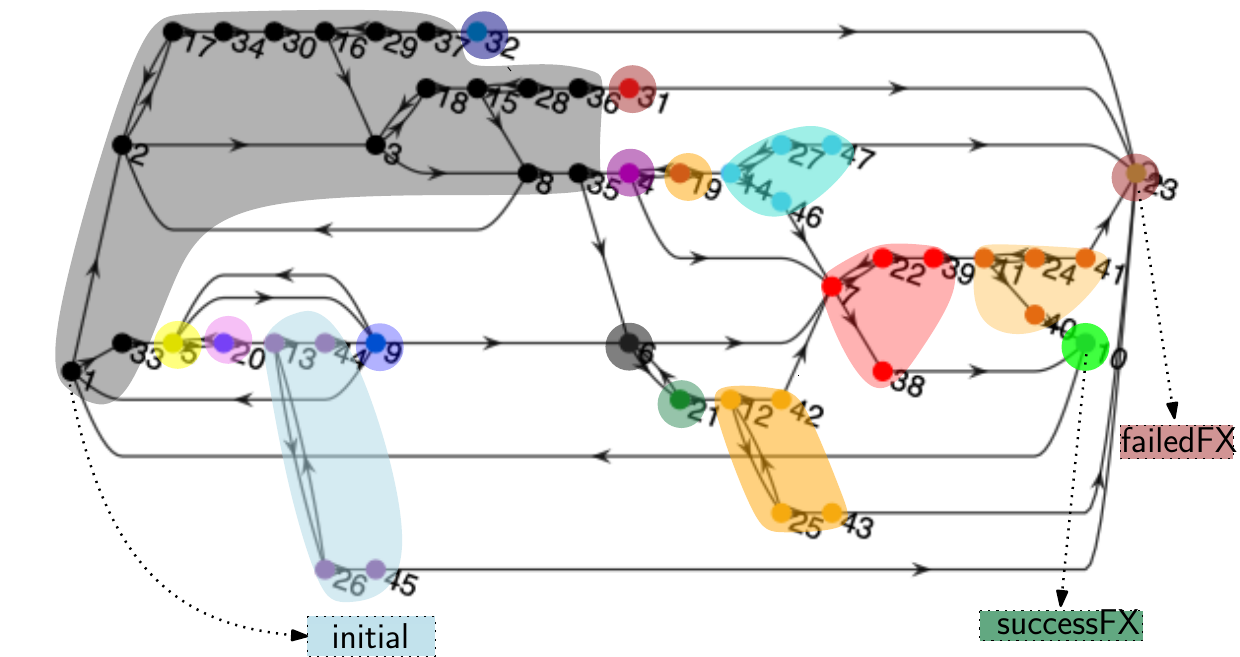}
    \caption{pDTMC fragmentation for the reachability reward property from the second row of Table~\ref{Propsummary} and the ``time'' reward function from Figure~\ref{fig:fxmodel}: 17~fragments were obtained, including 11~one-state fragments.}
    \label{fig:runningExampleApplicationCost}
  \end{subfigure}
  
           \vspace*{2mm}
\begin{subfigure}[b]{\hsize}
	\centering
    \includegraphics[width=0.82\textwidth]{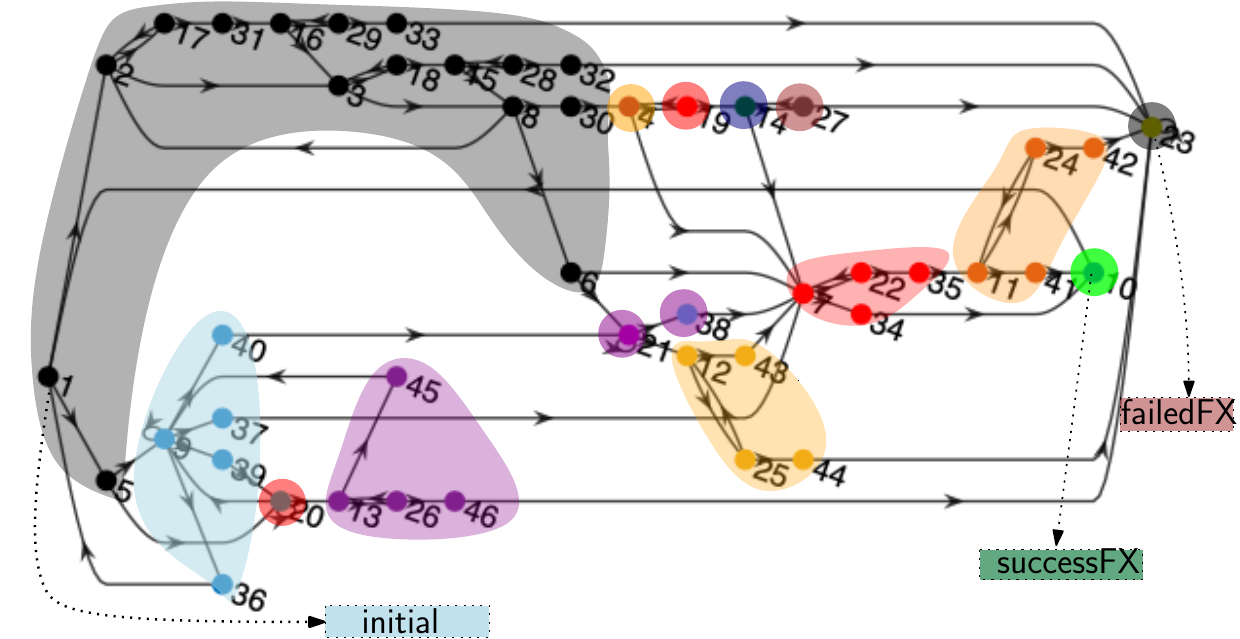}
    \caption{pDTMC fragmentation for the unbounded until property from the last row of Table~\ref{Propsummary}:  15~fragments were obtained, including nine one-state fragments.}
    \label{fig:runningExampleApplicationUnbounded}
\end{subfigure}
\caption{\approach\ fragmentation of the pDTMC from Figure~\ref{fig:fxmodel} for the PCTL properties from Table~\ref{Propsummary} and $\alpha=5$, with different shading used to highlight each fragment. 
\label{fig:runningExampleApplication}}
\end{figure}

\subsection{\emph{\approach} Application to the Motivating Example}\label{ssec:fpmcapplication}

\begin{table*}
\sffamily
    \centering
     \caption{Summary of \approach\ application to the pDTMC and properties of the FX system in Section~\ref{sec:example}}
     \label{table:illustration}
    \begin{tabular}{p{2mm} p{3.8cm} p{1.4cm} p{1.4cm} p{1.4cm} p{3.2cm} p{2.2cm}}
    \toprule
    & & \multicolumn{2}{c}{\textbf{Restructured pDTMC}$^\dagger$} & & \multicolumn{2}{c}{\textbf{\approach\ closed-form analytical model (cf.~Fig.~\ref{fig:approach})}}\\ \cmidrule{3-4} \cmidrule{6-7}
    \textbf{ID} & \textbf{Property type} & \textbf{States} & \textbf{Transitions} & \textbf{\approach\ time} & \textbf{Arithmetic operations} & \textbf{Evaluation time} \\ \midrule
  P1 & Reachability & 49 & 83 & 7.2s & 1456 & 2ms\\
  P2 & Reachability reward & 47 & 76 & 7.7s & 2020 & 30ms\\
  P3 & Unbounded until & 46 & 79 & 5.1s & 1224 & 5ms\\
           \bottomrule
           \\[-2mm]
           \multicolumn{7}{l}{$^\dagger$starting from the initial FX pDTMC model with 29~states and 58~transitions in Figure~\ref{fig:fxExample}}
    \end{tabular}
\end{table*}

We illustrate the use of \approach\ by using the pDTMC model and properties from our running example (cf.~Figure~\ref{fig:fxmodel} and Table~\ref{Propsummary}). The leading parametric model checkers Storm and PRISM time out without producing results for any of these properties within 60~minutes of running on the computer with the specification provided in Section~\ref{Sec:ExperimentalSetup}. The outcome of applying \approach\ to this pDTMC and each of the three PCTL properties from Table~\ref{Propsummary} (with fragmentation threshold $\alpha=5$) are summarised in  Figure~\ref{fig:runningExampleApplication} and Table~\ref{table:illustration}. 

Figure~\ref{fig:runningExampleApplication} depicts the pDTMC fragments generated by \approach, which are different for each property because the function \textsc{Fragmentation} from Algorithm~\ref{algorithm:fragmentation} starts by creating property-specific sets of one-state fragments in lines~\ref{l:initVisited} and~\ref{l:satisfy}. Table~\ref{table:illustration} shows:
\begin{itemize}
    \item the change in pDTMC size due to \approach\ restructuring (an increase of 59\%--69\% in the number of states, and 31\%--43\% in the number of transitions compared to the initial pDTMC from Figure~\ref{fig:fxExample});
    \item the time taken by the end-to-end \approach\ process from Figure~\ref{fig:approach} (under 8s for each property by running the tool presented in the next section on the computer with the specification provided in Section~\ref{Sec:ExperimentalSetup});
    \item the number of arithmetic operations from the algebraic formulae of the \approach\ closed-form analytical model for each analysed property, and the time required to evaluate these algebraic formulae for a given parameter valuation by running MATLAB on the computer from Section~\ref{Sec:ExperimentalSetup} (up to 30ms for the evaluation of the 2020-operation formulae of property P2).
    \end{itemize}
Given the large size of the \approach\ algebraic formulae for the three FX properties, we do not include them in the paper; they are provided on our project website~\url{https://www.cs.york.ac.uk/tasp/fPMC}.

\section{Implementation}\label{sec:implementation}

We developed a parametric model checking tool that implements the \approach\ algorithms presented in the previous section. This tool uses:
\begin{itemize}
\item the model checker PRISM, to identify the pDTMC states that satisfy $\Phi$, $\Phi_1$ and $\Phi_2$ in line~\ref{l:initVisited} of Algorithm~\ref{algorithm:fragmentation}, and to obtain the transition probability matrix $\mathbf{P}$ (used throughout the \approach\ algorithms) from a pDTMC specified in the PRISM modelling language (e.g., see Figure~\ref{fig:fxmodel});
\item the model checker Storm, to apply standard PMC to each pDTMC fragment and to the abstract model from Figure~\ref{fig:approach}.
\end{itemize}

As pDTMCs with small numbers of parameters are already handled extremely efficiently by Storm, our tool only invokes the end-to-end \approach\ approach for pDTMCs whose number of such parameters exceeds a user-configurable threshold $\beta\!\in\! \mathbb{N}_{>0}$. For pDTMCs with up to $\beta$ model parameters, the tool invokes Storm directly, and the PMC is performed on the unfragmented model. According to our experimental results (presented in Figure~\ref{fig:fxExampleheri}), $\beta$ values in the range $21..30$ work well for most models.

\section{Evaluation}
\label{sec:evaluation}

\subsection{Evaluation methodology}
\label{Sec:ExperimentalSetup}
We carried out extensive experiments to answer the research questions summarised below.

\medskip
\noindent
\textbf{RQ1 (Efficiency): Does \approach\ model fragmentation improve the efficiency of parametric model checking?} 
We assess if our \approach\ approach speeds up parametric model checking in comparison to PRISM~\cite{prism} and Storm~\cite{storm}, and whether it can handle pDTMCs that cannot be analysed by the two leading model checkers (with a 60-minute timeout). 

\medskip
\noindent
\textbf{RQ2 (Result complexity):  Does \approach\ reduce the complexity of the closed-form formulae generated by parametric model checking?} We assess whether the \approach-computed algebraic formulae are simpler (in terms of number of arithmetic operations) than those computed by the leading model checkers, and whether they can be evaluated faster than those produced by Storm.

\medskip
\noindent
\textbf{RQ3 (Configurability): How does the fragmentation threshold $\alpha$ affect the results of \approach?} 
We examine how different fragmentation threshold values affect \approach\ in terms of the number of operations from the computed closed-form formulae, and execution time for producing those formulae. 
\begin{table*}
    \sffamily
    \centering
     \caption{Key characteristics of the systems and pDTMC models used for the \approach\ evaluation}
     \label{Table:modelSummary}
    \begin{tabular}{p{4.9cm} p{3.5cm} p{3.8cm} p{4cm}}
    \toprule
    {} & \textbf{FX system} & \textbf{PL system} & \textbf{COM process} \\ \midrule

    \textbf{Application domain}&  Financial  & Vending machine controller & Communication protocol \\ \midrule
    \textbf{System type}&  Service-based system  & Software product line & Middleware \\ \midrule
    \textbf{Number of model variants} &  21  & 40 & 1  \\
    \textbf{Number of states} & 11--208 & 92--115 & 234 \\
    \textbf{Number of transitions} & 22--399 & 167--198 & 444 \\ 
    \textbf{Number of model parameters} & 11--71 & 10--198 & 20 \\ \midrule

    \multirow{3}{*}{\textbf{Analysed properties}} & Reachability \newline  Reward $\times$ 2 \newline Unbounded until & Reachability \newline   Unbounded until & Reachability $\times$ 21 \newline Reward \\ \midrule
    
    \textbf{Sample reachability property} & $\mathcal{P}_{=?}[\mathrm{F}\; \mathsf{successFX}]$  & $\mathcal{P}_{=?}[\mathrm{F}\; \mathsf{SUCCESS}]$ & $\mathcal{P}_{=?}[\mathrm{F}\; \neg\mathsf{a0}\wedge\neg\mathsf{a1}\wedge\ldots\wedge\neg\mathsf{a19}]$ \\ 
    
    \textbf{Sample reachability reward property} & $\mathcal{R}^\mathit{time}_{=?}[\mathrm{F}\; \mathsf{successFX} \vee \mathsf{failFX}]$ & -- & $\mathcal{R}^\mathit{coin\_flips}_{=?}[\mathrm{F}\; \mathsf{b} ]$ \\ 
    
    \textbf{Sample unbounded until property} & $\mathcal{P}_{=?}[\neg\mathsf{alarm} \; \mathrm{U} \; \mathsf{successFX}]$                & $\mathcal{P}_{=?}[\neg \mathsf{Function1} \; \mathrm{U} \; \mathsf{SUCCESS}]$            & -- \\ \bottomrule
    \end{tabular}
\end{table*}

\medskip
Three software systems and processes taken from related research~\cite{RaduePMC,ghezzi2013model,hajnal2019data,classen2010model,Gerasimou2015:ASE} and belonging to different application domains were used in our evaluation. These systems were selected because their Markov models contain: 
\begin{itemize}
    \item[(i)] multiple transition probabilities that can be meaningfully specified as functions over a set of system parameters;  
    \item[(ii)] configuration variables that can be instantiated to obtain pDTMCs of different sizes (i.e., with a broad range of state and transition numbers). 
\end{itemize}
The three systems are described below, and the significant differences between their key characteristics are summarised in Table~\ref{Table:modelSummary}. The full details of these models are available in~\cite{RaduePMC,ghezzi2013model,hajnal2019data,classen2010model,Gerasimou2015:ASE}.

\medskip
\noindent
\textbf{FX system.} 
We introduced this system in Section~\ref{sec:example}, and its pDTMC corresponding to the sequential execution strategy with retry (SEQ\_R) in Fig.~\ref{fig:fxmodel} and~\ref{fig:fxExample}.
For the \approach\ evaluation in this section, we also considered the additional strategies below (applied to between one and five functionally equivalent service implementations per FX operation):
\begin{itemize}
    \item SEQ---the services are invoked in order, stopping after the first successful invocation or after the last service fails;
    \item PAR---all services are invoked in parallel (i.e., simultaneously), and the operation uses the result returned by the first service terminating successfully;
    \item PROB---a probabilistic selection is made among the available services;
    \item PROB\_R---similar to PROB, but if the selected service fails, it is retried with a given probability (as in SEQ\_R).
\end{itemize}

\noindent
\textbf{PL system.} 
We used a pDTMC model of a product line (PL) system taken from~\cite{ghezzi2013model,classen2010model}. This pDTMC models the software controller of a vending machine that dispenses a user-selected beverage and, if applicable, takes payment from and gives back change to the user. 

The possible features of this system comprise: the beverage type (soda, tea, or both), the payment mode (cash or free), and the taste preference (e.g., add lemon or sugar). This variability enables the derivation of vending machines---and the specification of associated pDTMCs---with between four and 22 features. 

\medskip
\noindent
\textbf{COM process.} 
We considered a communication (COM) process among $n\geq 2$ agents taken from~\cite{hajnal2019data}, and inspired by the way in which honeybees emit an alarm pheromone to recruit workers and protect their colonies from intruders.  
Given the self-destructive defence behaviour in social insects (the recruited workers die after completing their defence actions), a balance between efficient defence and preservation of a critical mass of workers is required.  The induced pDTMC is a stochastic population model with $n$ parameters. The quantitative analysis of such stochastic models of multi-agent systems is often challenging because the dependencies among the agents within the population make the models complex. 

\medskip
Throughout the evaluation, \approach\ is compared to the leading PMC model checkers PRISM (version 4.6) and Storm (version 1.5.1), both with their default settings. All experiments were performed on a MacBook Pro with 2.7GHz dual Core Intel i5 processor and 8GB RAM, using a timeout of 60~minutes. For a fair comparison, we ensured that both PRISM and Storm can successfully process at least the simplest pDTMC of each systems. The following experimental data were collected:
\begin{enumerate}
\item the time required to compute the PMC formulae; 
\item the number of arithmetic operations in the PMC formulae;\footnote{PRISM and Storm produce a single PMC formula per property, whereas \approach\ yields a set of formulae per property (cf.~Figure~\ref{fig:approach}).} 
\item the time required to evaluate the PMC formulae (in MATLAB) for a parameter valuation. 
\end{enumerate}

To ensure the correctness of our \approach\ tool, we evaluated the \approach\ formulae produced for each analysed property using randomly generated combinations of parameter values, and confirmed that the resulting property value matched that produced by PRISM and Storm (subject to negligible rounding errors). For the purpose of this check, the PRISM and Storm results were obtained by running the probabilistic model checking on the non-parametric DTMC obtained by replacing the pDTMC parameters with the relevant combination of parameter values. While this is not a formal proof that the \approach\ tool was implemented correctly, we note that even small alterations of the non-trivial formulae generated by the tool yield noticeable changes in the evaluation results, so this random testing strongly suggests that the \approach\ tool operates correctly.

The source code of our \approach\ tool, the models, properties and results for all the experiments presented in the paper, as well as additional materials supporting the adoption of \approach\ are available on our project website at~\url{https://www.cs.york.ac.uk/tasp/fPMC}.

\subsection{Results and discussion}
\label{R&D}
\subsubsection{RQ1 (Efficiency)}
\label{sec:Efficiency}

\begin{table*}
\sffamily
\centering
\caption{
Parametric model checking times (in seconds, `--' indicates a timeout) for the 21~FX pDTMC variants and their four properties; \textsf{STG, \#SVC, \#S} and \textsf{\#T} represent the strategy used to invoke multiple functionally equivalent services for an FX operation, the number of such services, and the numbers of pDTMC states and transitions, respectively.}
\label{Table:FXTime}
\sffamily
\def\tabcolsep{2.2pt}
\begin{tabular}{rcrrr @{\hspace*{7mm}} rrr @{\hspace*{7mm}} rrr @{\hspace*{7mm}} rrr @{\hspace*{7mm}} rrr}
\toprule
\multicolumn{5}{c}{\hspace*{-7mm}\textbf{pDTMC variant}} & \multicolumn{3}{c}{\hspace*{-7mm}\textbf{P1 PMC time}} & \multicolumn{3}{c}{\hspace*{-5mm}\textbf{P2 PMC time}} & \multicolumn{3}{c}{\hspace*{-7mm}\textbf{P3 PMC time}} & \multicolumn{3}{c}{\textbf{P4 PMC time}} \\ \cmidrule(l{1pt}r{17pt}){1-5} \cmidrule(l{-1pt}r{17pt}){6-8} \cmidrule(l{1pt}r{17pt}){9-11} \cmidrule(l{1pt}r{17pt}){12-14} \cmidrule(l{0pt}r{2pt}){15-17}

\textbf{ID} & \textbf{STG} & \textbf{\#SVC} &  \textbf{\#S} &  \textbf{\#T} & 
\textbf{fPMC} &  \textbf{Storm} & \textbf{PRISM} &  
\textbf{fPMC} &  \textbf{Storm} & \textbf{PRISM} &  
\textbf{fPMC} &  \textbf{Storm} & \textbf{PRISM} &  
\textbf{fPMC} &  \textbf{Storm} & \textbf{PRISM}  \\ \midrule
 
 1 & --                  & 1 & 11 & 22 & 2.9 & 0.01 & 0.4 & 3.9 & 0.03 & 0.8 & 2.9 & 0.01 & 0.6 & 3.7 & 0.02 & 0.8  \\ \midrule
2 & \multirow{4}{*}{SEQ} 
                      & 2 & 17 & 34 & 3.5 & 2.5 & 57.1 & 4.7 & 2.6 & 5.7 & 3.4 & 0.5 & 17.5 & 4.5 & 2.5 & 5.6 \\
3&                      & 3 & 23 & 46 & 3.8 & -- & -- & 6.9 & -- & -- & 3.5 & 1127.5 & --  & 6.8 & -- & -- \\
4&                      & 4 & 29 & 58 & 5.0 & -- & -- & 8.6 & -- & -- & 5.0 & -- & -- & 9.0 & -- & -- \\
5&                      & 5 & 35 & 70 & 5.7 & -- & -- & 11.8 & -- & -- & 5.2 & -- & -- & 11.1 & -- & -- \\ \midrule
6& \multirow{4}{*}{PAR} 
                     & 2 & 40 & 36 & 5.6 & 1.8 & 2.4 & 6.2 & 3.3 & 3.5 & 5.0 & 0.6 & 2.0 & 6.4 & 1.6 & 3.4 \\
7&                 & 3 & 64 & 111 & 6.8 & -- & -- & 7.5 & -- & -- & 6.4 & -- & -- & 9.2 & -- & -- \\
8&                     & 4 & 112 & 207 & 14.3 & -- & -- & 21.9 & -- & -- & 14.7 & -- & -- & 32.5 & -- & -- \\
9&                     & 5 & 208 & 399 & 31.2 & -- & -- & 45.4 & -- & -- &  26.6 & -- & -- & 78.0 & -- & -- \\ \midrule
10&\multirow{4}{*}{PROB} 
                      & 2 & 23 & 46 & 3.4 & 0.4 & 10.0 & 5.8 & 0.6 & -- & 4.5 & 0.2 & 5.0 & 6.1 & 0.6 & -- \\
11&                      & 3 & 29 & 64 & 5.2 & 3.4 & -- & 8.6 & 4.9 & -- & 5.2 & 1.1 & -- & 7.6 & 5.0 & -- \\
12&                      & 4 & 35 & 82 & 6.7 & 27.0 & -- & 11.7 & 46.0 & -- & 6.7 & 10.0 & -- & 11.4 & 48.6 & -- \\
13&                      & 5 & 65 & 130 & 8.2 & 592.8 & -- & 16.5 & 611.3 & -- & 8.1 & 153.3 & -- & 16.3 & 580.3 & -- \\ \midrule
14&\multirow{4}{*}{SEQ\_R} 
                        & 2 & 29 & 58 & 7.3 & -- & -- & 8.4 & -- & -- & 5.1 & -- &  -- & 7.8 & -- & -- \\                      
15&                        & 3 & 41 & 28 & 18.8 & -- & -- & 35.0 & -- & -- & 18.9 & --  & -- & 35.2 & -- & -- \\
16&                        & 4 & 53 & 106 & 85.7 & -- & -- & 152.3 & -- & -- &  49.7 & -- & -- & 159.0 & -- & -- \\
17&                        & 5 & 65 & 130 & 496.1 & -- & -- & 1305.4 & -- & --  &  331.0 & -- & -- & 1335.9 & -- & --\\ \midrule
18& \multirow{4}{*}{PROB\_R} 
                         & 2 & 29 & 58 & 8.1 & 34.0 & 54.8 &  8.8 & -- & -- & 4.1 & 3.2 & -- &  8.3 & -- & -- \\
19&                         & 3 & 35 & 75 & 17.0 & -- & -- & 23.6 & -- & -- &  12.6 & -- & -- &  23.2 & -- & -- \\
20&                         & 4 & 41 & 93 & 65.3 & -- & -- & 79.5 & -- & -- &  52.2 & -- & -- &  80.4 & -- & -- \\
21&                         & 5 & 47 & 111 & 199.7 & -- & -- & 243.5 & -- & -- &  171.1 & -- & -- & 244.3 & -- & -- \\ \bottomrule
\end{tabular}
\end{table*}

\textbf{FX system.} We used \approach, PRISM and Storm to analyse pDTMC models corresponding to 21~variants of the FX system. The first of these system variants used a single service for each FX operation, and the remaining variants used one of the five execution strategies (i.e., SEQ, SEQ\_R, PAR, PROB or PROB\_R) and between two and five functionally equivalent services for each FX operation. For each of the 21~pDTMCs, four properties were analysed: the properties \textsf{P1}, \textsf{P2} and \textsf{P3} from Table~\ref{Propsummary}, and an additional reachability reward property (property \textsf{P4}) used to establish the expected cost of executing the FX workflow. 

The parametric model checking times for these experiments are presented in Table~\ref{Table:FXTime}. These results show that \approach\ successfully computed all PMC formulae well ahead of the 60-minute timeout for all four properties. It took \approach\ just $2.9s$ to analyse properties \textsf{P1} and \textsf{P3} for the simplest pDTMC variant (ID~1), and under $600s$ for analysing each of the four properties for most of the other models. Only the analyses of properties \textsf{P2} and \textsf{P4} for the pDTMC variant with ID~17 required more time, i.e., $1305.4$s and $1335.9$s, respectively. In contrast, Storm only completed the analysis for 31 of the $21\times 4 = 84$ model--property combinations before the 60-minute timeout. These combinations correspond to the simplest pDTMC variants (which Storm analysed slightly faster than \approach) across all execution strategies except the PROB strategy. For this strategy, Storm produced PMC formulae for all model-property combinations, but with an execution time that increased very quickly over the \approach\ time for the more complex PROB pDTMC variants with four and five functionally equivalent services per FX operation (i.e., the pDTMC variants with IDs 12 and 13 in the table). PRISM completed the analysis for even fewer model--property combinations: only 15 of the 84 PMC analyses returned results within 60~minutes. These results correspond again to the simplest pDTMC variants. 

\begin{table}
\sffamily
\centering
\caption{
Parametric model checking times (in seconds, `--' indicates a timeout) for the 40 PL pDTMC variants and their two properties; \textbf{\#F} and \textbf{\%PAR} represent the number of features in the model, and the percentage of parametric transitions, respectively.}
\label{Table:PLTime}
\sffamily
\def\tabcolsep{3.9pt}
\begin{tabular}{rrrrrr @{\hspace*{2mm}} rrr }
\toprule
\multicolumn{3}{c}{\hspace*{-1mm}\textbf{pDTMC variant}$^\dagger$} & \multicolumn{3}{c}{\textbf{Reachability}} & \multicolumn{3}{c}{\textbf{Unbounded until}} \\ \cmidrule(l{1pt}r{5pt}){1-3} \cmidrule(l{-1pt}r{5pt}){4-6} \cmidrule(l{0pt}r{1pt}){7-9}

\textbf{ID} & \textbf{\#F} & \textbf{\%PAR} & \textbf{fPMC} &  \textbf{Storm} & \textbf{PRISM} &  
\textbf{fPMC} &  \textbf{Storm} & \textbf{PRISM}  \\ \midrule
1 & \multirow{10}{*}{4} & 10 & 48.4 & 0.01 & 0.2 & 23.8 & 0.007 & 0.3  \\ 
2 & & 20 & 56.6 & 0.2 & 0.3 & 26.9 & 0.1 & 0.6 \\
3 & & 30 & 55.2 & 3.1 & 1.4 & 27.4 & 3.1 & 1.5\\ 
4 & & 40 & 56.3 & 32.0 & 5.1 & 29.3 & 22.6 & 5.3\\ 
5 & & 50 & 73.5 & 1671.1 & 105.5 & 37.4 & 711.6 & 56.8  \\ 
6 & & 60 & 81.8 & -- & -- & 42.2 & -- & -- \\ 
7 & & 70 & 86.5 & -- & -- & 51.8 & -- & -- \\ 
8 & & 80 & 92.2 & -- & -- & 50.2 & -- & -- \\ 
9 & & 90 & 93.0 & -- & -- & 52.1 & -- & -- \\ 
10 & & 100 & 93.1 & -- & -- & 51.2 & -- & -- \\ \midrule
11 & \multirow{10}{*}{16} & 10 & 46.8 & 0.3 & 14.9 & 16.3 & 0.1 & 5.1  \\ 
12 & & 20 & 50.1 & 15.7 & -- & 16.2 & 6.3 & -- \\ 
13 & & 30 & 19.9 & 23.4 & -- & 16.5 & 3.4 & -- \\ 
14 & & 40 & 19.8 & 31.6 & -- & 16.3 & 3.8 & -- \\ 
15 & & 50 & 22.6 & 164.6 & -- & 17.3 & 15.7 & -- \\ 
16 & & 60 & 22.4 & 164.9 & -- & 17.7 & 92.3 & -- \\ 
17 & & 70 & 23.1 & -- & -- & 17.8 & -- & -- \\ 
18 & & 80 & 24.0 & -- & -- & 19.1 & -- & --  \\ 
19 & & 90 & 26.4 & -- & -- & 18.9 & -- & --  \\ 
20 & & 100 & 23.8 & -- & -- & 18.9 & -- & --   \\ \midrule
21 & \multirow{10}{*}{18} & 10 &  14.0 & 0.04 & 0.9 & 13.5 & 0.04 & 0.3 \\ 
22 & & 20 & 13.9 & 0.1 & 1.1 & 13.3 & 0.1 & 0.5  \\ 
23 & & 30 & 14.1 & 0.5 & 2.4 & 13.3 & 0.5 & 1.5  \\ 
24 & & 40 & 13.9 & 18.7 & 41.6 & 13.5 & 16.4 & 27.6 \\ 
25 & & 50 & 14.9 & 911.3 & 86.1 & 14.2 & 840.6 & 69.8 \\ 
26 & & 60 & 14.3 & -- & 188.9 & 14.1 & -- & 78.5 \\ 
27 & & 70 & 14.6 & -- & -- & 14.1 & -- & -- \\ 
28 & & 80 & 17.7 & -- & -- & 15.3 & -- & -- \\ 
29 & & 90 & 17.4 & -- & -- & 16.1 & -- & -- \\ 
30 & & 100 & 17.1 & -- & -- & 15.6 & -- & -- \\ \midrule
31 & \multirow{10}{*}{22} & 10 & 38.6 & 0.03 & 1.8 & 37.8 & 0.035 & 0.5  \\ 
32 & & 20 & 45.9 & 3.5 & 53.2 & 41.1 & 3.3 & 23.4 \\ 
33 & & 30 & 42.4 & 42.3 & -- & 41.5 & 39.3 & -- \\ 
34 & & 40 & 42.4 & 793.4 & -- & 41.2 & 569.3 & -- \\ 
35 & & 50 & 42.3 & -- & -- & 42.1 & -- & -- \\ 
36 & & 60 & 42.5 & -- & -- & 42.0 & -- & -- \\ 
37 & & 70 & 43.0 & -- & -- & 43.3 & -- & -- \\ 
38 & & 80 & 43.3 & -- & -- & 43.5 & -- & -- \\ 
39 & & 90 & 48.4 & -- & -- & 48.9 & -- & -- \\ 
40 & & 100 & 48.6 & -- & -- & 49.2 & -- & -- \\\bottomrule
\\[-2.5mm]
\multicolumn{9}{l}{$^\dagger$pDTMCs sizes: four-feature models = 92 states, 167 transitions}\\
\multicolumn{9}{l}{\hspace*{2.1cm}16-feature models = 110 states, 193 transitions}\\
\multicolumn{9}{l}{\hspace*{2.1cm}18-feature models = 104 states, 183 transitions}\\
\multicolumn{9}{l}{\hspace*{2.1cm}22-feature models = 115 states, 198 transitions} 
\end{tabular}
\end{table}

\medskip
\noindent
\textbf{PL system.} We used \approach, PRISM and Storm to analyse pDTMC models corresponding to system configurations with four, 16, 18 and 22 software product line features, and with increasing numbers of parameters. To that end, we used different parameters for 10\%, 20\%, \ldots, 100\% of the transition probabilities of the models for the four system configurations, obtaining $4\times 10=40$ pDTMC variants. The reachability and unbounded until properties from Table~\ref{Table:modelSummary} were analysed for each of these pDTMC variants, and the time taken by these analyses are reported in Table~\ref{Table:PLTime}.

The results are similar to those obtained for the FX system. \approach\ produced all PMC formulae successfully in between 13.3 and 93.1 seconds, while Storm and PRISM completed only 50\% and 35\% of the PMC analyses, respectively, before the 60-minute timeout. The two existing model checkers could analyse the pDTMCs with low numbers of parameters, and performed their PMC faster than \approach\ for pDTMCs with the fewest parameters but increasingly slower than \approach\ for pDTMCs with more than approximately 40\% of their transition probabilities specified as parameters.

\begin{table}
\sffamily
\centering
\caption{
Parametric model checking times (in seconds, \mbox{`--'} indicates a timeout) for the COM model and its 22~properties}
\label{Table:COMTime}
\sffamily
\def\tabcolsep{10pt}
\begin{tabular}{rr @{\hspace*{6mm}} rrr}
\toprule
\multicolumn{2}{c}{\textbf{Property}} & \multicolumn{3}{c}{\textbf{PMC time}} \\ \cmidrule(l{1pt}r{7pt}){1-2} \cmidrule(l{-1pt}r{1pt}){3-5}
\textbf{ID} & \textbf{Type} &  
\textbf{fPMC} &  \textbf{Storm} & \textbf{PRISM} \\ \midrule
1 & \multirow{21}{*}{\rotatebox[origin=c]{90}{\parbox[c]{20mm}{\centering Reachability}}} 
& 12.8 & $<$1ms & 0.7 \\
2 && 12.4 & 0.001 & 0.6\\
3 && 11.8 & 0.004 & 0.8\\
4 && 12.0 & 0.004 & 0.4\\ 
5 && 12.9 & 0.03 & 1.2\\
6 && 12.9 & 0.01 & 0.6\\
7 && 12.5 & 0.4 & 3.0\\
8 && 12.7 & 0.1 & 0.9\\
9 && 12.6 & 0.4 & 4.8\\
10 && 12.8 & 0.5 & 1.9\\
11 && 11.6 & 3.8 & 3.9\\
12 && 11.9 & 8.4 & 9.4\\
13 && 11.8 & 1.7 & --\\
14 && 12.4 & 15.0& --\\
15 && 11.5 & 18.8 & --\\
16 && 12.5 & 94.3 & --\\
17 && 11.6 & 73.8 & --\\
18 && 11.6 & 316.6 & --\\
19 && 11.7 & 130.9 & --\\
20 && 11.3 & 39.2 & --\\
21 && 11.7 & 22.4 & --\\\midrule
22 & Reward &  41.7 & -- & --\\\bottomrule
\end{tabular}
\end{table}

\medskip
\noindent
\textbf{COM process.} As indicated in Table~\ref{Table:modelSummary}, a single pDTMC variant (with 234 states and 444 transitions) was analysed for the COM process. A number of 21~reachability properties and one reachability reward property (taken from~\cite{hajnal2019data}) were considered, and the times required to complete their parametric model checking are reported in Table~\ref{Table:COMTime}. 

Once more, \approach\ was the only approach that successfully analysed all properties. Storm completed the analysis of reachability properties 1--13 much faster than \approach, but took significantly longer to analyse the reachability properties 14--21, and timed out for the reachability reward property. PRISM completed the fewest analyses (12 out of 22) but produced the PMC formulae for these approximately $55\%$ of the properties faster than \approach. 

\medskip
\noindent
\textbf{Discussion.} \approach\ outperforms both Storm and PRISM in its ability to handle complex pDTMC with large numbers of parameters. In many of our PMC experiments with such models, fPMC completed its analysis within a few tens of seconds, while the other model checkers timed out after 3600 seconds. Thus, our approach often speed up the analysis for two or more orders of magnitude. Furthermore, the increase in the \approach\ analysis time as the models became more complex was consistently much slower than the increase of the analysis time for the other model checkers for the FX and PL systems, and the \approach\ PMC time was not affected much by the analysed property for the COM process. 

For some of the simpler pDTMC variants (in the case of the FX and PL systems) or properties (in the case of the COM process), Storm and, only occasionally, PRISM completed the analysis faster than \approach. These results are expected for the models and properties that Storm and PRISM can handle because the two leading model checkers use highly efficient internal representations (e.g., sparse matrices, binary decision diagrams) for DTMCs and sophisticated algorithms for their analysis. In contrast, \approach\ needs to perform fragmentation before leveraging the same functionality (by using Storm) for the resulting fragments and the abstract model induced by these fragments (cf.~Figure~\ref{fig:approach}). 

\begin{figure}
	\centering
	\includegraphics[width=0.95\linewidth]{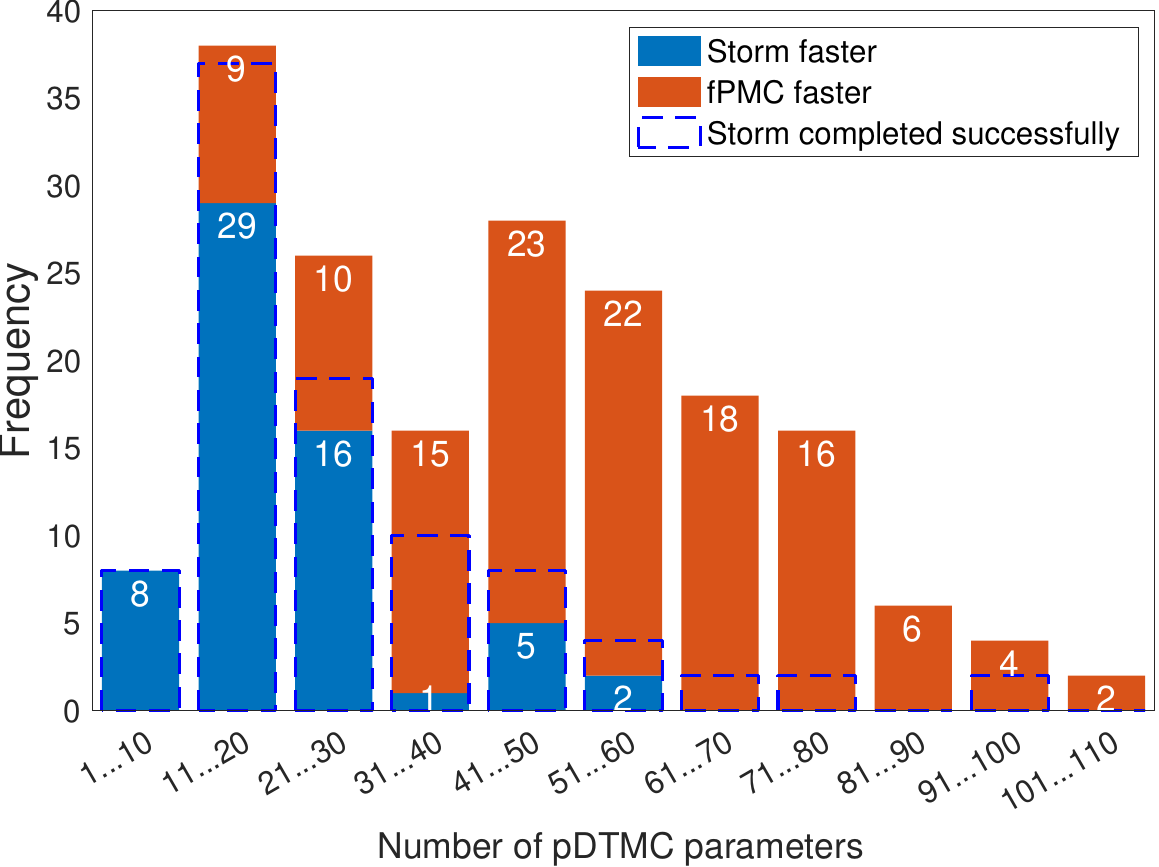}
		\caption{pDTMCs analyses completed faster by Storm and by \approach\ for models with different numbers of parameters. The dashed boxes show the total numbers of analyses completed by Storm within 60 minutes (\approach\ completed all analyses).}
    \label{fig:fxExampleheri}
\end{figure}

\begin{figure*}
     \centering
     \begin{subfigure}[b]{0.245\textwidth}
         \centering
         \includegraphics[width=\textwidth,height=4.2cm]{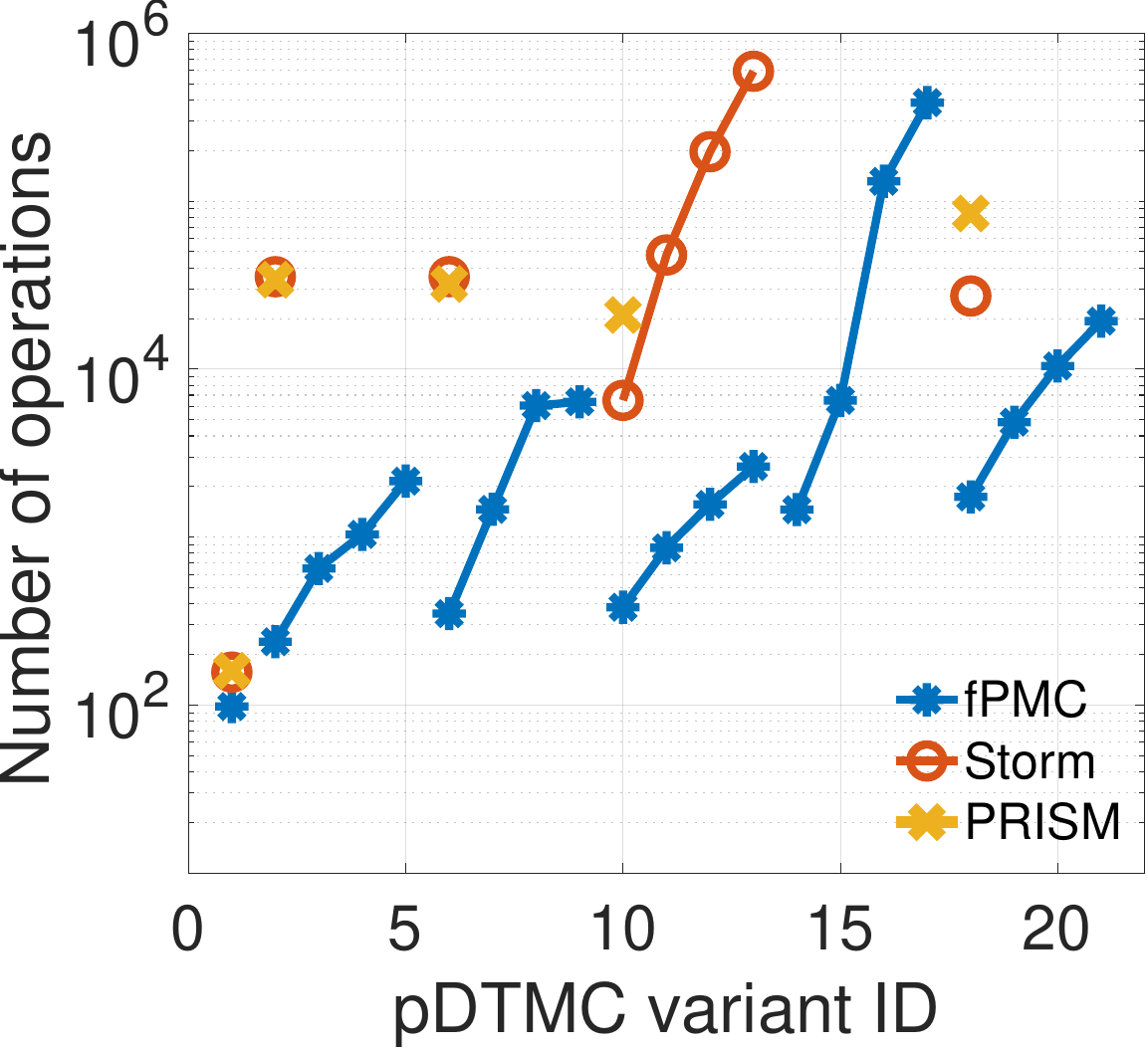}
         \caption{property \textsf{P1}}
         \label{fig:p1_fx}
     \end{subfigure}
     \hfill
     \begin{subfigure}[b]{0.245\textwidth}
         \centering
         \includegraphics[width=\textwidth,height=4.2cm]{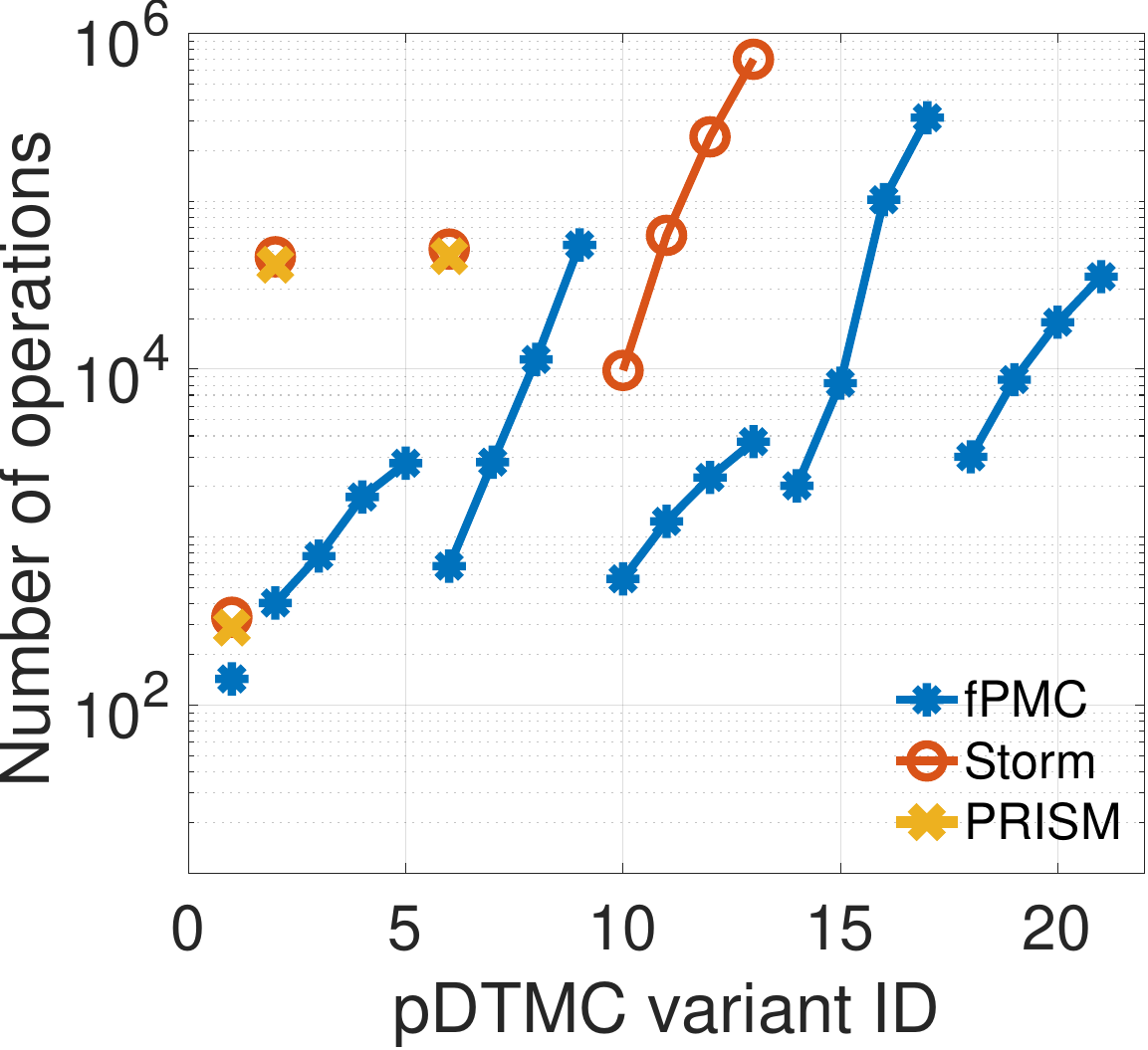}
         \caption{property \textsf{P2}}
         \label{fig:p2_fx}
     \end{subfigure}
     \hfill
     \begin{subfigure}[b]{0.245\textwidth}
         \centering
         \includegraphics[width=\textwidth,height=4.1cm]{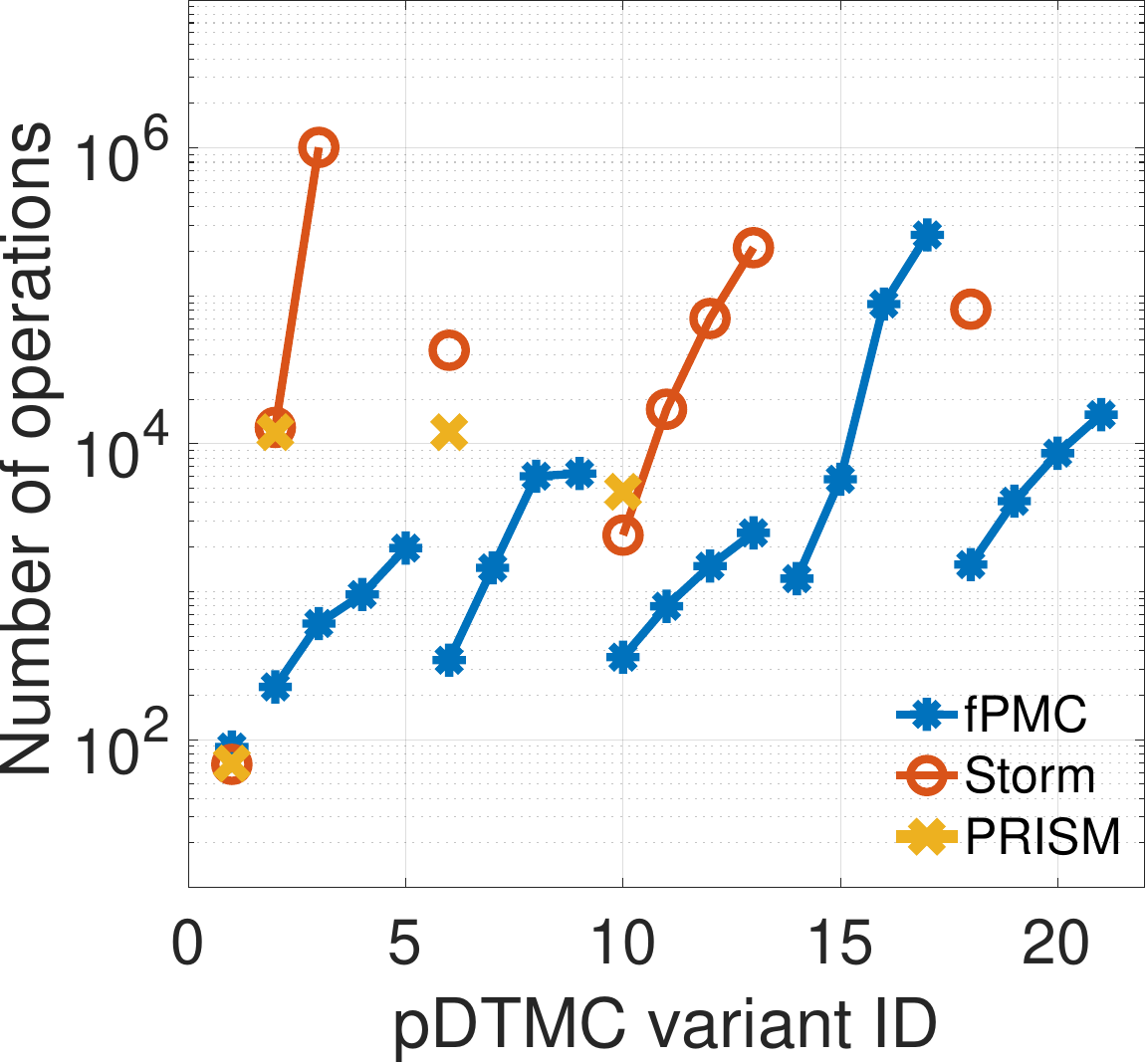}
         \caption{property \textsf{P3}}
         \label{fig:p3_fx}
    \end{subfigure}
    \hfill
     \begin{subfigure}[b]{0.245\textwidth}
         \centering
         \includegraphics[width=\textwidth,height=4.2cm]{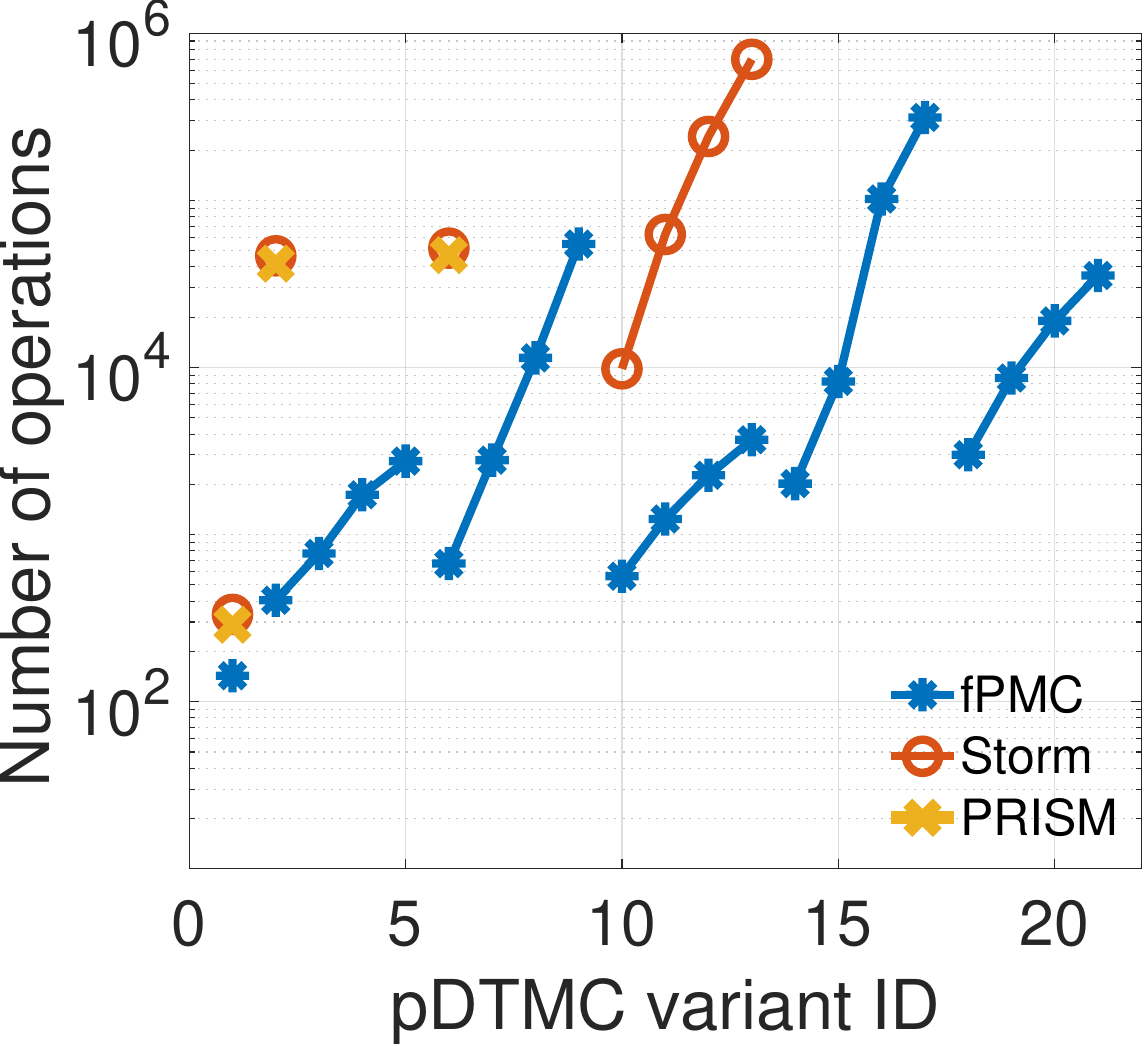}
         \caption{property \textsf{P4}}
         \label{fig:P4_fx}
     \end{subfigure}
        \caption{Number of operations in the PMC formulae for the FX pDTMC variants and properties from Table~\ref{Table:FXTime}, with the values corresponding to the same service combination strategy (SEQ, PAR, etc.) joined by continuous lines to improve readability}
        \label{fig:fx_op}
\end{figure*}

\begin{table*}
\sffamily
\centering
\def\tabcolsep{10pt}
\caption{MATLAB evaluation time for the FX PMC formulae (in seconds)}\label{MatlabFX}
\begin{tabular}{rccrrrrrrrr}
\toprule
\multicolumn{3}{c}{\textbf{pDTMC variant}} & \multicolumn{2}{c}{\textbf{P1}} & \multicolumn{2}{c}{\textbf{P2}} & \multicolumn{2}{c}{\textbf{P3}} & \multicolumn{2}{c}{\textbf{P4}} \\\cmidrule(lr){1-3}
\cmidrule(lr){4-5} \cmidrule(lr){6-7}\cmidrule(lr){8-9}\cmidrule(lr){10-11}
\textbf{ID}&\textbf{STG} & \textbf{\#SRV} & \textbf{fPMC} & \textbf{Storm} & \textbf{fPMC}      & \textbf{Storm}       & \textbf{fPMC}      & \textbf{Storm}       & \textbf{fPMC} & \textbf{Storm}\\
 \midrule
 1& -- & 1 & 0.001 & 0.004  & 0.001 & 0.002 & 0.001 & 0.002 & 0.001 & $<$ 1ms \\
 \midrule
 2 & SEQ & 2 & 0.002 & 3.5 & 0.005 & 8.2 & 0.002 & 7.9 & 0.001 & 0.3 \\ 
 \midrule
 6 & PAR & 2 & 0.006 & 3.5 & 0.006 & 8.9 & 0.002 & 2.4 & 0.004 & 0.3 \\
  \midrule
10 & \multirow{4}{*}{PROB} & 2 & 0.001 & 0.06 & 0.004 & 0.1 & 0.004 & 0.1 & 0.008 & 0.01\\
11& & 3  & 0.002 & 5.4 & 0.006 & 10.9  & 0.007 & 11.2 & 0.003 & 0.4\\
12& & 4  & 0.007 & 156.2  & 0.01 & 271.9  & 0.008 & 267.2  & 0.008 & 12.8 \\
13& & 5 & 0.013 & 4245.1  & 0.016 & 7507.1  & 0.016 & 7499.1  & 0.008 & 142.8 \\ 
\midrule
18 & PROB\_R & 2 & 0.006 & 17.2  & 0.01 & --$^\dagger$ & 0.01  & --$^\dagger$  & 0.007 & 0.09 \\
\bottomrule
\\[-2.5mm]
\multicolumn{11}{l}{$^\dagger$PMC formula unavailable for evaluation due to Storm analysis timeout}
\end{tabular}
\end{table*}

To exploit the capabilities of both \approach\ (which can efficiently analyse complex pDTMCs that other tools cannot handle) and Storm (which can efficiently analyse simple pDTMCs), our \approach\ tool employs the user-configurable threshold $\beta$ we mentioned in Section~\ref{sec:implementation}. For models with $\beta$ or more parameters, the tool performs the analysis by using \approach\ fragmentation, while for simple models with fewer than $\beta$ parameters Storm is used directly to perform monolithic PMC. To support the selection of a suitable value for the threshold $\beta$, we compared the execution times of \approach\ and Storm for pDTMC variants with different numbers of parameters from our evaluation. The result of this comparison is summarised by the histogram in Figure~\ref{fig:fxExampleheri}, which shows that most pDTMCs with up to 20 parameters were analysed faster by Storm, while almost all the pDTMCs with over 30 parameters were analysed faster by \approach. The two tools were each able to analyse faster a subset of the pDTMCs with numbers of parameters in the range $21\ldots30$, though Storm timed out before completing the analysis of several of these models. Further decomposing this set of pDTMCs into smaller ranges (e.g., between $21\ldots 25$ and $26\ldots 30$ parameters) does not yield a better separation into models handled faster by the two tools. As such, the comparison summarised in Figure~\ref{fig:fxExampleheri} suggests that setting the threshold $\beta$ to a value between $21$ and $30$ is likely to work well. While this rule of thumb may not always apply, we note that adopting it is never going to cause a problem: in borderline cases, one can easily analyse a pDTMC using both PMC tools.

\subsubsection{RQ2 (Result complexity)} 
\label{sssection:rq2}

For each of our three case studies and experiments presented in the previous section, we compared the number of arithmetic operations from the PMC formulae generated by \approach, Storm and PRISM, and the time required to evaluate the \approach\ and Storm formulae in MATLAB on the computer with the specification from Section~\ref{Sec:ExperimentalSetup}. We only considered Storm in the latter comparison because Storm completed significantly more PMC analyses than PRISM in our experiments (92 versus 55 out of  a total of 186 analyses). 

\medskip
\noindent
\textbf{FX system.} The sizes of the FX PMC formulae produced by \approach, Storm and PRISM are shown in Figure~\ref{fig:fx_op}. With one exception (for property \textsf{P3} of the pDTMC variant with ID 1), \approach\ generated formulae with significantly fewer operations than the other model checkers. This difference increases quickly for larger and more complex models, with an extreme case (for property \textsf{P1} of the pDTMC variant with ID 13) in which the formulae obtained by \approach\ contain over 225 times fewer operations than the Storm PMC formula (i.e., $2629$ versus $593426$ operations). 

The MATLAB evaluation times for the PMC formulae produced by Storm and \approach\ are reported (for the analyses completed by Storm) in Table~\ref{MatlabFX}. For the simplest models (e.g., pDTMC variants 1 and 10) all evaluations can be carried out within a few milliseconds. However, when the complexity of the model increases, the evaluation of the \approach\ formulae is significantly faster than that of the Storm formulae. This is particularly noticeable for the pDTMC variants corresponding to the PROB service-combination strategy, e.g., the evaluation of the \textsf{P2} property of pDTMC variant~13 took over 7500s when the Storm formula was used compared to only 16ms when the \approach\ formulae were used. Even for pDTMC variants for which Storm completed the analysis faster than \approach\ (e.g., those with IDs 2, 6, 10 and 11, cf.~Table~\ref{Table:FXTime}), the Storm PMC formulae are orders of magnitude larger than those computed by \approach, and therefore they require much longer time to evaluate. 

\begin{figure*}
     \centering
     \begin{subfigure}[b]{0.495\textwidth}
         \centering
         \includegraphics[width=0.9\textwidth,height =4cm]{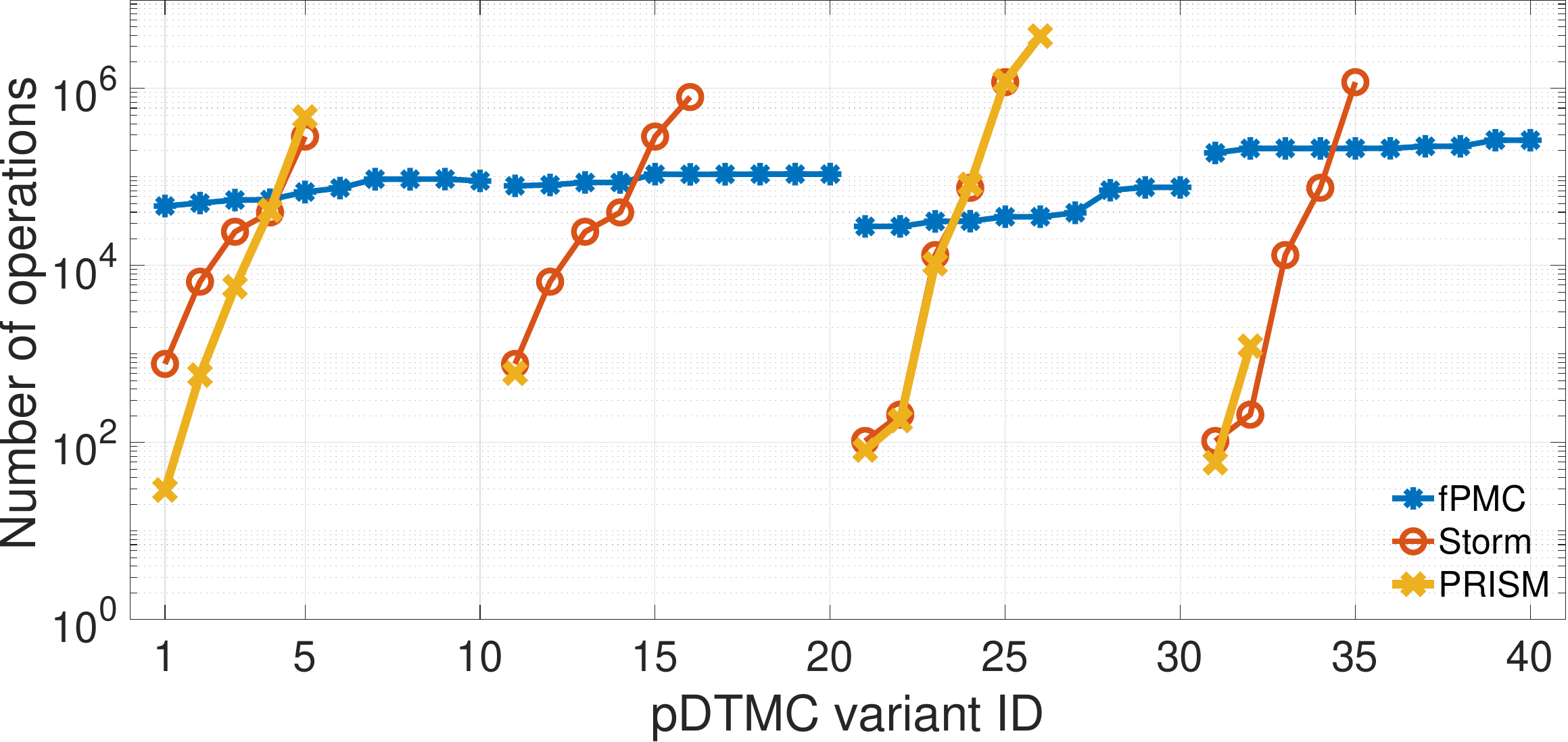}
         \caption{property reachability}
         \label{fig:p1_pl}
     \end{subfigure}
     \hfill
     \begin{subfigure}[b]{0.495\textwidth}
         \centering
         \includegraphics[width=0.9\textwidth,height =4cm]{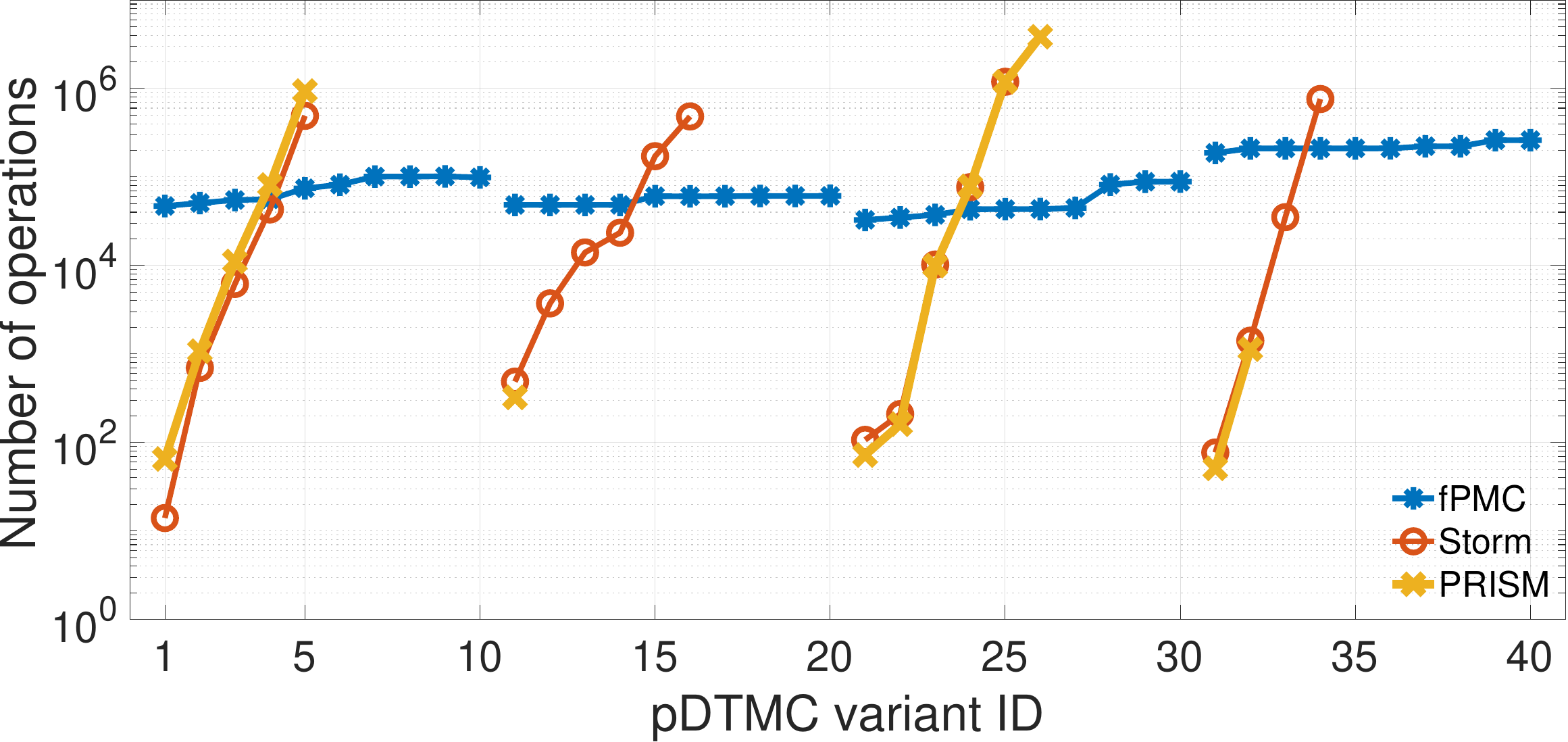}
         \caption{property unbounded until}
         \label{fig:p2_pl}
     \end{subfigure}
        \caption{Number of operations in the PMC formulae for the PL pDTMC variants and properties from Table~\ref{Table:PLTime}, with the values corresponding to the system with the same number of features joined by continuous lines to improve readability}
        \label{fig:PL_op}

\vspace*{-2mm}
\end{figure*}

\begin{table}
\sffamily
\centering
\caption{
MATLAB evaluation time for the PL PMC formulae (in seconds)}
\label{MatlabPL}
\sffamily
\def\tabcolsep{8pt}
\begin{tabular}{crcrrrr }
\toprule
\multicolumn{3}{c}{\textbf{pDTMC variant}} & \multicolumn{2}{c}{\textbf{Reachability}} & \multicolumn{2}{c}{\textbf{Unbounded until}} \\ \cmidrule(lr){1-3} \cmidrule(lr){4-5} \cmidrule(lr){6-7} 
\textbf{ID} & \textbf{\#F} & \textbf{\%PAR} & \textbf{fPMC} &  \textbf{Storm} &  
\textbf{fPMC} &  \textbf{Storm}\\ \midrule
1 & \multirow{5}{*}{4} & 10 & 5.5 & 0.001 & 5.8 & 0.001\\ 
2 & & 20 & 6.4 & 0.002 & 6.1 & $<$1ms\\ 
3 & & 30 & 7.1 & 0.06 & 6.7 & 0.07 \\  
4 & & 40 & 7.2 & 3.4 & 6.6 & 3.7\\ 
5 & & 50 & 11.4 & 1635.7 & 10.4 & 1214.1 \\ 
\midrule
11 & \multirow{6}{*}{16} & 10 & 1.6 & $<$1ms & 2.2 & $<$1ms \\ 
12 & & 20 & 21.3 & 0.05 & 6.1 & 0.02\\ 
13 & & 30 & 22.7 & 0.8 & 6.1 & 0.2 \\ 
14 & & 40 & 23.1 & 2.5 & 6.2 & 0.7 \\ 
15 & & 50 & 33.1 & 215.9 & 9.2 &68.2 \\ 
16 & & 60 & 33.9 & 4188.8 & 9.4 & 659.9 \\ 
 \midrule
21 & \multirow{5}{*}{18} & 10 & 1.5 & 0.001 & 2.2 & $<$1ms \\ 
22 & & 20 & 1.5 & 0.001 & 2.4 & 0.001 \\ 
23 & & 30 & 1.9 & 0.1 & 2.7 & 0.1 \\ 
24 & & 40 & 1.9 & 14.7 & 3.1 & 15.4 \\ 
25 & & 50 & 2.2 & 17113.7 & 3.2 & 16959 \\ 
 \midrule
31 & \multirow{4}{*}{22} & 10 & 116.1 & 0.001 & 106.2 & 0.001  \\ 
32 & & 20 & 160.4 & 0.005 & 142.2 & 0.005 \\ 
33 & & 30 & 165.3 & 2.3 & 145.9 & 2.1 \\ 
34 & & 40 & 122.1& 6049.8 & 153.1 & 5992.3\\ 
\bottomrule
\end{tabular}

\vspace*{-2mm}
\end{table}

\medskip
\noindent
\textbf{PL system.} Figure~\ref{fig:PL_op} shows the PMC formula sizes generated by the three model checkers for the PL system. As expected, with a gradual increase of parametric transitions (i.e., the percentage of pDTMC  transitions probabilities growing from 10\% to 100\%) in each model variant, larger algebraic formulae are obtain across all three model checkers, with  exponential growth for the formulae computed by Storm and PRISM, neither of which can handle the pDTMC variants with over 60\% of their transition probabilities specified as parameters. 

For simpler pDTMC variants (i.e., those with up to 40\% of their transitions specified as probabilities), the Storm and PRISM formulae have significantly fewer operators than the \approach\ formulae. We investigated this unexpected result, and found it to be due to a large number of parameters in the \approach\ abstract pDTMC model: one such parameter for each probability $\mathit{prob}_z$ of reaching an output state $z$ of a \approach\ fragment, as described in Section~\ref{subsec:fragmentationTheory}. As such, the \approach\ abstract models generated from the pDTMC variants with between 10--40\% parametric transition probabilities end up with more $\mathit{prob}_z$ parameters than the PL pDTMC variants they are obtained from. However, for these pDTMCs variants, many of the \approach\ fragments contain no or only a few PL parameters, and therefore multiple or even all $\mathit{prob}_z$ parameters for these fragments are in fact constant probabilities. We carried out separate experiments in which the constant values of such $\mathit{prob}_z$ parameters were used in the abstract \approach\ model instead of these parameters, and the size of the resulting \approach\ formulae became similar to that of the Storm and PRISM formulae for these pDTMC variants. Given these findings, we plan to include this simplification in the next version of our \approach\ tool.

A similar trend can be observed in the MATLAB evaluation time for the PL PMC formulae (Table~\ref{MatlabPL}). For the formulae derived for the pDTMC variants with up to 30\% parametric transitions, the Storm formulae can be evaluated within milliseconds, while the evaluation of the \approach\ formulae takes seconds (for the pDTMC variants modelling PL instances with four and 18 features), tens of seconds (for pDTMC variants modelling PL instances with 16 features) or even more than 100s (or pDTMC variants modelling PL instances with 22 features). Again, this difference is due to \approach\ operating with abstract models with unnecessarily many parameters, and can be resolve through the simplification explained earlier. 
For pDTMC variants with over 40\% parametric transitions, the \approach\ formulae are consistently and increasingly much faster to evaluate than those produced by Storm and PRISM, or these model checkers do not complete the PMC within 60 minutes.

\begin{figure}
 \centering
     \includegraphics[width=0.38\textwidth]{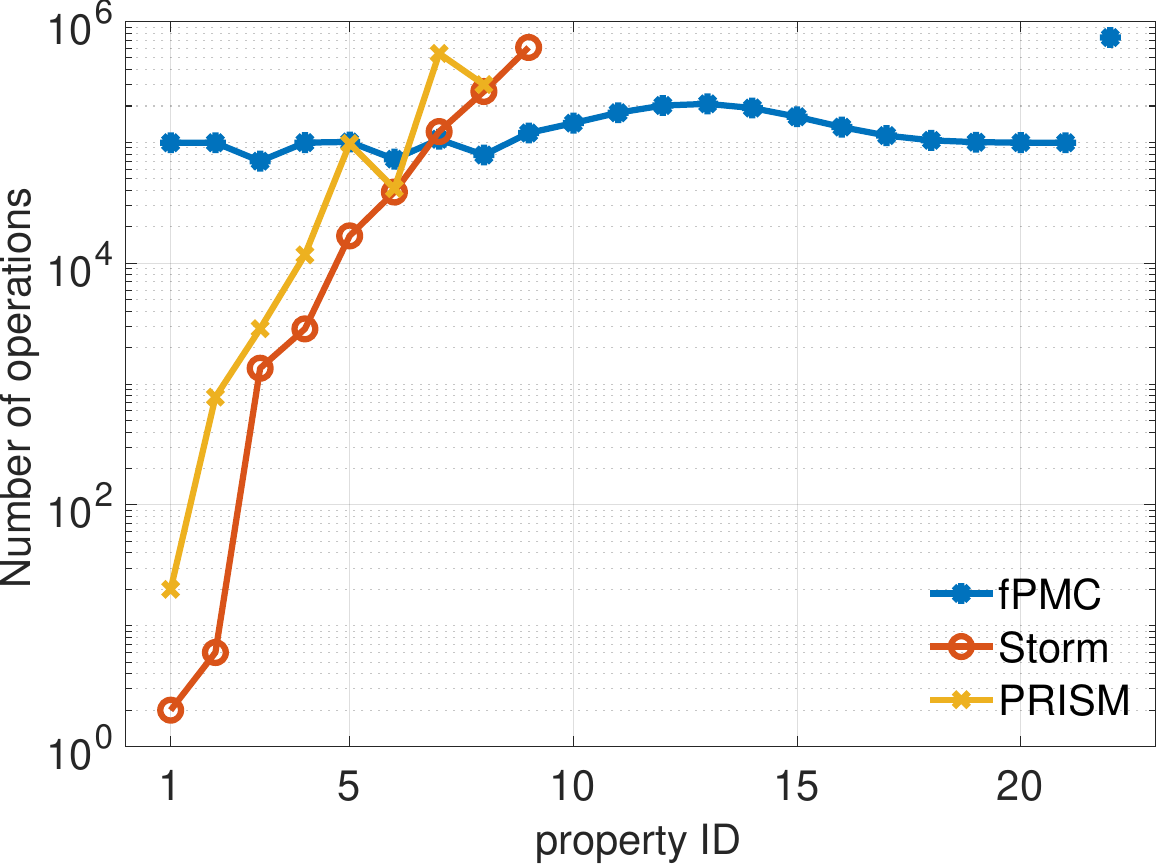}
    \caption{Number of operations in the PMC formulae for the COM pDTMC and properties from Table~\ref{Table:COMTime}; Storm and PRISM computed PMC formulae for several additional properties, but the multi-megabyte files required to store these extremely large formulae (which are available on our project website) were difficult to process, so their numbers of operations are not provided in the diagram}
    \label{fig:com_op}
\vspace*{-3mm}
\end{figure}

\begin{table*}
\sffamily
\centering
\def\tabcolsep{3.5pt}
\caption{MATLAB evaluation time for the COM process PMC formulae (in seconds)}\label{MatlabCOM}
\begin{tabular}{crrrrrrrrrrrrrrrrrrrrr}
\textbf{Property ID}&1&2&3&4&5&6&7&8&9&10&11&12&13&14&15&16&17&18&19&20&21\\
\midrule
\textbf{\approach} & 29 & 33 & 15 & 35 & 35 & 15 & 33 & 17 & 19  & 34 & 35 & 38 & 41 & 42 & 36 & 35 & 35 & 42 & 36 & 33 & 30 \\
\textbf{Storm} &0.02  & 0.001 & 0.01 & 0.008& 1 & 2  & 156 & 330  & 24528  & 14042& * & * & * & * & * & * & * & * & * & * & *  \\
\bottomrule
\\[-2.5mm]
\multicolumn{22}{l}{*PMC formula failed to be evaluated in MATLAB using the experimental machine in a 7-hour timeout}

\end{tabular}
\end{table*}

\medskip
\noindent
\textbf{COM process.} As shown in Figure~\ref{fig:com_op}, the PMC formula sizes generated by \approach, Storm and PRISM for the COM process follow a pattern similar to that obtained for the PL system. Thus, for the simplest properties (i.e., those with IDs between~1 and~5), the sizes for the Storm and PRISM formulae are much smaller than those of the \approach\ formulae; this is for the reason explained in our earlier discussion of the PL system results. However, as the analysed properties become more complex (i.e., because longer paths are required to reach the states from the reachability PCTL formulae), the number of operations from the Storm and PRISM formulae grows exponentially. In contrast, the size of the \approach-computed formulae remains relatively stable with the increased complexity of the properties. 

The MATLAB evaluation times for the Storm PMC formulae (Table~\ref{MatlabCOM}) are again increasing exponentially from a few millisecond for properties 1--4 to tens of thousands of seconds for properties~9 and~10, with MATLAB unable to complete the  evaluation before a seven-hour timeout for properties 11--21. MATLAB, however, managed to evaluate all formulae produced by \approach\ within at most 42s.   

\medskip
\noindent
\textbf{Discussion.}
The complexity of the PMC formulae (i.e., their number of arithmetic operations) in our three case studies increased with the complexity of the analysed pDTCM and PCTL property. For the Storm and PRISM formulae, this increase was exponential for all  case studies (cf.~Figures~\ref{fig:fx_op}, \ref{fig:PL_op}, and \ref{fig:com_op}). In contrast, the complexity of \approach\ formulae increased relatively little for the PL and COM case studies, and exponentially---but at a significantly slower rate than for the Storm and PRISM formulae---for the FX system. For low complexity PL and COM model-property combinations, Storm and PRISM computed PMC formulae with fewer operations than \approach. This was an unexpected behaviour that we investigated and explained in the analysis of the FX experiments, and for which we proposed a fix (to be implemented in the next version of our \approach\ tool).

The MATLAB evaluation time reflects the complexity of the PMC formulae generated by the three model checkers. The evaluation of the \approach\ formulae took milliseconds for the FX system, up to 165s for the most complex PL model-property combination, and up to 42s for the COM process. In contrast, the evaluation time for the PMC formulae generated by Storm (for the subset of model-property combinations that the model checker could analyse within 3600s) increased very rapidly with the complexity of these formulae---from milliseconds for the simplest formulae to several hours for the complex ones.  
The much faster evaluation made possible by \approach\ is highly beneficial, as it allows the parametric model checking of systems of far greater complexity than previously possible. Furthermore, self-adaptive systems that use parametric model checking  (e.g., \cite{Filieri2011,filieri2013probabilistic}) can leverage the simpler \approach\ formulae to perform their runtime decision-making much faster, and with significantly lower CPU overheads.

\begin{figure*}[t!] 
\begin{subfigure}{0.33\textwidth}
\includegraphics[width=\linewidth,height = 5.1cm]{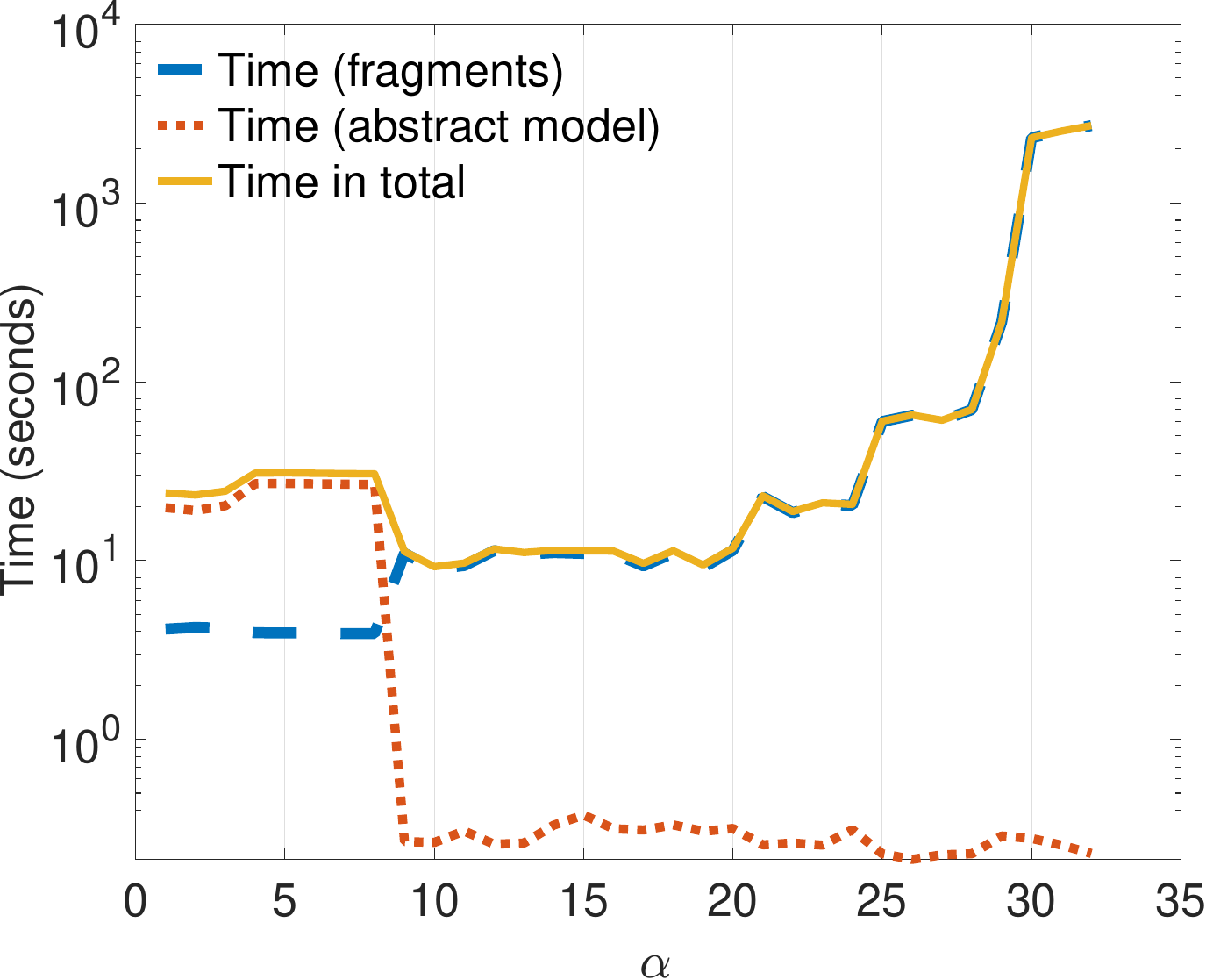}
\caption{Execution Time (ID~19, FX system)} \label{fig:a}
\end{subfigure}\hspace*{\fill}
\begin{subfigure}{0.33\textwidth}
\includegraphics[width=\linewidth]{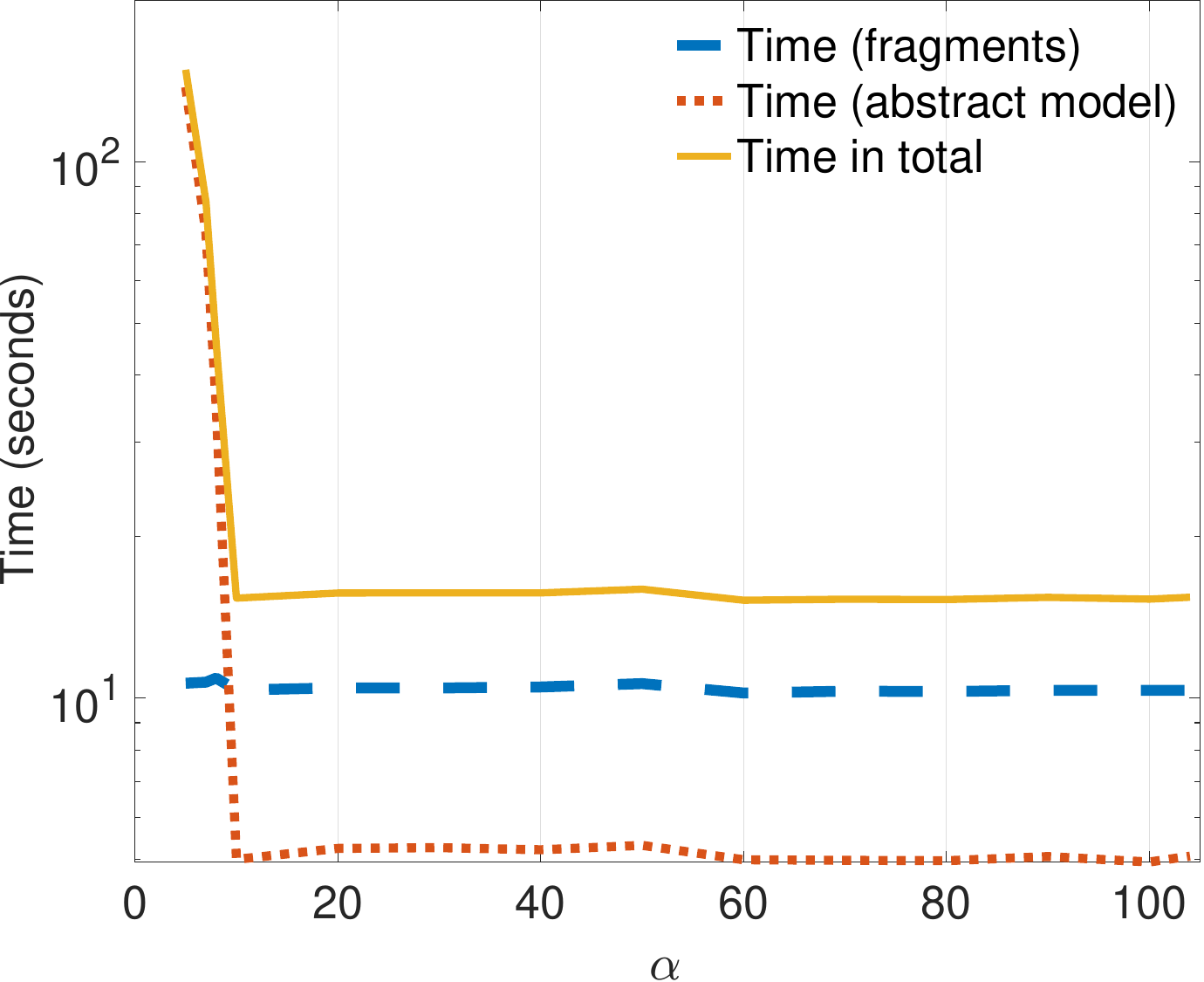}
\caption{Execution Time (ID~30, PL system)} \label{fig:b}
\end{subfigure}\hspace*{\fill}
\begin{subfigure}{0.33\textwidth}
\includegraphics[width=\linewidth]{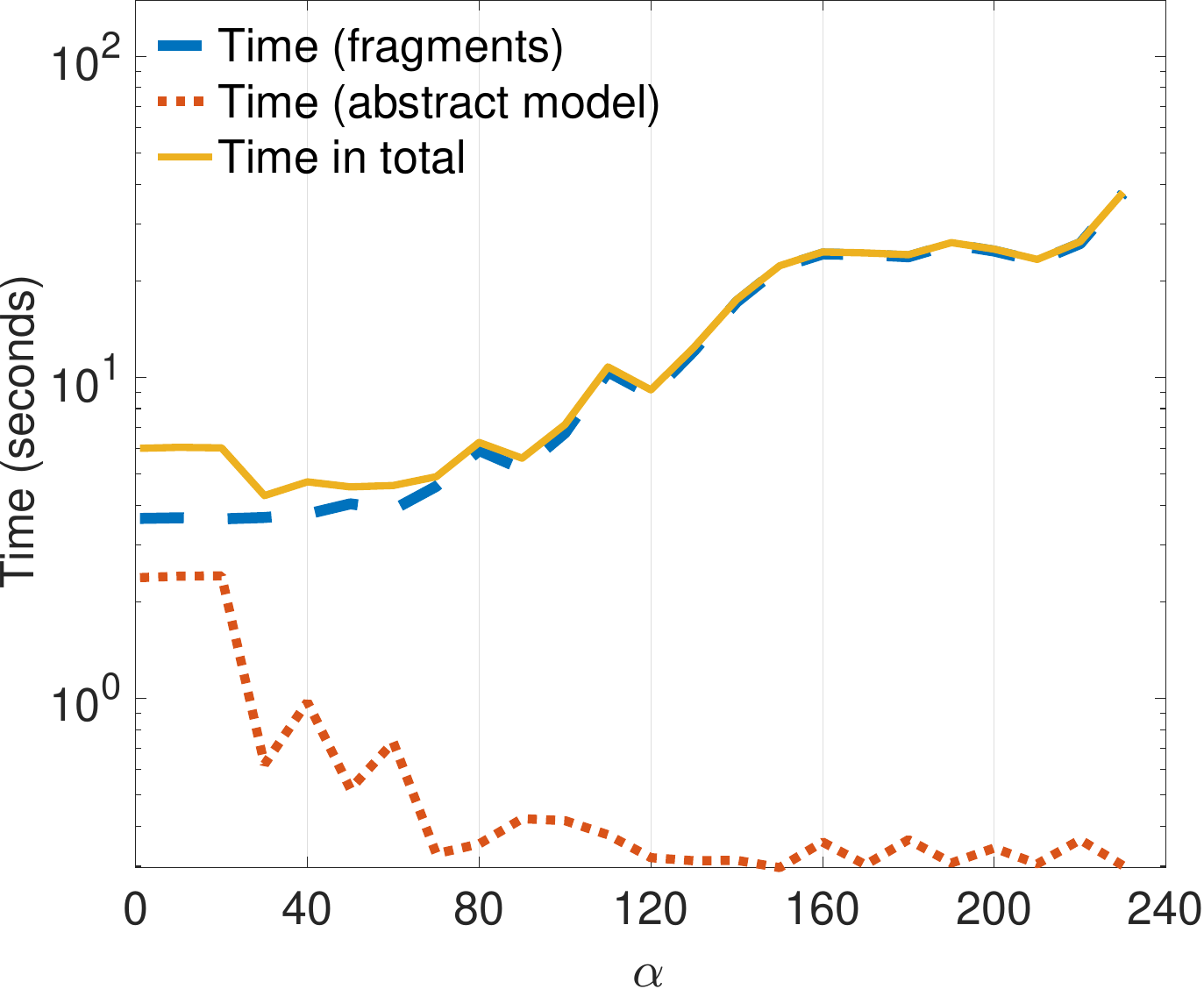}
\caption{Execution Time (ID~19, COM process)} \label{fig:c}
\end{subfigure}

\vspace{1cm}

\begin{subfigure}{0.33\textwidth}
\includegraphics[width=\linewidth,height = 5.1cm]{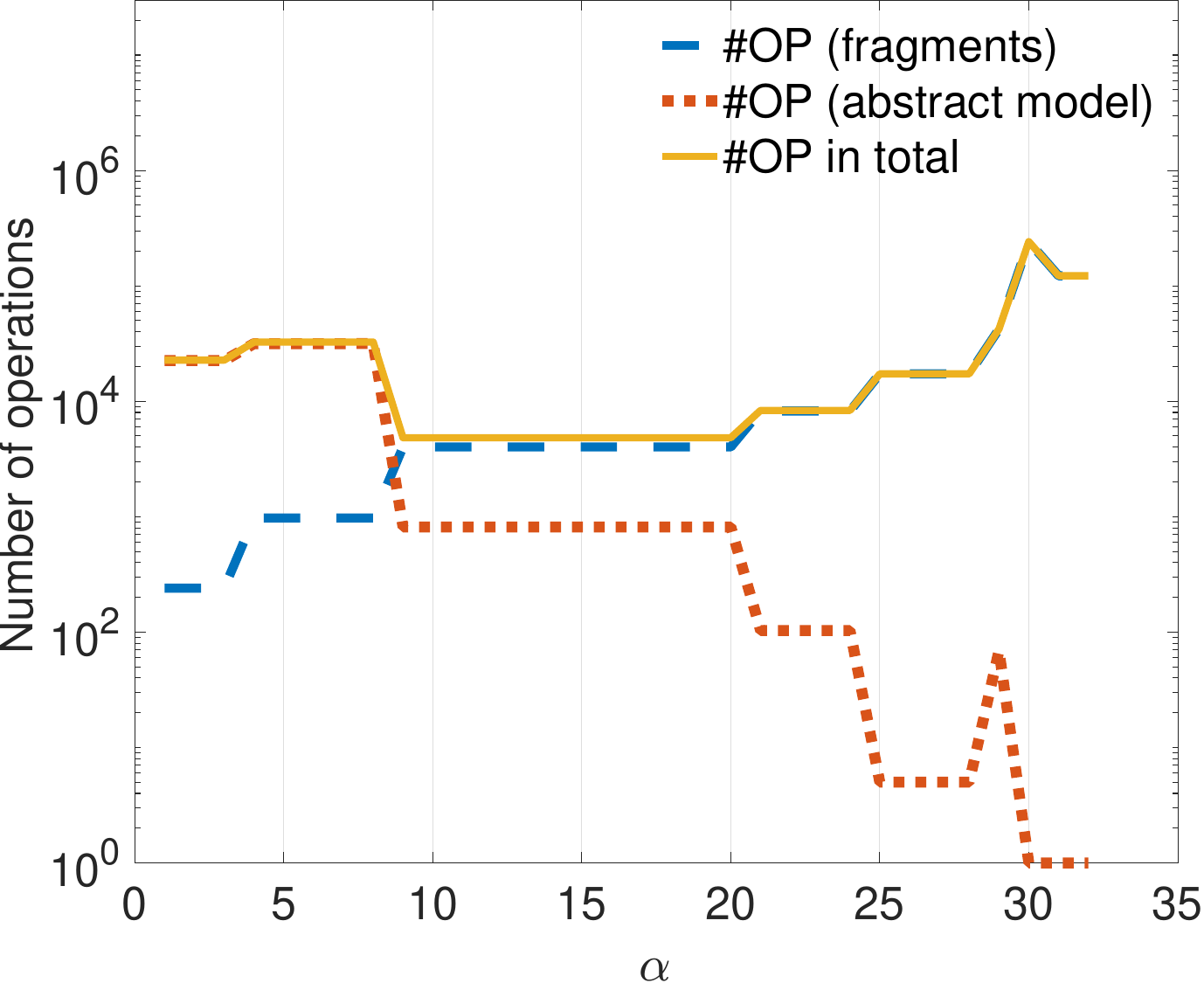}
\caption{Formulae size (ID~19, FX system)} \label{fig:d}
\end{subfigure}
\begin{subfigure}{0.33\textwidth}
\includegraphics[width=\linewidth,height = 5.1cm]{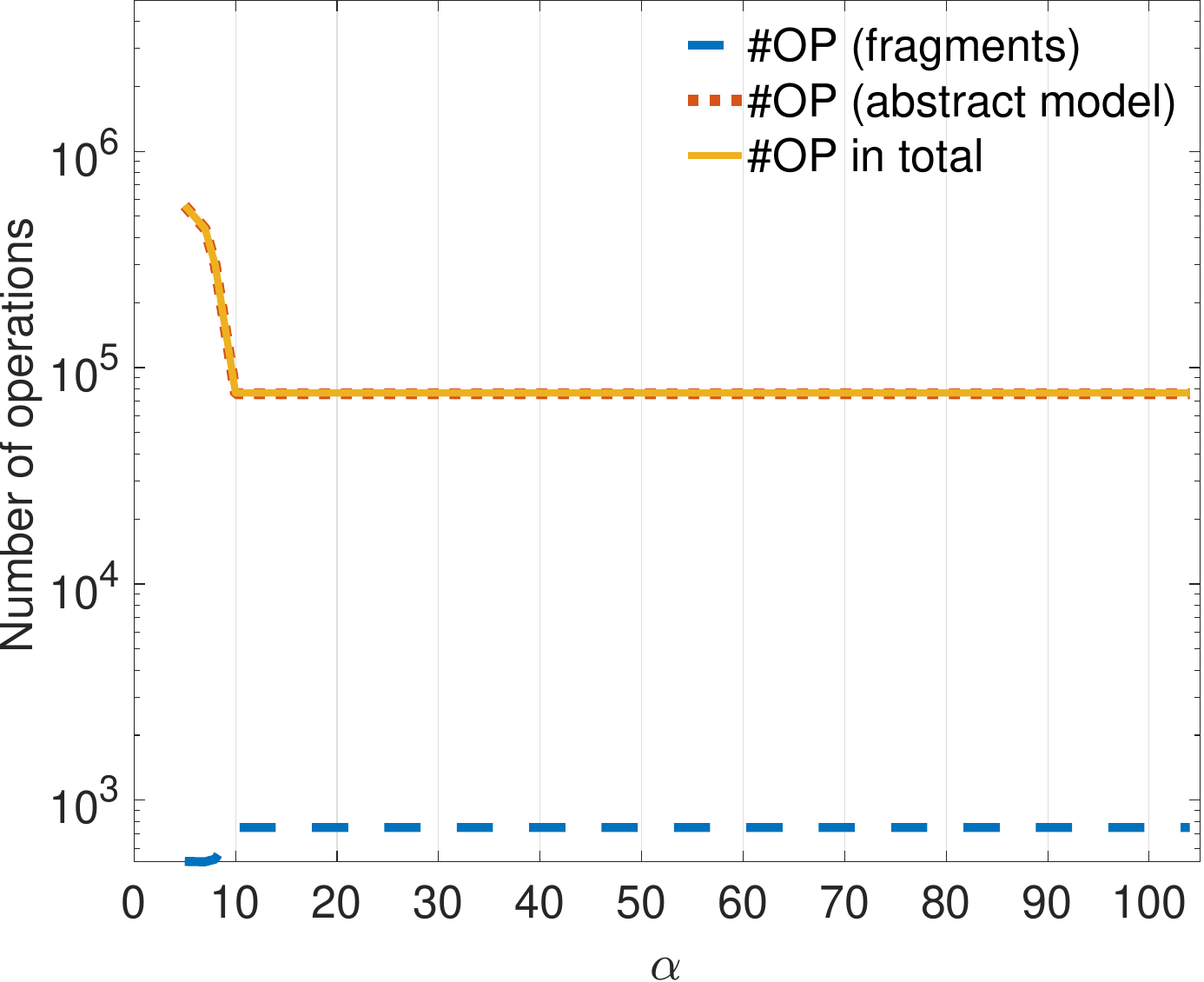}
\caption{Formulae size (ID~30, PL system)} \label{fig:e}
\end{subfigure}\hspace*{\fill}
\begin{subfigure}{0.33\textwidth}
\includegraphics[width=\linewidth,height = 5.1cm]{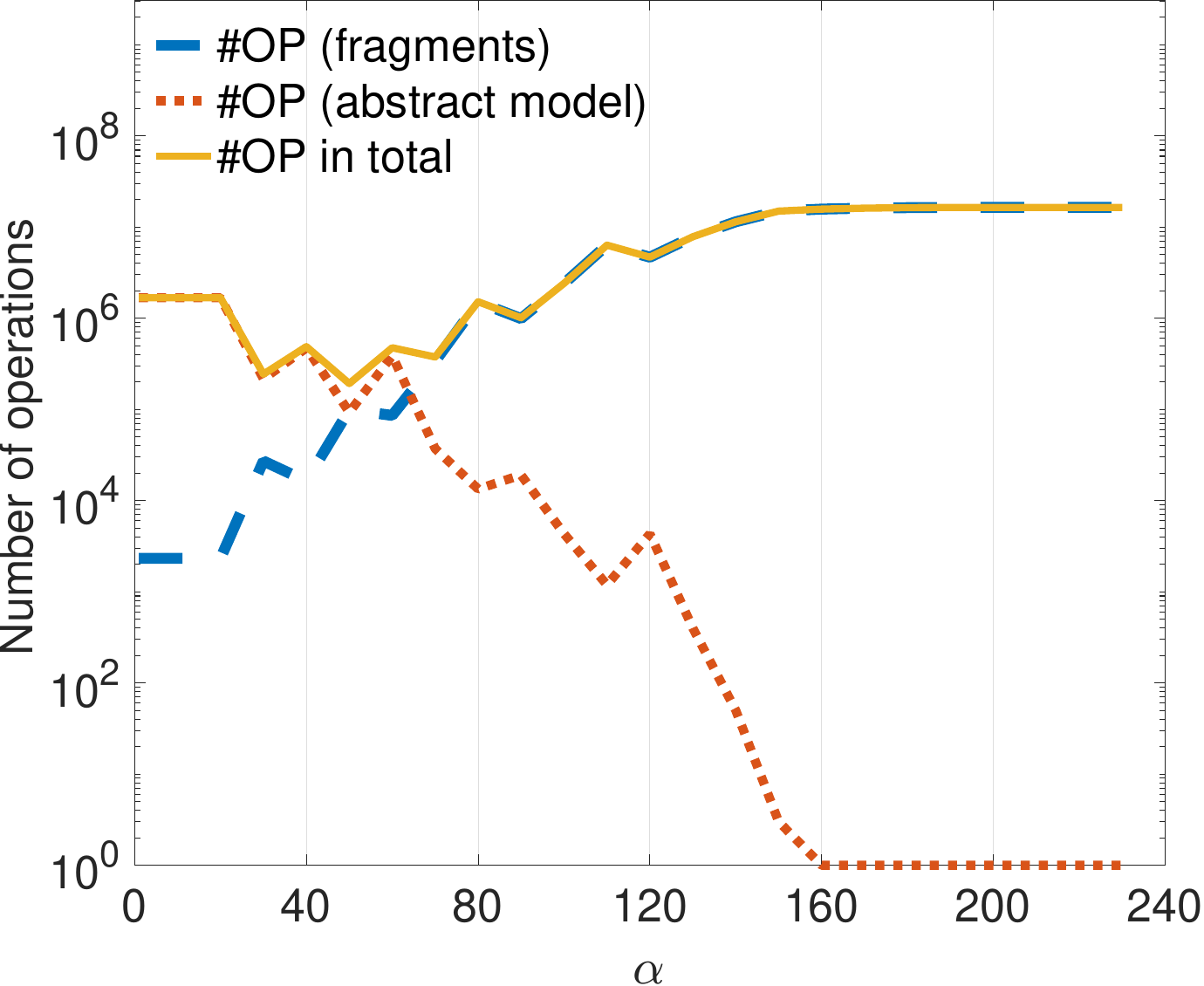}
\caption{Formulae size (ID~19, COM process)} \label{fig:f}
\end{subfigure}

\caption{\approach\ execution time for (a)~the reachability property of FX pDTMC variant 19, (b)~the reachability property of PL variant~30, and (c)~reachability property~19 of the COM pDTMC; and \approach\ formula complexity (i.e., number of operations) for the same pDTMC-property combinations of (d)~the FX system, (e)~the PL system, and (f)~the COM process.} 
\label{fig:alpha}
\end{figure*}

\subsubsection{RQ3 (Configurability)} 
We evaluated the impact of varying the \approach\ fragmentation threshold $\alpha$ on: (1)~the \approach\ execution time, and (2)~the number of operations in the resulting algebraic formulae. To that end, we randomly selected one pDTMC model--reachability property combination from each of our three case studies, and we performed its \approach\ analysis for all possible values of $\alpha$ (i.e., from 1 to the maximum number of states in the analysed pDTMC). The results of these experiments are summarised in Fig.~\ref{fig:alpha}.

\medskip
\noindent
\textbf{FX system.} Figures~\ref{fig:a} and ~\ref{fig:d} show the results from the analysis of the reachability property of FX pDTMC variant 19, which has 35~states. \approach\ completed the analysis successfully within 60~minutes for $1 \leq \alpha \leq 32$. For small $\alpha$ values ($\alpha < 10$), the time spent checking the abstract model was higher than that required to analyse the fragments. This is because smaller $\alpha$ values often lead to a larger abstract model by limiting the size of each fragment, and therefore increasing the number of fragments. For $\alpha\geq 10$, each fragment was allowed to grow larger, resulting in abstract models with fewer states. As a result, the time spent checking the abstract model  decreases (and the time spent checking the fragments grows) with the increase of $\alpha$. For $\alpha\geq 32$, \approach\ could not complete the analysis of its first fragment within the 60-minute timeout. This indicates that a single fragment can become too large and complex (as a result of $\alpha$ being too large) to be handle by the model checker that \approach\ uses for the analysis of individual fragments (i.e., Storm). The same pattern can be observed in the size of the obtained closed-form formulae: their numbers of operations increase in the fragments but decrease in the abstract model with the increase of $\alpha$. A mid-range $\alpha$ value of between 10 and 20 minimises both the \approach\ execution time and the complexity of the generated formulae.

\medskip
\noindent
\textbf{PL system.} Figures~\ref{fig:b} and~\ref{fig:e} show the results for the analysis of the reachability property of PL pDTMC variant 30. This pDTMC has 104 states, and \approach\ completed its analysis within 60 minutes for all $\alpha$ values between~4 and 104. For $1\leq \alpha\leq 3$, \approach\ timed out in the analysis of the abstract model, which was too complex. For $\alpha>3$ but still relatively small (i.e., $4 \leq \alpha \leq 8$), the time spent analysing the abstract model is higher than that required to analyse the fragments---and the sizes of the \approach\ algebraic formulae follow a similar pattern. As $\alpha$ increases further, a similar trend to that from Figure~\ref{fig:a} can be observed: the abstract model analysis time starts to decrease, and the fragment analysis goes up. However, increasing $\alpha$ beyond~10 has no noticeable impact on the analysis time and formula size. This is due to the fact that \approach\ partitions this pDTMC into fragments of up to 10~states ``naturally'', meaning that the forced fragment termination from line~\ref{Alg1Terminate} of Algorithm~\ref{algorithm:fragmentation} is not exercised for $\alpha>10$. As a consequence, the abstract model remains relatively complex even for large $\alpha$ values, and---irrespective of the value of $\alpha$---the vast majority of the operations from the \approach\ formulae for the PL reachability property come from the abstract model formula. 

\medskip
\noindent
\textbf{COM process.} Figures~\ref{fig:c} and~\ref{fig:f} show the results for the analysis of reachability property~19 of the COM pDTMC model. This pDTMC has 234~states, and \approach\ completed the analysis successfully for all possible values of $\alpha$, i.e., $1 \leq \alpha \leq 234$. Similar to the FX and PL experiments, increasing $\alpha$ led to lower analysis times and fewer formula operations for the abstract model, but resulted in higher analysis times and more formula operations for the fragments. As for the FX system, an intermediate $\alpha$ value (of between 30--70 in this case) yields the lowest total analysis time and number of operations.

\medskip
\noindent
\textbf{Discussion.} The fragmentation threshold $\alpha$ serves as a soft upper bound for deciding whether \approach\ should continue model traversal (adding further states to its stack, cf.~line~\ref{l:check-alphaS} of Algorithm~\ref{algorithm:fragmentation}) or invoke fragment termination by forcing the currently analysed state to become an output fragment state (cf.~line~\ref{l:Formation1} of Algorithm~\ref{algorithm:fragmentation}). The experiments summarised in Figure~\ref{fig:alpha} suggest that the selection of optimal $\alpha$ values depends on many factors, including pDTMC structure and number of states. Trying all the possible values of $\alpha$ as what we did here is impractical. Thus, the determination of optimal $\alpha$ values (other than by exhaustive search) remains an open research question. Nevertheless, this result suggests useful rules of thumb for the selection of suitable $\alpha$ values. First, adopting a default $\alpha$ value of between 10--30 is likely to produce good results (by guiding \approach\ to create abstract models with at least one order of magnitude fewer state than the original pDTMC). Second, increasing $\alpha$ may help if the time to analyse the abstract model is too high (or the formula produced by this analysis is too complex); conversely, decreasing $\alpha$ may help if the time to analyse one of the fragments is too high (or the formulae from the fragment analyses are too complex). Used in conjunction with hill climbing~\cite{selman2006hill}, the latter thumb rule could enable the optimisation of the threshold $\alpha$, e.g., to ensure that \approach\ yields formulae with the lowest number of operations possible.

\section{Threats to Validity}
\label{sec:validity}

\noindent
\textbf{Construct validity threats.} 
These threats may be caused by over-simplifications and invalid assumptions made during in devising the evaluation experiments. To avoid them, we carried out the evaluation of \approach\ by using case studies based on two software systems and a communications process that are freely available from other published software engineering projects. Furthermore, the pDTMCs and properties evaluated in our paper were also used in related research~\cite{RaduePMC,ghezzi2013model,hajnal2019data,classen2010model,Gerasimou2015:ASE}.  

\medskip
\noindent
\textbf{Internal validity threats.} 
To avoid these threats---which could introduce bias in the identification of cause-effect relationships in our experiments---we evaluated \approach\ by answering three independent research questions (cf.~Section~\ref{Sec:ExperimentalSetup}). To further mitigate the potential bias, the evaluation results were compared against the leading probabilistic model checkers Storm~\cite{storm} and PRISM~\cite{prism}. In addition, we performed the evaluation using multiple variants of the systems from our case studies, and the correctness of each result produced by \approach\ was individually checked using the approach described in Section~\ref{Sec:ExperimentalSetup}. Finally, we published the source code and data from our experiments online, in order to allow the replication and verification of our results.

\medskip
\noindent
\textbf{External validity threats.} These threats could affect the generalisability of our findings. As summarised in Table~\ref{Table:modelSummary}, we mitigated them in our evaluation by applying \approach\ to three types of systems (i.e., service-based systems~\cite{ameller2016survey,sun2010modeling}, software product lines~\cite{clements2002software,apel2013software}, and communication processes~\cite{dressler2010survey,dressler2010bio}) taken from different application domains. Moreover, we used multiple model variants and properties, allowing us to test our approach on a wide range of pDTMC structures and sizes, and thus to show that \approach\ provides consistently better performance in terms of faster computation time, fewer arithmetic operations in the derived algebraic formulae, and faster evaluation of these formulae than the model checkers PRISM and Storm for complex models. Finally, we eased the use of \approach\ in practice by providing tool support for our approach. However, additional experiments are needed to confirm that \approach\ is applicable to a wider range of pDTMC models and to other application domains.

\section{Related Work}
\label{sec:relatedwork}

Parametric model checking (PMC) was firstly introduced by Daws~\cite{Daws:2004:SPM:2102873.2102899} less than two decades ago. The technique enables the analysis of a DTMC when some or all of its transition probabilities are specified as rational functions over the parameters of the modelled system. PMC produces an algebraic formula for each analysed property. Such formulae are used in the design and verification of software systems, e.g., to compare alternative system designs~\cite{ghezzi2011verifying,ghezzi2013model}, to dynamically evolve the configuration of self-adaptive software~\cite{calinescu2009cads,de2017software,Filieri2011,filieri2013probabilistic}, and for system parameter analysis and synthesis~\cite{filieri2015supporting,calinescu2017synthesis,calinescu2018efficient,hahn2011synthesis}. Despite this wide adoption, the computationally intensive nature of PMC limits the scalability and applicability of the multiple software engineering methods that rely on it. 

To the best of our knowledge, the research aimed at improving the efficiency of PMC is limited to the approaches proposed  in~\cite{hahn2011probabilistic,gainer2018accelerated,Jansen2014,RaduePMC,baier2020parametric}. For the purpose of comparing these approaches to \approach, we consider them organised into two classes: those which (like Daws' original PMC technique) operate on the entire pDMTC under analysis~\cite{hahn2011probabilistic,gainer2018accelerated,baier2020parametric}, and those which (like our \approach\ approach) operate by partitioning this pDTMC into components that are then analysed individually~\cite{Jansen2014,RaduePMC}. We term these classes of approaches standard PMC and compositional PMC, respectively.

\medskip
\noindent
\textbf{Standard PMC approaches.} 
The first class of PMC approaches are complementary to our work. Any of them can be used in conjunction with \approach, to ensure that the individual probabilistic model checking of the \approach\ fragments and \approach\ abstract model is carried out efficiently. We summarise these approaches below.

In comparison to Daws' initial PMC technique~\cite{Daws:2004:SPM:2102873.2102899}, the technique proposed by Hahn et al.~\cite{hahn2011probabilistic} yields a significant improvement in PMC performance. Both PMC approaches derive an algebraic formula for the probability of reaching a set of parametric Markov chain states specified in a PCTL path formula. However, instead of computing a regular expression by exploiting the pDTMC structure as in~\cite{Daws:2004:SPM:2102873.2102899}, the PMC technique from~\cite{hahn2011probabilistic} produces a rational expression and leverages symmetry and ``cancellation'' properties of rational formulae to simplify this expression. The cancellation involves computing the greatest common divisor (GCD) of the denominator and numerator polynomials of the rational PMC formula, and using this GCD polynomial to simplify the formula. The model checkers PARAM~\cite{param} and PRISM~\cite{prism} implement this technique.

According to our classification, Jansen et al.~\cite{Jansen2014} propose a hybrid PMC approach, and we discuss its PMC ``components'' separately. The standard PMC component of~\cite{Jansen2014} consists of sophisticated partial polynomial factorisations that support the efficient simplification of large PMC rational expressions. This PMC approach is implemented by the parametric model checker Storm~\cite{storm}. 

More recently, Baier et al.~\cite{baier2020parametric} have introduced a new technique for obtaining simplified PMC formulae. This technique avoids the computationally expensive calculation of GCD polynomials (which the PMC approaches from~\cite{hahn2011probabilistic,Jansen2014} rely on) by leveraging fraction-free Gaussian elimination, which is an existing efficient method for solving systems of parametric linear equations. This method is implemented in the Storm model checker.

A different type of approach to speeding up PMC is proposed by Gainer et al.~\cite{gainer2018accelerated}. This approach involves the stepwise elimination of the states of the analysed pDTMC, through a process that resembles the mapping of finite automata to regular expressions. The outcome of this elimination is a directed acyclic graph encoding of the PMC result instead of the usual rational formula produced by other PMC techniques. The authors' evaluation of the approach (implemented in their ePMC/\textsc{IscasMC} model checker~\cite{10.1007/978-3-319-06410-9_22}) shows that it can outperform the formula-based PMC engine used by PRISM by up to two orders of magnitude.

\medskip
\noindent
\textbf{Compositional PMC approaches.} These approaches operate in a similar way to \approach. As such, we will compare each of them to our work.

The PMC approach devised by Jansen et al.~\cite{Jansen2014} improves the efficiency of parametric model checking by decomposing the state transition graph induced by the pDTMC under analysis into  strongly connected components (SCCs). PMC expressions are then computed independently for each SCC, and then combined to obtain the final PMC output in the form of a single rational formula over the parameters of the system modelled by the pDTMC. Because it is predetermined by the SCCs of the analysed pDTMC, this decomposition (which is implemented by the parametric model checker Storm) is very rigid. In particular, it may produce SCCs that are too large to be analysed efficiently, and that cannot be further decomposed.
In contrast, the \approach\ fragmentation of a pDTMC is much more flexible. The analysed pDTMC is partitioned into fragments whose size is guided by the fragmentation threshold $\alpha$. These fragments can include one or more small SCCs. Most importantly, \approach\ can split any SCCs that are too large to be analysed individually into multiple fragments. The experimental results from Section~\ref{sec:evaluation} provide ample evidence about the benefits of this flexible pDTMC partitioning. Furthermore, the \approach\ fragmentation can, in theory, be applied repeatedly, e.g., to partition a large fragment into sub-fragments (although this still has to be confirmed experimentally).  

The \approach\ theoretical foundation comprises two complementary parts. The first part, which we introduced in~\cite{RaduePMC}, defines the method for using pDTMC fragments to speed up parametric model checking---but does not provide any method for partitioning a pDTMC into fragments. Therefore, the solution from~\cite{RaduePMC} can only be applied when its users are able to exploit domain knowledge in order to manually specify the fragments of the analysed pDTMC. This limitation represents a significant barrier for the practical adoption of fragmentation-based PMC. The second part of the \approach\ theoretical foundation, which is introduced in this paper and complements our results from~\cite{RaduePMC}, removes this barrier by providing a tool-supported method for the automated fragmentation of pDTMCs.

\section{Conclusion}
\label{sec:conclusion}

We presented \approach, a tool-supported technique for software performability analysis through compositional parametric model checking. \approach\ supports the efficient analysis of reachability, unbounded until and reachability reward properties of parametric discrete-time Markov chains by automatically partitioning these models into fragments that can be analysed independently. The results of these analyses are then combined into a system of closed-form algebraic expressions that represent the solution of the initial parametric model checking problem. 

To evaluate \approach, we used it to analyse 28 PCTL properties of 62 pDTMC variants modelling three types of software systems (i.e., service-based systems, software product lines, and middleware) from different application domains. The experimental results show that \approach\ can analyse pDTMCs with over 10--20 parameters much faster than previous PMC techniques, and---in many cases---when these techniques cannot complete their analyses within 60~minutes on a standard computer. Furthermore, our evaluation showed that the algebraic expressions generated by \approach\ for such models comprise considerably fewer operations and are much faster to evaluate than those produced by previous PMC techniques. 

In future work, we will explore several opportunities for extending the applicability and efficiency of \approach. First, we will examine the possibility to apply \approach\ fragmentation repeatedly, to partition pDTMC fragments that may be too large or too complex to analyse as a whole into sub-fragments. This opportunity, which is unique to our compositional PMC technique, has the potential to support the analysis of complex pDTMCs that cannot be handled by any existing model checkers. 

Second, we aim to enhance the \approach\ fragmentation algorithm with the ability to generate close-to-optimal fragments, i.e., fragments that: (i)~are non-trivial in terms of structure, size and number of parameters; (ii)~can be analysed efficiently; and (iii)~produce PMC expressions of acceptable complexity. One option for obtaining such fragments is to adapt the \approach\ fragmentation threshold to the characteristics of the fragment under construction. 

Last but not least, we plan to improve our \approach\ tool. In particular, we will implement the simplification identified in Section~\ref{sssection:rq2}. To that end, we will update the \approach\ tool to ensure that abstract model parameters associated with constant-valued fragment reachability properties are replaced with the actual values of those properties. Additionally, we will explore options for selecting an effective PMC technique (e.g., standard or compositional) for the analysis of a given pDTCM, pDTMC fragment or pDTMC strongly connected component, paving the way for the development of a highly efficient hybrid parametric model checker that uses the available PMC techniques together.

\section*{Acknowledgements}

This project has received funding from the UKRI project EP/V026747/1 `Trustworthy Autonomous Systems Node in Resilience', and the ORCA-Hub PRF project `COVE'. The authors are grateful to the research groups who have developed the EPMC/IscasMC, PARAM, PRISM and Storm probabilistic and parameteric model checkers: our work would not have been possible without the significant theoretical advances and tools introduced by these research groups.

\section*{CRediT authorship contribution statement}

\textbf{Xinwei~Fang:} Conceptualization, Investigation, Methodology, Software, Validation, Visualization, Writing – original draft.
\textbf{Radu~Calinescu:} Conceptualization, Formal Analysis, Funding Acquisition, Investigation, Methodology, Supervision, Validation, Visualization, Writing -- original draft, Writing -- review \& editing. 
\textbf{Simos~Gerasimou:} Conceptualization, Funding Acquisition, Investigation, Methodology, Software, Supervision, Visualization, Writing -- original draft.
\textbf{Faisal~Alhwikem:} Methodology, Software.

\bibliographystyle{IEEEtranS}
\bibliography{ref}

\end{document}